\documentclass[11pt,a4paper]{article}
\usepackage{amssymb,amsmath}
\usepackage{color}
\usepackage{graphicx,graphics}
\usepackage{mathtools}
\usepackage[english]{babel}
\usepackage[utf8]{inputenc}
\usepackage{epsfig,url}
\usepackage{bbm,theorem}
\usepackage{a4wide}
\usepackage{enumerate}
\usepackage{relsize}
\usepackage{hyperref}
\usepackage{mathrsfs}

\usepackage{soul}

\usepackage[square,sort,comma,numbers]{natbib}

\DeclareFontFamily{U}{mathx}{\hyphenchar\font45}
\DeclareFontShape{U}{mathx}{m}{n}{
      <5> <6> <7> <8> <9> <10>
      <10.95> <12> <14.4> <17.28> <20.74> <24.88>
      mathx10
      }{}
\DeclareSymbolFont{mathx}{U}{mathx}{m}{n}
\DeclareFontSubstitution{U}{mathx}{m}{n}
\DeclareMathAccent{\widecheck}{0}{mathx}{"71}
\DeclareMathAccent{\wideparen}{0}{mathx}{"75}

\newtheorem{theorem}{Theorem}[section]
\newtheorem{definition}[theorem]{Definition}
\newtheorem{corollary}[theorem]{Corollary}
\newtheorem{lemma}[theorem]{Lemma}
\newtheorem{proposition}[theorem]{Proposition}

{\theorembodyfont{\upshape}
\newtheorem{remark}[theorem]{Remark}

}
\numberwithin{equation}{section}
\numberwithin{theorem}{section}

\newcommand{\qed}{\hfill$\Box$}
\newenvironment{proof}{\begin{trivlist}\item[]{\em Proof:}\/}{%
\qed\end{trivlist}}
\newenvironment{proofof}[1]{%
\begin{trivlist}\item[]{\em Proof of #1}\ }{\qed\end{trivlist}}

\newcommand{\E}{{\mathbb E}}
\newcommand{\Z}{{\mathbb Z}}

\newcommand{\R}{{\mathbb R}}
\newcommand{\C}{{\mathbb C\hspace{0.05 ex}}}
\newcommand{\N}{{\mathbb N}}
\newcommand{\T}{{\mathbb T}}

\newcommand{\cf}[1]{{\mathbbm 1}_{\{#1\}}}

\newcommand{\abs}[1]{\left| #1 \right|}

\newcommand{\ci}{{\rm i}}
\newcommand{\re}{{\rm Re\,}}

\newcommand{\rme}{{\rm e}}
\newcommand{\rmd}{{\rm d}}

\newcommand{\FT}[1]{\widehat{#1}}
\newcommand{\IFT}[1]{{\widecheck{#1}}}

\DeclareMathOperator*{\dom}{dom}

\newcommand{\mean}[1]{\langle #1\rangle}

\newcommand{\defem}[1]{{\em #1\/}}

\newcommand{\regabs}[1]{\langle #1 \rangle}

\newcommand{\wick}[1]{\mathlarger{:}#1\mathlarger{:}}

\newcommand{\cnls}[1]{{\mathcal{C}_{\text{NLS},#1}^{\lambda}}}
\newcommand{\ccl}[1]{{\mathcal{C}_{\text{CL},#1}^{\lambda}}}
\newcommand{\cbos}{{\mathcal{C}^{\lambda}}}
\newcommand{\cbiq}{{\widetilde{\mathcal{C}}^{\lambda}}}

\newcommand{\sos}[1]{{\textbf{\textcolor{red}{#1}}}}

\newcommand{\ind}{{\mathbbm{ 1}}}

\newcommand{\xl}{{\mathcal{X}_L}}

\newcommand{\pp}{\beta}
\newcommand{\tildepp}{\widetilde{\pp}}

\newcommand{\norm}[1]{\Vert #1\Vert}
\newcommand{\subnorm}[2]{\norm{#1}_{#2}}
\newcommand{\ppnorm}[2]{\subnorm{#1}{\pp,#2}}
\newcommand{\phinorm}[1]{\subnorm{#1}{{\Phi}}}
\newcommand{\maxnorm}[1]{\subnorm{#1}{\pp,\max}}
\newcommand{\sobolevnorm}[2]{\norm{#1}_{#2,1}}
\newcommand{\supsobolevnorm}[2]{\norm{#1}_{#2,\infty}}

\newcommand{\mvp}[1]{{\sobolevnorm{\FT{V}}{#1}^2}}
\newcommand{\mv}{{ \mvp{\frac{2}{3}} }}

\newcommand{\Cphi}{{C_{\phinorm{\cdot}}(\mathcal{X}_L,\C)}}
\newcommand{\Cmax}{{C_{\maxnorm{\cdot}}(\mathcal{X}_L,\C^2)}}
\newcommand{\CT}[1]{{C_{\subnorm{\cdot}{#1}}([0,#1]^2\times\mathcal{X}_L,\C^2)}}

\newcommand{\gin}{{G_{\text{in}}}}
\newcommand{\win}{{W_{\text{in}}}}
\newcommand{\tmax}{ {T_{*}} }

\newcommand{\Jpoint}{ {\Psi} }
\newcommand{\Jppoint}{ {\widetilde{\Psi}} }
\newcommand{\lattice}{ \Lambda }
\newcommand{\onelattice}{ \Lambda_1 }
\newcommand{\duallattice}{ {\lattice^{\!*}} }
\newcommand{\oneduallattice}{ {\onelattice^{\!*}} }

\newcommand{\rs}[1]{{{\tilde{R}_s^{#1}}}}

\newcommand{\jcite}[1]{{[\sos{Missing ref}]}}

\newcounter{jlisti}

\title{
Finite lattice kinetic equations  for 
bosons, fermions, and discrete NLS
}

\author{
Jani Lukkarinen \thanks{\emailjani}\\[0.5em]
    Sakari Pirnes \thanks{\emailsakari}\\[0.5em]
Aleksis Vuoksenmaa \thanks{\emailaleksis}\\[0.5em]
$\,^*$,$\,^\dag$,$\,^\ddag$\UHaddress}
\date{\today}
\newcommand{\email}[1]{E-mail: \tt #1}
\newcommand{\emailjani}{\email{jani.lukkarinen@helsinki.fi}}
\newcommand{\emailsakari}{\email{sakari.pirnes@helsinki.fi}}
\newcommand{\emailaleksis}{\email{aleksis.vuoksenmaa@helsinki.fi}}
\newcommand{\UHaddress}{\em University of Helsinki, Department of Mathematics and Statistics\\
\em P.O. Box 68, FI-00014 Helsingin yliopisto, Finland}

\begin{document}

\maketitle

\begin{abstract}


We introduce and study finite lattice kinetic equations for bosons, fermions, and discrete NLS.
For each model this closed evolution equation provides an approximate description for the evolution of the appropriate covariance function in the system. It is obtained by truncating the cumulant hierarchy and dropping the higher order cumulants in the usual manner.
To have such a reference solution should simplify controlling the
full hierarchy and thus allow estimating the error from the truncation.
The harmonic part is given by nearest
neighbour hopping, with
arbitrary symmetric interaction potential of coupling strength $\lambda>0$.
We consider the well-posedness of the resulting 
evolution equation up to finite kinetic times  on a finite but large enough lattice. 
We obtain decay of the solutions and upper bounds that are independent of $\lambda$ and 
depend on the lattice size only via some Sobolev type norms of the interaction potential and initial data.
We prove that the solutions are not sensitive  to 
how the energy conservation delta function is approximated. 
\end{abstract}

\bigskip


\bigskip


\newpage
\tableofcontents

\section{Introduction}

Recent years have seen steady progress towards mathematical understanding of
wave kinetic theory which is the main tool currently used for studying wave
turbulence.  The main focus has been on the deterministic evolution given by the
nonlinear Schr\"odinger equation on a $d$-dimensional torus of length $L$,
$(L\T)^d$, using sufficiently ``chaotic'' random initial data. The goal is then
to control the accuracy of the kinetic evolution equation for the covariance of the
Fourier transform of the field, for example, by showing that the error goes to
zero in a suitable scaling limit.  From the large literature on the topic, let
us pick the works of \cite{buckmaster_onset_2021, collot_derivation_2025}, where the accuracy of the kinetic approximation is proven to hold for subkinetic times. These methods are based on
suitably tailored iterative perturbation expansions in the coupling $\lambda>0$
of either the moments using Duhamel's formula or the solution using Picard
iteration, in order to control the evolution of the moments of the random field.

A culmination of such analysis of perturbation series is given in the series of
works by Deng and Hani \cite{deng_full_2023, deng_propagation_2024} where the initial data is either
normal distributed or given by choosing uniform random phases for the initial
Fourier modes. One also needs to remove a collection of ill-behaved initial data
which occur with a probability that goes to zero in the scaling limit.  For
instance, in \cite{deng_full_2023} the error is shown to vanish, uniformly up to some
time $\tmax \lambda^{-2}$ for some sufficiently small, but fixed
kinetic time $\tmax>0$, and in a limit where $L=\lambda^{-1}$ and $\lambda\to 0$
(since the notations
between various works on such systems vary, we have here converted their result to
the notations used in this paper, as defined in Sec.~\ref{sec:microscopic}.).  In the
latter work \cite{deng_long_2024}, the authors extend the analysis to longer
kinetic times, up to any $\tmax>0$ for which the kinetic equation has
sufficiently regular solutions, and also for larger tori, with
$L=\lambda^{-p}$, $p>1$.

It has been suggested that cumulants, instead of moments, of the field would be
a better quantity to control the evolution of the error in systems of the above
type.  This was suggested in \cite{lukkarinen_wick_2016} where also a
generic relation between Wick polynomials and time derivatives of cumulants is
pointed out.  There is no rigorous analysis of the resulting evolution
hierarchy of cumulants in that reference, although a natural way of arriving at
the wave kinetic evolution equation is pointed out therein.  In a recent
preprint by Dymov \cite{dymov_kinetic_2025}, a method is proposed to control the error
via a cumulant hierarchy  up to a subkinetic time $t\le \lambda^{-2+\delta}$,
$\delta>0$, in a setup similar to \cite{deng_full_2023}.  The main difference between
these two cases is an addition of a small stochastic perturbation  in \cite{dymov_kinetic_2025}, aiming to drive the system towards one of the canonical
Gibbs measures and thus introducing additional ``chaoticity'' to the system.
More details about this setup and its analysis can be found in the earlier
work \cite{dymov_large-period_2023}.

Here, we focus on a different aspect of such wave kinetic evolution problems
for three models:
the discrete nonlinear Schr\"odinger equation (DNLS) and 
interacting bosons and fermions hopping on a finite lattice. For the
comparison between these lattice models, see e.g.  \cite{lukkarinen_not_2009}. 
Such discrete evolution equations for bosons and fermions arise naturally from 
tight-binding approximation in solid-state physics \cite{solid_state_physics,
band_theory}.  For justification of the quantum Boltzmann equation in the weak
coupling limit we refer to \cite{erdos_quantum_2004,
benedetto_considerations_2004,
benedetto_weak-coupling_2005,benedetto_n-body_2007,lukkarinen_not_2009}.

We assume generic 
initial data whose distribution is both gauge and translation invariant.
Implementing the ``closure assumption'' by assuming the initial
data is such that after the first Duhamel iterate higher order cumulants
contribute only to the error terms, we obtain a closed evolution equation for
the covariance. We call this the ``truncated hierarchy'' (more details are
given below in Sec.~\ref{sec:microscopic} and in the Appendix
\ref{sec:derivtrunchier}).  How long will the resulting {\it finite lattice kinetic equation}
be well-posed \defem{without} taking any scaling limits?

We consider a finite $d$-dimensional
square lattice with side length $L$ and $d\ge3$. 
The free part of the evolution can lead to more
involved dispersion relations than those of the continuum equation, and here we
also extend the analysis to more general interaction potentials.  In addition,
part of the well-posedness goal for the truncated hierarchy is to ensure
that the solutions remain uniformly bounded for all large enough lattice sizes
$L$. One motivation for studying such a problem is to have a solution which
should converge to a solution of the standard wave kinetic
equation when $L\to\infty$, but which could still be used as a reference
solution for the dominant part of the evolution when computing the error with
respect to the original covariance function.  To have such a reference solution
should simplify controlling the full hierarchy and thus the above mentioned
error terms.

Our main result is to show that indeed the evolution equation given by the
truncated hierarchy is well-posed in the above sense, at least for finite
kinetic times $\tmax\lambda^{-2}$.  Moreover, we can control the decay of the 
solution in the lattice variable, indicating decay of spatial correlations of the 
original field. The truncated hierarchy yields an evolution
equation involving a convolution in time, which is a memory term of the type
appearing in renewal equations, but we also show that, up to errors which are
bounded as $\lambda^p$ for a certain $p>0$, it is possible to replace it by a
solution to a simpler evolution equation where the convolution is replaced
by an approximate energy conservation $\delta$-function. 

Perhaps not too surprisingly, the analysis of the approximate kinetic equations
involves some complex estimates and bookkeeping. Fortunately for us, the
discrete nonlinear Schrödiger equation is structurally quite similar to the
systems of weakly interacting bosons and fermions on a lattice, even though
these three are physically different systems. We are thus also able to treat
the equations which one obtains using the closure assumption for the cumulant hierarchies of the weakly interacting bosonic and
fermionic systems, as explained in Sec.~\ref{sec:microscopic}.. This allows us
to consider them in parallel in the following proofs and to obtain comparable
results in all three settings.

\subsection{Underlying microscopic systems and their wave kinetic theory}
\label{sec:microscopic}

The main dynamical variables for the microscopic model are the field variables
$a(x,t)$ where $t\in \R$ denotes time and $x$ a position vector on a
$d$-dimensional square lattice $\lattice\subset \Z^d$, 
$d\ge3$  with periodic boundary conditions and side length $L\in\N$, $L\ge 2$.
We have $\lattice=\onelattice^d$, where $\onelattice\subset\Z$ is a one-dimensional lattice.
We use a parametrization in which the origin lies at the center of the lattice.
For example, if $L$ is odd, then
\begin{align*}
    \onelattice=\left\{ -\frac{L-1}{2},\dots,\frac{L-1}{2} \right\}\,.
\end{align*}
There always holds $|\lattice|=L^d$ and $|x_i|\leq L/2$ for $x\in\lattice$ and $i=1,\dots,d$.  For the first model considered in this work, namely DNLS, 
$a$ is a complex-valued field, i.e., each
$a(x,t)\in \C$.  For weakly interacting bosons (fermions), $a(x,0)$ is the bosonic (fermionic)
annihilation operator for the one-particle state given by the Kronecker delta
at $x$, and $a(x,t)$ is the corresponding operator evolved to time $t$.  We
will also need the adjoints of the field, $a^*(x,t)$: for DNLS, the
correspondence is simply given by complex conjugation, $a^*(x,t)\coloneqq
(a(x,t))^*\in \C$, while in the bosonic (fermionic) case $a^*(x,t)$ is obtained from
the time-evolution of the bosonic (fermionic) creation operator $a^*(x,0)$ which is related to
the operator adjoint of $a(x,0)$.  We will also employ a condensed notation
where a sign-parameter $\sigma\in \{-,+\}\coloneqq\{-1,+1\}$ is added to denote the
adjoint: we set
\[
 a_t(x,+) \coloneqq a(x,t)\,, \qquad a_t(x,-) \coloneqq a^*(x,t)\,.
\]

The evolution equations for all three cases are determined by the Hamiltonian
\begin{align}
\label{eq:hamiltonian}
 H \coloneqq \sum_{x,y\in \Lambda} \alpha(x-y) a^*(x) a(y) +
 \lambda \frac{1}{2} \sum_{x,y\in \Lambda;\,x\ne y} V(x-y) a^*(x) a^*(y)
 a(y) a(x)\,,
\end{align}
where $\alpha:\Lambda \to \R$ is called the \defem{harmonic hopping potential}
while $V:\Lambda \to \R$ is called the \defem{interaction potential} with a
coupling $\lambda>0$.  Periodic boundary conditions are implemented by assuming
that the difference ``$x-y$'' is computed using periodic arithmetic, i.e., it denotes
$(x-y)\bmod \Lambda$.  In this work, we only consider harmonic nearest
neighbour hopping potentials $\alpha$. The interaction potential $V$ may be
any function, but some decay will be needed to obtain $L$-independent bounds. We will always
assume that both potentials are symmetric under lattice reflections, e.g., that
$V(-y)=V(y)$ for all $y\in \Lambda$.

Albeit through completely different routes, for all three models the Hamiltonian \eqref{eq:hamiltonian}
results in the following evolution equations of the fields: for all $t\in \R$
and for any $x\in \Lambda$, $\sigma \in \{\pm\}$, we then have
\begin{align}\label{eq:atxevol}
 \partial_t a_t(x,\sigma) = -\ci \sigma \sum_{y\in \Lambda} \alpha(x-y) a_t(y,\sigma) -\ci\sigma \lambda  \sum_{y\in \Lambda} V(x-y)
 \begin{cases}
  a_t(y,-) a_t(y,+)a_t(x,+)\,, & \sigma=+ \\
  a_t(x,-) a_t(y,-)a_t(y,+)\,, & \sigma=-
 \end{cases}
 \,.
\end{align}
More details about the derivation and well-posedness of this evolution equation
from $H$ in the DNLS case may be found in \cite{lukkarinen_weakly_2011} and for the bosons in
\cite{lukkarinen_not_2009} or by adapting the results
stated for fermions in \cite{cadamuro_kinetic_2018}.  

Without the interaction term, if $\lambda=0$, the evolution problem
(\ref{eq:atxevol}) can be solved explicitly using discrete Fourier transform.
Namely, we define the \defem{dual lattice} of $\Lambda$ as  
$$
\duallattice \coloneqq \frac{1}{L}\lattice\subset \left(-\frac{1}{2},\frac{1}{2}\right]^d =: \T^d,
$$
and for any function $f$ from the lattice $\lattice$ to a vector space, we set its finite discrete Fourier transform 
\begin{align}
    \label{eq:FT}
     \FT{f}(k) \coloneqq \sum_{y\in \Lambda} \rme^{-\ci 2\pi k\cdot y} f(y)\,, \qquad k\in \Lambda^*\,.
\end{align}
Similarly, the discrete inverse Fourier transform $\IFT{g}$ of a function $g$ from the dual
lattice $\duallattice$ to a vector space is defined by 
\begin{equation}
    \label{eq:IFT}
   \IFT{g}(x) \coloneqq 
     \int_{\duallattice}\!\rmd k\,  g(k)
     \textrm{e}^{\textrm{i} 2\pi k \cdot x}\,, \qquad x\in \lattice\,,
\end{equation}
where $\int_\duallattice\rmd k$ is a shorthand notation for $\frac{1}{|\duallattice|}\sum_{k\in\duallattice}$.
It is straightforward to check that the discrete Fourier transform and 
the discrete inverse Fourier transform 
are pointwise invertible by each other.
Taking now the Fourier transform of the evolution equation \eqref{eq:atxevol} we get
\begin{align}
    \nonumber
    &\partial_t \FT{a}_t(k_0,\sigma)=-\ci\sigma\omega(k_0)\FT{a}_t(k_0,\sigma)
    \\
    &
    \label{eq:FTa_evol}
    -\ci \sigma\lambda
    \int_{(\duallattice)^3}\rmd k'_1 \rmd k'_2 \rmd k'_3
    \delta_L(k_0-k'_1-k'_2-k'_3)
    \FT{V}(k'_1,k'_2,k'_3;\sigma)
    \FT{a}_t(k'_1,-)\FT{a}_t(k'_2,\sigma)\FT{a}_t(k'_3,+),
\end{align}
where $\delta_L(k)=|\lattice|\ind(k{ = }0 \mod\duallattice)$.  Here and in the following, 
we use the shorthand notation $\ind(\cdot)$ for generic characteristic functions:
$\ind(P)=1$ if the condition $P$ is true, and 
$\ind(P)=0$. otherwise.
Here,
\[
 \omega(k) \coloneqq \FT{\alpha}(k)
\]
is the \defem{dispersion relation} of the interaction $\alpha$ and we define
$\FT{V}(k_1,k_2,k_3;+)\coloneqq \FT{V}(k_1+k_2)$ and $\FT{V}(k_1,k_2,k_3;-)\coloneqq
\FT{V}(k_2+k_3)$.  By the symmetry assumption, both $\omega(k)$ and
$\FT{V}(k)\in \R$.  Clearly, if $\lambda=0$, we then have for all $t,k,\sigma$,
\[
 \FT{a}_t(k,\sigma)=\rme^{-\ci \sigma\omega(k) t} \FT{a}_0(k,\sigma)\,.
\]

In the present work, we only consider nearest neighbour harmonic interactions and, for
simplicity, scale the interaction parameter to one: we assume that there is
some $c_0\in \R$ such that
\[
 \alpha(x) \coloneqq c_0 \ind(x=0) - \frac{1}{2}\sum_{i=1}^d
 \left(\ind(x=e_i)+\ind(x=-e_i)\right)\,,
\]
where $e_i \in \R^d$ denotes the Cartesian unit vector in direction $i$.  Then,
$\omega(k)$ is obtained by evaluation at $k\in \Lambda^*$ of the trigonometric
polynomial
\begin{align}
    \label{eq:omega}
    \omega(k)= c_0-\sum_{i=1}^{d} \cos(2\pi k_i), \qquad k\in\T^d\,.
\end{align}
Due to this relation, we can simply use the trigonometric polynomial to obtain
the correct eigenvalues for any lattice size $L$.

Next we will work out the truncated hierarchy for the DNLS case, which has
already been discussed in the earlier work
\cite{lukkarinen_wick_2016} for the onsite interactions, with
$V(x)=\ind(x=0)$.  We also use the condensed notations introduced there for
moments, cumulants, and Wick polynomials.  
We consider random initial data given by a measure that has moments up to all orders and is both gauge and translation invariant. All of these three properties are then preserved by the time evolution. 
More precisely, by translation and gauge invariance we mean that the measure is invariant under the lattice
translations and under rotations of the total phase of the lattice
field, respectively.
By gauge invariance, any moment $\E[a_t^I]$
and cumulant $\kappa[(a_t)_I]$ for a sequence
$I=((k_\ell,\sigma_\ell))_{\ell=1}^n$ of length $n$ which has an imbalance of
adjoints, i.e., for which $\sum_{\ell=1}^n \sigma_\ell\ne 0$, must be zero.  In
particular, all odd moments and cumulants vanish, including the expectation
values $\E[a_t(k,\sigma)]$ which correspond to the cases where $n=1$.
 
Similarly, translation invariance implies that the Fourier transformed fields satisfy
\[
 \E[\FT{a}_t^I] = \delta_L\!\left(\sum_{\ell=1}^n k_\ell\right)
 \frac{1}{|\Lambda|}\E[\FT{a}_t^I]\,,\qquad
 \kappa[(\FT{a}_t)_I] = \delta_L\!\left(\sum_{\ell=1}^n k_\ell\right)
 \frac{1}{|\Lambda|}\kappa[(\FT{a}_t)_I]\,.
\]
Therefore, we can conclude that for the covariance, when $n=2$,
for each $T\ge 0$, $\lambda>0$, and
$k\in \Lambda^*$, we can find $W_T^{\lambda,L}(k)\in \R$ such that
\begin{align}
    \delta_L(k'+k)W_T^{\lambda,L}(k)= \E[\FT{a}_{T\lambda^{-2}}(k',-)\FT{a}_{T\lambda^{-2}}(k,+)]\,,
\end{align}
for all $k'\in \Lambda^*$.  In addition, since the mean of the field is zero,
these are also equal to the corresponding elements of the covariance matrix at
$t=T\lambda^{-2}$, and all other elements of the covariance matrix are zero.

We have collected in Appendix \ref{sec:derivtrunchier}
the argument how the closure assumption and the truncated cumulant hierarchy
result in an evolution equation for $W$. In short:
\begin{align}
    W^{{\lambda,L}}_T(k_0) &= W_0(k_0)+
    \int_0^T \rmd t
    \cnls{+}(W_t^{\lambda,L},{T-t})(k_0) + (\text{higher order terms})\,,
    \label{eq:WrealevoNLS}
    \end{align}
    where DNLS collision operator  is defined by 
    setting $\theta=+$ in
\begin{align}
    \nonumber
    \cnls{\theta}(h,\tau_0)(k_0) &\coloneqq \pi \int_{(\duallattice)^3} 
    \rmd k_1\rmd k_2\rmd k_3 
    \delta_L(k_0+k_1-k_2-k_3)
    \delta_{\lambda,\tau_0}(\omega_0+\omega_1-\omega_2-\omega_3)
    \\
                         &\times
    \left( \FT{V}(k_1-k_2)+\theta\FT{V}(k_1-k_3) \right)^2
    \Big( h_1h_2h_3+h_0h_2h_3-h_0h_1h_3-h_0h_1h_2 \Big).
    \label{eq:cnls}
\end{align}
The case ``${\theta=-}$'' will be needed later to express the collision operator 
for fermions. Here we use the shorthand
notations $\omega_j=\omega(k_j)$ and $h_j=h(k_j)$. Moreover, for each $\tau_0,
\Omega\in\R$ we have defined
\begin{align}
     \delta_{\lambda,\tau_0}(\Omega)
     \coloneqq \frac{1}{2\pi}
     \int^{\lambda^{-2}\tau_0}_{-\lambda^{-2}\tau_0} \rmd s
    e^{-is\Omega}
    .
    \label{eq:energy_delta}
\end{align}
We note that $\delta_{\lambda,\tau_0}(0)=\lambda^{-2}\tau_0/\pi$ and
$\delta_{\lambda,\tau_0}(\Omega)=
\frac{\sin(\lambda^{-2}\tau_0\Omega)}{\Omega}$ for $\Omega\ne0$. This means that
there always holds $\delta_{\lambda,\tau_0}(\Omega)\in\R$.

The corresponding arguments for cumulants of weakly
interacting bosons and fermions are more involved. (In the quantum case, cumulants are often called truncated correlation functions.)  We are not aware of a clean reference in
the lattice case but they will be part of our forthcoming work
\cite{jani_sakari}.  However, in analogy with the continuum case and the
above DNLS computations, we postulate here that, assuming a translation and
gauge invariant initial state $\E_0$,
the bosonic and fermionic truncated hierarchies satisfies the following equation
for the Wigner function which is defined by the corresponding operator formula
\begin{align}
    \delta_L(k'+k) W_T^{\lambda,L}(k)= \E[\FT{a}_{T\lambda^{-2}}(k',-)\FT{a}_{T\lambda^{-2}}(k,+)]\,.
\end{align}
We can consider all three of the above truncated hierarchy equations simultaneously
by introducing new parameters  $q\in\{0,1\}$ and $\theta\in\{-,+\}$, where
$(q,\theta)=(0,+)$ corresponds to the DNLS, and $(q,\theta)=(1,+)$ to the bosonic case,
and $(q,\theta)=(1,-)$ to the fermionic case.  For either choice
of $q$ and $\theta$, we then find
that
\begin{align}
    W^{{\lambda,L}}_T(k_0) &= W_0(k_0)+
    \int_0^T \rmd t
    \cbos(W_t^{\lambda,L},{T-t})(k_0)
    \label{eq:Wrealevo}
                    \ + \ \text{"higher order terms"},
    \end{align}
where the collision operator is defined as
\begin{align}
    \cbos(h,\tau_0)(k_0)
    &\coloneqq \theta\cnls{\theta}(h,\tau_0)(k_0)+q\ccl{\theta}(h,\tau_0)(k_0)
\end{align}
for $h:\duallattice\to\R$,
$\tau_0>0$ and $k_0\in\duallattice$.
The collision operator $\cnls{\theta}$ has already been defined in (\ref{eq:cnls}),
and the collision operator $\ccl{\theta}$, which only appears in the bosonic and fermionic cases, is
defined by
\begin{align}
    \nonumber
    \ccl{\theta}(h,\tau_0)(k_0) &\coloneqq \pi \int_{(\duallattice)^3}
    \rmd k_1\rmd k_2\rmd k_3 
    \delta_L(k_0+k_1-k_2-k_3)
    \delta_{\lambda,\tau_0}(\omega_0+\omega_1-\omega_2-\omega_3)
    \\
                         &\times
    \left( \FT{V}(k_1-k_2)+\theta\FT{V}(k_1-k_3) \right)^2
    \Big( h_2h_3-h_0h_1\Big).
    \label{eq:ccl}
\end{align}

We will be studying the finite-lattice kinetic equation, a truncated version of the equation \eqref{eq:Wrealevo},
given by
\begin{align}
    \label{eq:Wevo}
    W_T(k_0) = W_0(k_0) + \int_0^Tdt \cbos(W_t,\tau(T,t))(k_0).
\end{align}

Here we suppose that $\tau:\R_+^2\to\R$ satisfies for all $T,t\in\R_+$ either 
\begin{align}
    \label{eq:taudiff}
    \tau(T,t)=T-t  
\end{align}
or
\begin{align}
    \label{eq:tauconst}
    \tau(T,t)=T_0  
\end{align}
for some constant $T_0>0$. 
The two possibilities \eqref{eq:taudiff} and \eqref{eq:tauconst} for $\tau$ correspond,
respectively, to either having  in the collision operator $\cbos$ the memory term
$\delta_{\lambda,T-t}$ resulting a convolution in time or the simpler
approximate energy conservation $\delta$-function $\delta_{\lambda,T_0}$ which
is constant in time.
Note that in \eqref{eq:Wrealevo} we have $\tau(T,t)=T-t$.  Therefore, the
simpler evolution equation, the equation \eqref{eq:Wevo} with $\tau(T,t)=T_0$,
can be seen as a further approximation after the truncation.  The motivation for
this is that the function $\Omega\mapsto\delta_{\lambda,T_0}(\Omega)$ converges
uniformly in $t$ and $T$ to the $\delta$-function as $\lambda\to0$. Therefore
one can at least formally see that at the limit $L\to \infty$, $\lambda\to0$
the truncated equation \eqref{eq:Wevo} with $\tau(T,t)=T_0$ converges to the
corresponding wave kinetic equation for \((q,\theta)=(0,+)\) and Nordheim-Boltzmann equation for \(q=1\):  
\begin{align}
    \label{eq:Wkinetic}
    W_T(k_0) &=W_0(k_0)+
    \int_0^T \rmd t
    \mathcal{C}(W_t)(k_0),\quad k_0 \in \T^d
                        \\
    \nonumber
    \mathcal{C}(h)(k_0)&= 
    \pi 
    \int_{(\T^d)^3}\rmd k_{1}\rmd k_{2}\rmd k_{3}
    \delta(k_0+k_1-k_2-k_3)
    \delta(\omega_0+\omega_1-\omega_2-\omega_3)
    \\
    \nonumber
   &\times
    \left( \FT{V}(k_1-k_2)+\theta\FT{V}(k_1-k_3) \right)^2
    \\
    \nonumber
   &\times
    \Big(
   \theta(h_1h_2h_3+h_0h_2h_3
   -h_0h_1h_3-h_0h_1h_2)
   +{q \left( h_2h_3-h_0h_1 \right) }
   \Big).
\end{align}
 For the definition
of the measure in the collision operator \eqref{eq:Wkinetic} see
\cite{lukkarinen_global_2015}.
The well-posedness, finite time blow-up, and condensation for the 
    bosonic case $(q,\theta)=(1,+)$ to \eqref{eq:Wkinetic} in \(\R^d\) have been studied in
\cite{lu_isotropic_2004, escobedo_blow_2014, escobedo_finite_2015,
lu_long_2016, li_global_2019}.
For fermions $(q,\theta)=(1,-)$, physical initial data satisfies $0\leq W_0\leq1$. 
Indeed, it is known that then in $\R^d$ the solutions exist globally and
remain between $0$ and $1$ \cite{dolbeault_kinetic_1994, lu_stability_2003}. For several results concerning the well-posedness, finite time blow-up, and condensation of the wave kinetic equation \((q,\theta) = (0,+)\) in the isotropic setting, we refer to \cite{escobedo_blow_2014, germain_optimal_2020}, while the stability properties of the equation are studied in \cite{collot_stability_2024, menegaki__2024}.

\section{Main results}
\label{sec:results}

We have two main results, namely Theorem \ref{thm:Wwell_posedness} and Theorem
\ref{thm:error}. Our first 
main result (Theorem \ref{thm:Wwell_posedness}) proves well-posedness of
the finite lattice kinetic equation \eqref{eq:Wevo} and uniform upper bounds for its
solutions. We do this for both choices \eqref{eq:taudiff} and
\eqref{eq:tauconst} of $\tau$, and up to kinetic time $\tmax$ which 
depends on the lattice size $L$ only via some Sobolev type norms of the
interaction potential and initial data.  Our other main result (Theorem
\ref{thm:error}) gives an error estimate 
between two solutions obtained from the finite lattice kinetic equation \eqref{eq:Wevo} with
the same initial data but two different choices for the function $\tau$. 
We prove that the error between two such solutions is
proportional to 
$\lambda^{p}$, where the exponent $p>{0}$ depends on the lattice dimension $d$ and a parameter $\pp$ specifying the norms used.
Therefore, changing the original $\tau(T,t)=T-t$ into a much simpler constant function
$\tau(T,t)=T_0$ gives only an error that vanishes when $\lambda\to0$. 

It is not obvious that the finite lattice kinetic equation \eqref{eq:Wevo} has
solutions up
to kinetic times independent of $L$.  
 The main challenge for both Theorems is
controlling the oscillatory integrals produced by the $\delta_{\lambda,\tau(T,t)}$ term in the collision operator 
$\cbos$. This is because 
the trivial bound 
$|\delta_{\lambda,\tau_0}|\leq \lambda^{-2}\tau_0/\pi$ is not enough for
obtaining the well--posedness for kinetic times nor the desired decay for the error estimates of
Theorem \ref{thm:error}. Our solution to these problems revolves around 
introducing a {\it collision control map} $F[W]$ which has a closed evolution equation and 
can be used to rewrite  the cubic collision operator $\cbos$ as a bilinear operator 
with $F[W]$ and $W$ its inputs. Given $W\in C([0,\tmax]\times\duallattice,\C)$, the 
collision control map $F[W]$ is a continuous function from $[0,\tmax]^2\times\xl$ to $\C^2$, where
\begin{align*}
    \xl \coloneqq  \R^d \times \duallattice \times \T^d \times \{+,-\} \times \lattice.
\end{align*}
The two component functions $F[W]^{(1)}, F[W]^{(2)} \colon [0,\tmax]^2 \times\mathcal{X}_L \to \C$ of $F[W]$ are defined by
setting 
\begin{align}
    F[W]_{t_1,t_2}^{(n)}(R,   k',u,\sigma,x) 
    \coloneqq 
    \int_{\duallattice} \rmd k_0 
    \rme^{\ci \phi(k_0, R, u)} 
    W_{t_1}(\sigma k_0) ^{n-1}W_{t_2}(k'+k_0)
    \rme^{\ci2\pi k_0\cdot x}
    \label{eq:F}
\end{align}
for $n\in\{1,2\}$, $(t_1,t_2)\in[0,\tmax]^2$ and $(R,   k',u,\sigma,x)\in\xl$.
Here the generic phase function \(\phi \colon \R^d \times \R^d \times \R^d \to \R\) is defined as
\begin{align}
    \label{eq:phi}
\phi(k;R,u) \coloneqq \sum_{i=1}^d R_i \cos(2\pi(k_i + u_i)).
\end{align}
Given a function $\win:\duallattice\to\C$, we will use the shorthand notation
$F[\win]$ to denote the function $F[W]_{0,0}\colon\xl\to\C^2$ that is given by
the formula \eqref{eq:F} with $W_0=\win$ and $t_1=t_2=0$.

The definition \eqref{eq:phi} of the generic phase function $\phi$ is sufficient for studying
the
nearest neighbour \eqref{eq:omega} dispersion relation $\omega$.
This is because all needed linear combinations of functions $\omega$ can be represented as a single 
$\phi$ due to the summation rule of generic phase functions in 
Lemma \ref{lem:cossum}. Allowing for a general class of dispersion relations $\omega$ is 
not completely straightforward. At the very least, one would need to find another collection 
of generic phase functions which would  satisfy a similar summation rule and
some propagator estimates as in Proposition \ref{prop:propagator_multi_d}.

Note that the collision control map $F[W]$ contains all the information of $W$, i.e., 
\begin{align}
    \label{eq:FisFTofW}
    \IFT{W}_T(x)
    = 
    F[W]_{0,T}^{(1)}(0,0,0,+,x) 
    , \quad x\in\lattice, T\in[0,\tmax].
\end{align}
This means that having some control of $F[W]$ gives automatically control of $W$.

Our main results hold for all sufficiently large lattices. Recall that the size 
of the $d$-dimensional square lattice $\lattice$ is characterized by its side length 
$L\ge2$. For a chosen  parameter $\pp\in(0,1-2/d)$, the lower bound for $L$ is given by 
\begin{align}
    \label{eq:Lzero}
    L_\pp(\lambda,\tau_0)\coloneqq 
    \max((4\tau_0\lambda^{-2})^{1/\pp}, \tau_0\lambda^{-2}).
\end{align}
Here, $\lambda>0$ is the coupling parameter and we take $\tau_0>0$ to 
be the maximum of $\tau$ over  
the studied time interval, where the function $\tau$ appears in the 
evolution equation \eqref{eq:Wevo}. Note that $L_\pp(\lambda,\tau_0)\to \infty$
if  $\lambda^{-2}\tau_0\to\infty$, and $L_\pp(\lambda,\tau_0)\to 0$ if 
$\lambda^{-2}\tau_0\to0$. In particular, any $L\ge2$ satisfies  $L\ge L_\pp(\lambda,\tau_0)$ 
if $\lambda\ge
2\sqrt{\tau_0/L^{\pp}}$.

We will next present the precise statements of our main results. The detailed definitions of supremum
norms $\subnorm{\cdot}{\infty}$, $\maxnorm{\cdot}$ and $\ppnorm{\cdot}{T}$ and
Sobolev type norms $\supsobolevnorm{\cdot}{p}$ and $\sobolevnorm{\cdot}{p}$
used in the Theorems will be given in Section \ref{sec:norms} after the main 
results.

\begin{theorem}
    \label{thm:Wwell_posedness}
     \textup{(Well-posedness)}
    Take a lattice dimension $d\ge3$ and a side length $L\in\N$ with $L\ge2$. 
    Take any symmetric interaction potential
    $V:\lattice\to\R$.
    Consider initial data $\win\colon \duallattice \to
    \R$.
    Let $\pp\in(0,1-2/d)$.
    Take a final kinetic time $\tmax>0$ satisfying
    \begin{align}
        \label{eq:tmax_thm1}
        \tmax< \frac{1}{c_{\pp,1}
        \mv (1+\maxnorm{F[\win]})^2}.
    \end{align}
    Here, $c_{\pp,1}=\max(c_{\pp,2},\frac{\sqrt{3}}{6}c_{3})$ with the constants
    $c_{\pp,2}, c_{3}>0$ depending only on the lattice dimension $d$ and $\pp$, specified
    in \eqref{eq:c2n3}.

    Take a coupling parameter $\lambda>0$ and $\tau\colon\R_+^2\to\R$ such that one of the
    following holds
    \begin{enumerate}
        \item 
            \label{item:thm1case1}
         $\tau(T,t)=T-t$ and $L\ge L_\pp(\lambda,\tmax)$, or
        \item
            \label{item:thm1case2}
            $\tau(T,t)=T_0$ for some constant $T_0>0$ and $L\ge L_\pp(\lambda,T_0)$.
    \end{enumerate}
    Take parameters $q\in\{0,1\}$ and $\theta\in\{-,+\}$.
    Then there exists a unique $W\in
    C([0,\tmax]\times\duallattice,\R)$ that satisfies equation
    \eqref{eq:Wevo} with $\tau$ and $W_0=\win$. The solution $W$ depends continuously
    on the initial data $\win$. Moreover, for all $T\in [0, \tmax]$ there holds
    \begin{align}
        \label{eq:Wsupbound}
        \subnorm{W_T}{\infty}
        &\leq
        \subnorm{\win}{\infty}
        \rme^{c_{3}\sobolevnorm{\FT{V}}{\frac{1}{6}}^2\int_0^T\rmd t \maxnorm{F[W]_{t,t}}} ,
    \end{align}
    and
    \begin{align}
        \label{eq:Fnormpropagation}
        \sup_{t_1,t_2\in[0,T]}
        \maxnorm{F[W]_{t_1,t_2}}
        &\leq
        \frac{\maxnorm{F[\win]}}{\sqrt{1-c_{\pp,2}\mv\maxnorm{F[\win]}T}}.
    \end{align}
    In particular, then also 
    \begin{align}
        \supsobolevnorm{W_T}{1/3}
        \label{eq:Wsobolevbound}
        &\leq
        \sup_{t_1,t_2\in[0,T]}
        \maxnorm{F[W]_{t_1,t_2}}
        \leq
        \frac{2}{\sqrt{3}}
        \maxnorm{F[\win]}.
    \end{align}
\end{theorem}

\begin{remark}
    Note that the solution $W$ for real-valued initial data $\win$ remains real-valued
    for any of the finite lattice kinetic equations considered here. However, we cannot 
    guarantee propagation of other natural properties of $W$, such as positivity and, for 
    fermions, boundedness by $1$.
    Also, the control of the solution $W$ after the final kinetic time $\tmax$ requires
    future work, for example, to study its global well-posedness or blow-up.
\end{remark}

Note that Theorem \ref{thm:Wwell_posedness} is stated for fixed $\lambda$ and $L$.
The functions $\FT{V}$ and $\win$ depend on $L$ via
their domain $\duallattice$ and $F[W]$ also by its definition. Also, the weight
in the norm $\maxnorm{\cdot}$ depends on $L$.

The upper bounds \eqref{eq:tmax_thm1}--\eqref{eq:Wsobolevbound} in  Theorem
\ref{thm:Wwell_posedness} are
independent of $\lambda$ and depend on $L$ only via the norms $\subnorm{\win}{\infty}$,
$\sobolevnorm{\FT{V}}{2/3}$, $\sobolevnorm{\FT{V}}{1/6}$, $\maxnorm{F[\win]}$ and $\maxnorm{F[W]_{t,t}}$.
Therefore, if one considers a sequence of solutions in $L$, then one
obtains upper bounds that are also independent of $L$ if
these norms are uniformly
bounded in $L$. For example, 
consider sequences  of initial data $\win^{(L)}$ and 
symmetric interaction potential $V^{(L)}$ for which there are constants $C_W,C_V,C_F>0$
such that $\subnorm{\win^{(L)}}{\infty}\leq C_W$,
$\maxnorm{F[\win^{(L)}]}\leq C_F$ and
$\sobolevnorm{\FT{V}^{(L)}}{\frac{2}{3}}^2\leq C_V$ for all $L\ge2$.
Fix final kinetic time $\tmax>0$ satisfying 
\[
    \tmax < \frac{1}{c_{\pp,1} C_V(1+C_F)^2}.
\]
Then there exists a unique
solution
$W^{(L)}$ for each $L$ on the fixed interval $[0,\tmax]$ as long as $L\ge
L_\pp(\lambda,\tau_0)$, where $\tau_0=\tmax$ in the case \ref{item:thm1case1} of Theorem \ref{thm:Wwell_posedness}
 and 
 $\tau_0=T_0$ in the case \ref{item:thm1case2}. Here, $\lambda$  and $T_0$
 may also depend on $L$. Moreover, the constants $C_W,C_V,C_F$ can be used to 
 obtain uniform upper bounds from \eqref{eq:Wsupbound}--\eqref{eq:Wsobolevbound}.

The Sobolev type norm 
$\sobolevnorm{\win}{1/3}$ can be used to bound $\maxnorm{F[\win]}$. By Corollary
\ref{cor:initialdata}
there exists a constant $C>0$, depending only on $\pp$ and the dimension $d$, such that 
\begin{align}
    \label{eq:maxnormboundedbyW}
     \maxnorm{F[\win]} \leq C(\sobolevnorm{\win}{1/3}+\sobolevnorm{\win}{1/3}^2).
\end{align}
Therefore, to find a constant $C_F$ that bounds $\maxnorm{F[\win^{(L)}]}$ uniformly in $L$,
it is sufficient to bound $\sobolevnorm{\win}{1/3}$.
In fact, Corollary \ref{cor:initialdata} is one of the main reasons for the
definition of the norms $\maxnorm{\cdot}$, $\sobolevnorm{\cdot}{1/3}$ and
$\sobolevnorm{\cdot}{\infty}$, including 
the definition of their weight functions. 

Note that \eqref{eq:tmax_thm1} together with \eqref{eq:Fnormpropagation}  proves
that the $\maxnorm{\cdot}$ norm of the collision control map $F[W]$ propagates
uniformly
in the time interval $[0,\tmax]$. 
However, we do not have similar control directly for  $W$. For example, in 
\eqref{eq:Wsupbound}, we cannot bound $\maxnorm{F[W]_{t,t}}$ by $\subnorm{\win}{\infty}$
without a constant factor proportional to $L$. On the other hand, we are not able to bound 
$\sobolevnorm{W_T}{1/3}$ with $\maxnorm{F[\win]}$ as is done in
\eqref{eq:Wsobolevbound} for the norm $\supsobolevnorm{W_T}{1/3}$, which would 
imply propagation of $\sobolevnorm{W_T}{1/3}$.

\begin{definition}
    \label{def:error}
    \textup{(Error term)} Suppose $d\ge3$ is given.  Let $\pp\in(0,1-2/d)$ 
     and define $\tildepp=(1-\pp)/2$.  Let $\lambda>0$ and $T>0$. Suppose both $\tau,\widetilde\tau:\R^2_+\to\R_+$ satisfy 
    one of the conditions \eqref{eq:taudiff} or \eqref{eq:tauconst}. We  define
    $\mathcal{E}_\pp(\lambda,\tau,\widetilde\tau, T)$ as follows
    \begin{enumerate}[(a)]
        \item if $\tau=\widetilde{\tau}$ then
            \begin{align}
                \mathcal{E}_\pp(\lambda,\tau,\widetilde{\tau},T) = 0,
            \end{align}
        \item if both $\tau$ and $\widetilde{\tau}$ are constants, $T_0>0$ and $\widetilde{T}_0>0$, respectively,  
             then
            \begin{align}
                \mathcal{E}_\pp(\lambda,\tau,\widetilde{\tau},T) &= 
         \eta_{d,\pp} T\left| T_0^{-(d\tildepp-1)}-\widetilde{T}_0^{-(d\tildepp-1)} \right|  
        \lambda^{2d\tildepp-2},
            \end{align}
            where $\eta_{d,\pp}=(d\tildepp-1)^{-1}(2+\frac{4}{\pi\pp})^d2^{-d\tildepp}$,
        \item if $d\tildepp\ne 2$ and only one of  $\tau$ and $\widetilde{\tau}$ is a constant and the 
            constant is $T_0>0$, then
            \begin{align}
                \mathcal{E}_\pp(\lambda,\tau,\widetilde{\tau},T) &=
                \alpha_d T 
         \min\left(1,\frac{\lambda}{\sqrt{\min(T,T_0)}}\right)^{p_d},
            \end{align}
            where $p_d=2d\tildepp-2$ for $d\tildepp<2$ and $p_d=2$ for $d\tildepp\ge2$, and \\
            ${\alpha_d=1+\eta_{d,\pp}(1+\max(|d\tildepp-2|^{-1},1))}$,
        \item if $d\tildepp = 2$ and only one of  $\tau$ and $\widetilde{\tau}$ is a constant and the 
            constant is $T_0>0$, then
            \begin{align}
                \mathcal{E}_\pp(\lambda,\tau,\widetilde{\tau},T) &=
                \left(
                1+
                \eta_{d,\pp}
                \left(
                1
                +
                \ln\left(\frac{\min(T_0,T)}{\lambda^2}+\rme\right)
                \right)
                \right)
                 T 
         \min\left(1,\frac{\lambda}{\sqrt{\min(T,T_0)}}\right)^{2}.
            \end{align}
    \end{enumerate}
\end{definition}
Note that $1<d\tildepp < d/2$, and therefore the case $d\tildepp=2$ is never possible for $d=3,4$.
We always have 
\begin{align}
                \mathcal{E}_\pp(\lambda,\tau,\widetilde{\tau},T) \to 0
\end{align}
if $\lambda\to0$ or $T\to 0$. The error function $\mathcal{E}_\pp(\lambda,\tau,\widetilde{\tau},T)$
is also continuous in $\lambda$ and $T$.

\begin{theorem}
    \label{thm:error}
     \textup{(Error estimate)}
    Take a lattice dimension $d\ge3$ and a side length $L\in\N$ with $L\ge2$. 
    Take any symmetric interaction potential
    $V:\lattice\to\R$.
    Consider initial data $\win\colon \duallattice \to
    \R$.
    Let $\pp\in(0,1-2/d)$.
    Take a final kinetic time $\tmax>0$ satisfying
    \[
         \tmax\leq\frac{1}{2c_{\pp,2}
         \sobolevnorm{\FT{V}}{2/3}^2\maxnorm{F[\win]}}.
    \]
    Here, the constant $c_{\pp,2}>0$, depending only on $\pp$ and the dimension $d$, is specified 
    in \eqref{eq:c2n3}.

    Take a coupling parameter $\lambda>0$ and
    $\tau,\widetilde{\tau}\colon\R_+^2\to\R$ such that one of the following
    holds
    \begin{enumerate}
        \item 
            $\tau(T,t)=T-t$ and $\widetilde{\tau}(T,t)=\widetilde{T}_0$ for some
            constant $\widetilde{T}_0>0$ and $L\ge L_\pp(\lambda,\max(\tmax,\widetilde{T}_0))$, or
        \item
            $\tau(T,t)=T_0$ and $\widetilde{\tau}(T,t)=\widetilde{T}_0$ for some
            constants $T_0,\widetilde{T}_0>0$ and $L\ge L_\pp(\lambda,\max(T_0,\widetilde{T}_0))$.
    \end{enumerate}
    Take parameters $q\in\{0,1\}$ and $\theta\in\{-,+\}$.
    Let $W, \widetilde{W} \in
    C([0,\tmax]\times\duallattice,\R)$, with  $W_0=\widetilde{W}_0=\win$. Suppose $W$
    is a solution to equation \eqref{eq:Wevo} with $\tau$ and $\widetilde W$ is a
    solution to equation
     \eqref{eq:Wevo} with $\widetilde \tau$. Then for all  $T\in[0,\tmax]$ there holds 
     \begin{align}
         \subnorm{W_{T}-\widetilde{W}_{T}}{\infty}
         &\leq 
         \Big(\sup_{T' \in[0,T]}
         \subnorm{\widetilde{W}_{T'}}{\infty}\Big)
         \mathcal{A}_T\Big(\sobolevnorm{\FT{V}}{1/6},\sobolevnorm{\FT{V}}{2/3},\ppnorm{F[\widetilde{W}]}{T}\Big)
         \mathcal{E}_\pp(\lambda,\tau,\widetilde{\tau},T),
     \end{align}
     and
     \begin{align}
         \ppnorm{F[W]-F[\widetilde{W}]}{T} 
         &\leq
         c_{\pp,4}\sobolevnorm{\FT{V}}{2/3}^2 \ppnorm{F[\widetilde{W}]}{T}^2
         \mathcal{E}_\pp(\lambda,\tau,\widetilde{\tau},T).
     \end{align}
     Here $\mathcal{A}_T:\R_+^3\to\R_+$,
     \begin{align}
         \mathcal{A}_T(\rho_{{1}},\rho_{{2}},m) = c_{5} 
         \rho_{{1}}^2
          m
          (1
          +
          c_{\pp,4}
          c_{6}
          \rho_{{2}}^2
          m
          T
          )
          \rme^{8c_{3}
          \rho_{{1}}^2
          m
          T
          }
     \end{align}
     and the constants $c_{3},c_{\pp,4},c_{5},c_{6}>0$, depending only on $\pp$ and the
     dimension $d$, are specified in \eqref{eq:c2n3}--\eqref{eq:c8}.
     In particular, then also
     \begin{align}
         \supsobolevnorm{W_{T}-\widetilde{W}_{T}}{1/3}
         \leq 
         c_{\pp,4} \sobolevnorm{\FT{V}}{2/3}^2 \ppnorm{F[\widetilde{W}]}{T}^2
         \mathcal{E}_\pp(\lambda,\tau,\widetilde{\tau},T).
     \end{align}
\end{theorem}

The strictly positive constants seen in the above Theorems \ref{thm:Wwell_posedness} and 
\ref{thm:error} are 
\begin{alignat}{2}
        \label{eq:c2n3}
        c_{\pp,2} &= 32(2^{2d/3}+3)c_{\pp,8}, \quad c_{3} &&=2^{d/6}(2^{d/6}+3)c_{6},
        \\
        \label{eq:c4n5}
        c_{\pp,4}&= 32(2^{2d/3}+3)c_{\pp,9}^d, \quad c_{5} &&= 2^{d/6}(2^{d/6}+3)8,
\end{alignat}
where
\begin{align}
        \label{eq:c6inf}
         c_{6}&=\inf\{c_{\pp,6}\colon 0<\pp<1-2/d\},
        \\
        \label{eq:c6}
        c_{\pp,6} &= 2+2^{d(1-\pp)/2+1}\left( 1+\frac{2}{\pi\pp} \right)^d \frac{2}{d(1-\pp)-2},
        \\
        \label{eq:c7}
        c_{\pp,7} &=
        \frac{1}{4}
        \left( 
            c_{\pp,9}\sqrt{\pi}
            \frac{\Gamma(d(1-\pp)/4-1/2)}{\Gamma(d(1-\pp)/4)}
        \right)^d,
        \\
        \label{eq:c8}
        c_{\pp,8} &= \max \left( 2d (c_{\pp,9})^d, 2c_{\pp,7} {(1-1/d)^{-d}} \right) ,
        \\
        \label{eq:c9}
       c_{\pp,9} &=  2^{7/6}\sqrt{3}\mean{16^{1/(1-\pp)}/2}^{1/3}.
\end{align}
Here, $\Gamma$ denotes the gamma function.  All these constants
\eqref{eq:c2n3}--\eqref{eq:c8} are strictly positive and depend
only on $\pp$ and the lattice dimension $d$.  We will be using these constants
\eqref{eq:c2n3}--\eqref{eq:c8} throughout the paper.

\subsection{Definitions of norms and function spaces}
\label{sec:norms}

Given two topological spaces $X$ and $Y$, we use $C(X,Y)$ to denote the set of
continuous functions from $X$ to $Y$. For the lattice $\lattice$, dual lattice
$\duallattice$, the set $\{-,+\}$, and other discrete sets we will use the corresponding discrete
topology, i.e., every subset is an open set.

For each $p\in\R$  we define Sobolev type norms $\sobolevnorm{\cdot}{p}$ and
$\supsobolevnorm{\cdot}{p}$ by setting 
\begin{align}
    \label{eq:sobolev}
    \sobolevnorm{f}{p}
    &\coloneqq
    \sum_{y\in\lattice} \left( \prod_{j=1}^d \mean{y_j}^{p} \right) |\IFT{f}(y)|,
    \\
    \label{eq:supsobolevnorm}
    \supsobolevnorm{f}{p}
    &\coloneqq
    \sup_{y\in\lattice} \left( \prod_{j=1}^d \mean{y_j}^{p} \right) |\IFT{f}(y)|,
\end{align}
for all $f\colon \duallattice\to \C$. Here, 
$\regabs{\cdot}:\R\to [1,\infty)$, $\regabs{r}=\sqrt{1+r^2}$ is the regularized
absolute value.
We will use $\subnorm{\cdot}{\infty}$ to denote the standard
sup norm, i.e.  $\subnorm{f}{\infty}=\sup_{x\in\dom(f)}|f(x)|$ for any complex
or real valued function $f$ with $\dom(f)$ being its domain.

Since $\mv$ appears quite often in the text, we introduce the shorthand
notation 
\begin{align}
    \label{eq:mv}
    M_V\coloneqq \mv.
\end{align}

For each $T>0$, $L\ge 2$ and $\pp\in(0,1-d/2)$, the $\ppnorm{\cdot}{T}$ norm, used for the collision
control map $F[W]$, is defined for any function $G:[0,T]^2\times\xl\to\C^2$ by
setting
\begin{align}
    \label{eq:norm_T}
    \ppnorm{G}{T}&\coloneqq\sup_{t_1,t_2\in[0,T]}
    \maxnorm{G_{t_1,t_2}},
\end{align}
where the $\maxnorm{\cdot}$ norm is defined 
for every $G:\xl\to\C^2$ by setting
\begin{align}
    \label{eq:norm_max}
    \maxnorm{G}
    &\coloneqq
    \max(\phinorm{G^{(1)}},\phinorm{G^{(2)}}).
\end{align}
Here, the norm $\phinorm{\cdot}$ is is defined 
for  every $G:\xl\to\C$ by setting
\begin{align}
    \label{eq:norm_phi}
    \phinorm{G} &\coloneqq \sup_{(R,k',u,\sigma,x)\in\xl}
   \Phi^d(R,x)|G(R,k',u,\sigma,x)|.
\end{align}
Here, the weight function $ \Phi^d: \R^d\times\lattice\to(0,\infty)$,
$\Phi^d(R,x)=\prod_{j=1}^d\Phi(R_j,x_j)$ is a \(d\)-fold tensor product of the
one dimensional function $\Phi: \R\times\onelattice\to (0,\infty)$, given by
\begin{align}
    \label{eq:weight_phi}
    \Phi(R_1,x_1) &=
    \regabs{x_1}^{1/3-\gamma(R_1)}
        \max 
        \left( 
            1,
            {\frac{ \varphi(R_1)} {\regabs{x_1}^{1/2}}}
        \right).
\end{align}
Here, assuming $L\ge2$ and $\pp\in(0,1-2/d)$ are given, the functions $\varphi,\varphi_1, \varphi_2  : \R\to[1,\infty)$
are defined by
\begin{align}
    \varphi(r) &=
    \ind(|r|\leq L^{\pp})\varphi_1(r)
    +
    \ind(|r|>L^{\pp})\varphi_2(r),
    \\
    \varphi_1(r) &= \langle r \rangle^{1/2},
    \\
    \varphi_2(r) &= \frac{\langle L^{\pp} \rangle^{1/2}}{\langle
    |r|-L^{\pp}\rangle^{\pp/2}},
\end{align}
and $\gamma:\R\to[0,1/3]$,
\begin{align}
         \gamma(r)
         =
     \begin{cases}
         0
         &, \ |r|\leq a_L  \\
     \frac{1}{3}\frac{r-a_L}{b_L-a_L} &,  a_L < |r| <b_L \\
         1/3
                           &, \ |r|\ge b_L
    \end{cases},
\end{align}
where $a_L = \frac{3}{2}L^{\pp}$ and $b_L=\max(2L^{\pp}, L/8)$. How $\gamma$ behaves
between $a_L$ and $b_L$ is not particularly important, except that it is an
increasing function and has a Lipschitz constant decaying faster than $L^{-\pp}$ for
all large enough $L$ (see Lemma \ref{lem:3Phi_bound}). In fact, the Lipschitz
constant of $\gamma$ is bounded by $\frac{32}{3}L^{-1}$ whenever $2L^{\pp}\leq L/8$, i.e., for
$L\ge 16^{\frac{1}{1-\beta}}$.

The weight function $\Phi$ is continuous. It is also symmetric in both arguments, i.e.,
    $\Phi(-R_1,x_1)=\Phi(R_1,x_1)$ and $\Phi(R_1,-x_1)=\Phi(R_1,x_1)$. 
    For all $R_1\in\R$ and $x_1\in\lattice$, there holds 
    $$1\leq\Phi(R_1,x_1)\leq \mean{L/2}^{1/3}.$$
On the interval $[0,L^{\pp}]$  the function $\varphi$ is strictly increasing and
is bounded below by $1$, while on the interval $(L^{\pp},\infty)$ the function
$\varphi$ is strictly decreasing and $\varphi(r)\to 0$ as $|r|\to\infty$. Moreover, 
$\varphi(r)\leq \regabs{L^{\pp}}^{1/2}$ for all $r\in\R$.

There are {two} main motivations for the definition \eqref{eq:weight_phi} of the weight
$\Phi$. The first is related to the initial data $F[\win]$ of the collision control map.
The main structure of $\Phi$ is mirroring Corollary \ref{cor:initialdata}, so that  
we can obtain the bound
\eqref{eq:maxnormboundedbyW} that allows us to use the initial data $\win$ to
check whether $\maxnorm{F[\win]}$ is bounded in $\lambda$ and $L$. To obtain the previously
discussed
bound \eqref{eq:maxnormboundedbyW}, the functions $\gamma(R_1)$ and $\varphi_2(R_1)$ are needed to get 
boundedness in $x_1$ and $R_1$, respectively, when $|R_1|>L/8$. 
In fact, there exists a constant $C_\pp>0$ such that $\Phi(R_1,x_1)<C_\pp$
for all $L\ge2$, $R_1>L/8$ and $x_1\in\onelattice$. Therefore, the  bound \eqref{eq:maxnormboundedbyW} follows automatically from
Corollary
\eqref{cor:initialdata}.
The specific 
forms of $\gamma$ and $\varphi_2$ are essentially chosen so that Lemma
\ref{lem:3Phi_bound} holds. The other main motivation for the definition 
weight $\Phi$, especially for $\varphi_1$, is that we will use it to integrate 
over the challenging $\int_{\tau_0\lambda^{-2}}^{\tau_0\lambda^{-2}}\rmd s (\dots)$ integral discussed in the beginning of Section \ref{sec:results}.

We define three Banach spaces corresponding to the above defined norms, namely, 
\begin{align}
    \Cphi &\coloneqq \{G\in C(\xl,\C) \ : \ \phinorm{G}<\infty \}
    \\
    \Cmax&\coloneqq \{G\in C(\xl,\C^2) \ : \ \maxnorm{G}<\infty \}
    \\
    \CT{T_0} &\coloneqq \{G\in C([0,T_0]^2\times\xl,\C^2) \ : \ \ppnorm{G}{T_0}<\infty \}.
\end{align}

\subsection{Structure and key ideas for proofs}
\label{sec:keyideas}

The reasons for defining the collision control map $F[W]$ through \eqref{eq:F} are twofold. The first reason is that we can rewrite the collision operator $\cbos$ as a bilinear operator 
$\cbiq(\tau_0)$, namely, 
\begin{align}
    \label{eq:bilinear_ops}
    \cbiq(\tau_0)[F[W]_{t,t},W_t]&= 
    \mathcal{C}_{{q}}^\lambda(W_t,\tau_0).
\end{align}
The other reason is that $F[W]$ itself satisfies a closed evolution equation. More precisely, there is a bilinear operator $\mathcal{I}_\tau$ such that  if $W$ is
a solution to the evolution equation \eqref{eq:Wevo} with $\tau$, then
\begin{align}
    \label{eq:bilinear_ops2}
    F[W]=\mathcal{I}_\tau[F[W],F[W]].
\end{align}
When proving the well-posedness Theorem \ref{thm:Wwell_posedness}, the idea is
to solve $F[W]$ from \eqref{eq:bilinear_ops2} and control it in the weighted sup
norm $\ppnorm{\cdot}{\tmax}$. The weight $\Phi^d$ used in the norm
 is chosen so that, after using \eqref{eq:bilinear_ops}
 together with $|F[W]|\leq \ppnorm{F[W]}{\tmax}/|\Phi^d|$, the remaining
$1/|\Phi^d|$ can be used to eliminate the problems arising from the \(\delta_{\lambda, \tau(T,t)}\)-term which have been discussed in the beginning of Section \ref{sec:results}.

To control the error between the two solutions $ W$
and $\widetilde{W}$ in Theorem \ref{thm:error}, we employ a similar strategy. 
First, control the difference  $\ppnorm{F[W]-
F[\widetilde{W}]}{\tmax}$ of the collision control maps by using properties 
of the bilinear operator $\mathcal{I}_\tau$ that appears in the evolution equation
\eqref{eq:bilinear_ops2}. Then use this together with \eqref{eq:bilinear_ops} to control 
the difference of the solution $ W$ and $\widetilde{W}$.

The main structure for the definition \eqref{eq:weight_phi} of the weight
function $\Phi$ comes from propagator estimates that will be presented in the
following Section \ref{sec:lemmata}. The definition of the bilinear operator
$\cbiq(\tau_0)$ will be given in Section \ref{sec:F} which also includes a proof of 
how \eqref{eq:bilinear_ops} can be used, together with
$F[W]$, to rewrite the collision operator $\cbos$. In Section
\ref{sec:fixpoint_operator}, we will write out the bilinear operator
$\mathcal{I}_\tau$ and show that it gives the evolution equation
\eqref{eq:bilinear_ops2} of $F[W]$. Required estimates related to the
weight function $\Phi$ and the operator $\mathcal{I}_\tau$ will be given in
Section \ref{sec:estimates}.

In Section \ref{sec:fixpoint_eqs}, we use Banach
fixed--point theorem and Gr\"onwall's inequality to prove that there exists a
unique $G$ satisfying the fixed--point equation $G=\mathcal{I}_\tau[G,G]$ and that for a given $G$ there
is a unique $\mathcal{G}$ satisfying $\mathcal{G}=\mathcal{I}_\tau[G,\mathcal{G}]$. 
Therefore, by uniqueness $\mathcal{G}=G$ whenever the initial data of $\mathcal{G}$ and $G$ are equal. 
Banach fixed-point
theorem and Gr\"onwall's inequality are also used to prove that for a given $G$
there is a unique $W$ satisfying an evolution equation similar to
$\eqref{eq:Wevo}$ but with the collision
operator $\cbos(W,\tau_0)$ replaced by $\cbiq(\tau_0)[G,W]$.

After noticing
that $F[W]=\mathcal{I}_\tau[G,F[W]]$, we get by the aforementioned uniqueness
results that $G=F[W]$. Hence, we are able to use \eqref{eq:bilinear_ops} and
conclude that $W$ is a unique solution to the evolution equation
\eqref{eq:Wevo} with the original collision operator $\cbos(W,\tau_0)$, finishing the proof of the well-posedness Theorem \ref{thm:Wwell_posedness}. This
will be done in Section \ref{sec:proofs} together with the proof of the error
Theorem \eqref{thm:error}.

\section{Preliminaries}
\label{sec:lemmata}

In this Section, we will go through preliminary estimates needed for the proof of the main results.

We will be using the Parseval's Theorem in the following form. 
\begin{theorem}
    \label{thm:parseval}
    \textup{(Parseval's Theorem)}
    For every $f,g\colon \duallattice\to\C$ there holds
    \begin{align}
        \label{eq:parseval}
        \int_{\duallattice} \rmd k f(k)g(k)^* = \sum_{x\in\lattice} \IFT{f}(x)\IFT{g}(x)^*.
    \end{align}
\end{theorem}

The following Lemma \ref{lem:regabs} considers the regularized absolute value $\mean{\cdot}=\sqrt{1+(\cdot)^2}$.
\begin{lemma}
    \label{lem:regabs}
    Let $n\in\N$. Then 
    \begin{align}
        \label{eq:regabs}
        \frac{\mean{x}}{n\prod_{j=1}^{n-1}\mean{y_j}} 
        &\leq 
        \mean{x+\sum_{j=1}^{n-1}y_j }
        \leq
        n\mean{x}\prod_{j=1}^{n-1}\mean{y_j }
    \end{align}
    for all $x\in\R$ and $y\in\R^{n-1}$. 
\end{lemma}
\begin{proof}
    The inequality on the right follows from 
    \begin{align*}
        \mean{x+\sum_{j=1}^{n-1}y_j }
        \leq
        n\max(\mean{x},\mean{y_1},\dots,\mean{y_{n-1}})
        \leq
        n\mean{x}\prod_{j=1}^{n-1}\mean{y_j }.
    \end{align*}
    Now the right hand side implies 
    \begin{align*}
        \mean{x} &= 
        \mean{x+\sum_{j=1}^{n-1}y_j-\sum_{j=1}^{n-1}y_j }
        \leq
        n\mean{x+\sum_{j=1}^{n-1}y_j}\prod_{j=1}^{n-1}\mean{y_j }.
    \end{align*}
    Hence, the left hand side of \eqref{eq:regabs} holds.
    
\end{proof}
\begin{corollary}
    \label{cor:modregabs}
    Let $n\in\N$. Then 
    \begin{align}
        \label{eq:modregabs}
        \frac{\mean{x}}{n\prod_{j=1}^{n-1}\mean{y_j}} 
        &\leq 
        \mean{(x+\sum_{j=1}^{n-1}y_j)\bmod\onelattice }
        \leq
        n\mean{x}\prod_{j=1}^{n-1}\mean{y_j }
    \end{align}
    for all $x\in\onelattice$ and $y\in\Z^{n-1}$. 
\end{corollary}
\begin{proof}
    The right hand side follows from Lemma \ref{lem:regabs} and the fact that, with 
    our parametrization for $\onelattice$, there always holds $|z\bmod\onelattice|\leq|z|$ for all 
    $z\in\Z$. 

    Take the $m\in\Z$ for which $x+\sum_{j=1}^{n-1}y_j+Lm=
    (x+\sum_{j=1}^{n-1}y_j)\bmod\onelattice$. Then
    $|x\bmod\onelattice|=|(x+Lm)\bmod\onelattice|\leq |x+Lm|$, and therefore
    the right hand side of \eqref{eq:regabs} implies 
    \begin{align*}
        \mean{x} &\leq  \mean{x+Lm}\leq
        \mean{x+Lm+\sum_{j=1}^{n-1}y_j-\sum_{j=1}^{n-1}y_j }
        \leq
        n\mean{x+\sum_{j=1}^{n-1}y_j+Lm}\prod_{j=1}^{n-1}\mean{y_j }.
    \end{align*}
    Hence, the left hand side of \eqref{eq:modregabs} holds.
\end{proof}

\begin{definition}
    \label{def:phasefunction}
 We use the following conventions for the argument map of complex numbers: If $z\in \C$ and $z\ne 0$ we set 
 $\arg(z)$ to denote the unique value in $(-\pi,\pi]$ for which $\rme^{\ci \arg(z)}=z/|z|$.  We also set $\arg(0)=0$. The restriction of $z\mapsto \arg(z)$ to $\C\setminus (-\infty,0]$ is a continuous map onto $(-\pi,\pi)$.
 For $z\in \C^n$ we denote $\arg(z)\in (-\pi,\pi]^n$, defined componentwise
 by $\arg(z)_i=\arg(z_i)$.
 \end{definition}

Recall the definition \eqref{eq:phi} of the generic phase function $\phi$. The key 
property why we have chosen that particular definition and why we restrict ourselves to 
the nearest neighbour dispersion relation $\omega$ is to be able to use the following Lemma \ref{lem:cossum}.

\begin{lemma}
    \label{lem:cossum}
Suppose $n\in \N$ and 
$R(\ell),u(\ell)\in \R^d$ for $\ell=1,2,\ldots,n$.  Then there are $\widetilde{R}$, $\widetilde{u}\in \R^d$ such that 
for all $k'\in \R^d$, we have a summation rule
\[
 \sum_{\ell=1}^n \phi(k';R(\ell),u(\ell)) =
 \phi(k';\widetilde{R},\widetilde{u})\,.
\]
In particular, we may choose here $\widetilde{R}$ and $\widetilde{u}$ equal to
\[
 \widetilde{R}(R,u)_i = \left|\widetilde{z}(R,u)_i\right|\,,
 \qquad \widetilde{u}(R,u)_i = \frac{1}{2\pi}\arg(\widetilde{z}(R,u)_i)\,,
\]
where $\widetilde{z}(R,u)\in \C^d$ is defined componentwise by
\[
 \widetilde{z}(R,u)_i = \sum_{\ell=1}^n R(\ell)_i \rme^{\ci2\pi u(\ell)_i} \,, \qquad i=1,2,\ldots,d\,.
\]
\end{lemma}
\begin{proof}
 Using the above definitions, we have
 \[
   \sum_{\ell=1}^n \phi(k';R(\ell),u(\ell))
   =\sum_{\ell=1}^n \sum_{i=1}^d R(\ell)_i\, \re\!\!\left(\rme^{\ci2\pi(k'_i+u(\ell)_i)}\right)= 
    \sum_{i=1}^d \re\!\!\left(\rme^{\ci2\pi k'_i}\widetilde{z}(R,u)_i\right)\,.
 \]
 Since $\widetilde{z}_i = \widetilde{R}_i \rme^{\ci 2\pi \widetilde{u}_i}$ with $\widetilde{R}_i,\widetilde{u}_i\in \R$, the claim follows.
\end{proof}

\begin{proposition}
    \label{prop:propagator_multi_d}
 For $d$-dimensional system, let us define
  \[
  Q(x;R,u,L) \coloneqq  \int_{\duallattice}\! \rmd k \,\rme^{\ci 2 \pi k\cdot x+ \ci \phi(k;R,u)}\,,\qquad x\in \lattice\,,\ R,u\in \R^d\,,\ L\ge 2\,,
 \]
and its infinite volume counterpart
\[
 I(y;R,u)\coloneqq 
 \int_{\T^d}\! \rmd k\, \rme^{\ci 2 \pi k\cdot y+ \ci \phi(k;R,u)}\,, \quad y\in \Z^d\,,\ R,u\in \R^d \,.
\]
There are constants $C,\delta_0>0$, which depend only on $d$, such that for all  $L\ge 2$,
$|R|_\infty\le \frac{L}{8}$, $u\in \R^d$, and 
$x\in\lattice$, we have
\[
 \left|Q(x;R,u,L)-I(x;R,u)\right|\le C 
 \rme^{-\delta_0\frac{1}{2}L} \,.
\]
Denoting $|y|_1 \coloneqq \sum_{i=1}^d |y_i|$, we also have
\[
 |I(y;R,u)|\le \rme^{-2\delta_0(|y|_1-2 |R|_1)_+}\,, \quad y\in \Z^d\,,
\]
and we can find a constant $C$ such that 
\[
 |I(y;R,u)|\le C 
 \prod_{i=1}^d \regabs{y_j}^{\frac{1}{6}} \mean{R_i}^{-\frac{1}{2}} \,, \quad y\in \Z^d\,.
\]
In particular,
\[
 |I(y;R,u)|\le C 
 \prod_{i=1}^d \regabs{y_i}^{-\frac{1}{3}}
 \min
 \left( 
 1, \sqrt{\frac{\regabs{y_i}}{\mean{R_i}}} 
 \right) 
 \,, \quad y\in \Z^d\,.
\]
\end{proposition}

The proof of Proposition \ref{prop:propagator_multi_d} will be given in
Appendix \ref{sec:propagator}. The way we will use Proposition \ref{prop:propagator_multi_d} is
via the following Corollary \ref{prop:propagator_multi_d} that implies the bound 
\eqref{eq:maxnormboundedbyW} discussed after Theorem \ref{thm:Wwell_posedness}. 

\begin{corollary}
    \label{cor:initialdata}
 Suppose $L\ge 2$ and $W_1,W_2:\duallattice\to\C$. 
 For $R,u\in \R^d$ and $k_1,k_2\in \duallattice$ and $y\in\lattice$, define
  \[
  B(R,u;k_1,k_2,y) \coloneqq \int_{\duallattice} \!\rmd k\, \rme^{\ci \phi(k;R,u)}
  W_1(k+k_1) W_2(k+k_2)\rme^{\ci2\pi k\cdot y}\,.
  \]
  
  There is a constant $C$, which is independent of $L$, such that for all
 $R,u,k_1,k_2,y$, we have
 \[
     |B(R,u;k_1,k_2,y)| 
     \le
     C 
     \sobolevnorm{W_1}{1/3}
     \sobolevnorm{W_2}{1/3}
     \prod_{j\in J_R} \regabs{y_j}^{-\frac{1}{3}}
   \min
 \left( 
 1, \sqrt{\frac{\regabs{y_j}}{\mean{R_j}}} 
 \right)
 \]
 where $J_R\coloneqq\{j \ \colon |R_j|\leq\frac{L}{8}\}$.
 In particular, 
 \[
  \left|\int_{\duallattice} \!\rmd k\, \rme^{\ci \phi(k;R,u)} W_1(k+k_1)\rme^{\ci2\pi k\cdot y}\right| 
  \le C
     \sobolevnorm{W_1}{1/3}
   \prod_{j\in J_R} \regabs{y_j}^{-\frac{1}{3}}
   \min
 \left( 
 1, \sqrt{\frac{\regabs{y_j}}{\mean{R_j}}} 
 \right)
 \]
\end{corollary}
\begin{proof}
    Let $g_i:\lattice\to\C$ be the inverse Fourier transform of $W_i$, i.e. $g_i=\IFT{W_i}$ for 
    $i=1,2$.

    If $g_2(x)=\cf{x=0}$, then $W_2(k)=1$ and $ \sobolevnorm{W_2}{1/3} =1$. Thus,
    the second bound is indeed a consequence of the first one.

Without loss of generality, we may assume $J_R=\{1,2,\dots,|J_R|\}$. Let $P\colon\R^d\to\R^{|J_R|}$ be 
the projection $Px=(x_1,x_2,\dots,x_{|J_R|})$.
Inserting the Fourier transforms explicitly, we find
\begin{align*}
 &  B(R,u;k_1,k_2,y) = \sum_{x_1,x_2\in \lattice}
   g_1(x_1) g_2(x_2) 
   \int_{\duallattice} \!\rmd k\, \rme^{\ci \phi(k;R,u)-\ci 2\pi (k+k_1)\cdot
x_1-\ci 2\pi (k+k_2)\cdot x_2+\ci2\pi k\cdot y}
\\ 
 & =\sum_{x_1,x_2\in \lattice}
   g_1(x_1) g_2(x_2) \rme^{-\ci 2\pi (k_1\cdot x_1+k_2 \cdot x_2)}
   Q(Py-Px_1-Px_2;PR,Pu,L)
   \\
   &\quad\times \prod_{j=|J_R|+1}^d \int_{\oneduallattice}\rmd k \rme^{\ci
   R_j\cos(2\pi(k+u)_j)-\ci2\pi(k)_i(x_1)_j-\ci2\pi(k)_j(x_2)_j+\ci2\pi (k)_j(y)_j}.
\end{align*}
From the definition of $J_R$ and $\int_{\oneduallattice}\rmd k=1$, together with
Proposition \ref{prop:propagator_multi_d}, we obtain
\begin{align*}
 &
 | B(R,u;k_1,k_2,y)|
 \\ & \quad
 \le 
 \sum_{x_1,x_2\in \lattice}
   |g_1(x_1)| |g_2(x_2)| |I(Py-Px_1-Px_2;PR,Pu)|
   + \sum_{x_1,x_2\in \lattice}
   |g_1(x_1)| |g_2(x_2)|C 
 \rme^{-\delta_0\frac{1}{2}L}
 \\ & \quad
 \le 
 \sum_{x_1,x_2\in \lattice}
  C
   |g_1(x_1)| |g_2(x_2)|  
  \prod_{j\in J_R}
    \regabs{z_j}^{-\frac{1}{3}}
   \min
 \left( 
 1, \sqrt{\frac{\regabs{z_j}}{\mean{R_j}}} 
 \right) 
   + \sobolevnorm{W_1}{1/3}\sobolevnorm{W_2}{1/3}   C 
 \rme^{-\delta_0\frac{1}{2}L}
 \,,
\end{align*}
where $z = y-x_1-x_2$.

Lemma \ref{lem:regabs} implies 
$\regabs{z_i}^{-\frac{1}{3}}\leq 3^{\frac{1}{3}}
\regabs{(x_1)_i}^{\frac{1}{3}}\regabs{(x_2)_i}^{\frac{1}{3}}\regabs{y_i}^{-\frac{1}{3}}$ and 
$\regabs{z_i}^{\frac{1}{6}}\leq 3^\frac{1}{6}
\regabs{(x_1)_i}^{\frac{1}{6}}\regabs{(x_2)_i}^{\frac{1}{6}}\regabs{y_i}^{\frac{1}{6}}$. Using these estimates 
together with $\regabs{\cdot}^{\frac{1}{6}}\leq \regabs{\cdot}^{\frac{1}{3}}$ implies
\begin{align*}
\prod_{j\in J_R}
    \regabs{z_j}^{-\frac{1}{3}}
   \min
 \left( 
 1, \sqrt{\frac{\regabs{z_j}}{\mean{R_j}}} 
 \right) 
 \leq 
 \left( 
\prod_{j\in J_R}
\regabs{(x_1)_j}^{\frac{1}{3}}
\regabs{(x_2)_j}^{\frac{1}{3}}
 \right) 
\prod_{j\in J_R}
    \regabs{y_j}^{-\frac{1}{3}}
   \min
 \left( 
 1, \sqrt{\frac{\regabs{y_j}}{\mean{R_j}}} 
 \right) 
\end{align*}

To handle the second term, we notice that $\prod_{j\in J_R}
\mean{R_j}^{\frac{1}{2}}\rme^{-\delta_0\frac{1}{2}L}\le L^{\frac{d}{2}}
\rme^{-\delta_0\frac{1}{2}L}\le C(\delta_0,d)$ for all $L$. 
We conclude that
the stated bound holds.

\end{proof}

\section{Collision control map}
\label{sec:F}

As discussed in {Section \ref{sec:results}}, the main motivations for the 
collision control map $F[W]$ are that it has a closed evolution equation, and 
we can use it to rewrite the cubic collision operator $\cbos$ into a bilinear one 
$\cbiq$ with $F[W]$ and $W$ as its inputs. In the following subsection \ref{sec:rewriteC},
we will define the bilinear operator $\cbiq$ and prove that it can be used to rewrite 
the collision operator $\cbos$. After that, in Subsection \ref{sec:fixpoint_operator} we will
define a bilinear operator $\mathcal{I}_\tau$ and show that it indeed gives an 
evolution equation for the collision control map $F[W]$.

Recall from the definition \eqref{eq:F} of the collision control map $F[W]$ that 
\begin{align}
    F[W]_{t_1,t_2}^{(n)}(R,   k',u,\sigma,x) 
    \coloneqq 
    \int_{\duallattice} \rmd k_0 
    \rme^{\ci \phi(k_0, R, u)} 
    W_{t_1}(\sigma k_0) ^{n-1}W_{t_2}(k'+k_0)
    \rme^{\ci2\pi k_0\cdot x},
\end{align}
for $W\in C([0,\tmax]\times,\C)$, $n\in\{1,2\}$, $(t_1,t_2)\in[0,\tmax]^2$ and $(R,   k',u,\sigma,x)\in\xl$. Recall 
the definition of the phase function $\phi$ from \eqref{eq:phi}.

Next, we will go through the linear combinations of the phase terms that arise in the analysis of the collision control map. We will also introduce simplifying shorthand notation.
We define two functions $z_0,z_1:\R\to\C$,
\begin{align}
    z_0(k)\coloneqq  
    1-\rme^{-\ci 2\pi k}, \qquad z_1(k) \coloneqq z_0(k+1/2) = 1 {+} \rme^{-\ci 2 \pi k}.
\end{align}
Note that $|z_0(k)| = 2|\sin(\pi k)|$ and $z_0(k+1/2)=z_0(k-1/2)$.
Recall from \eqref{eq:omega} that our dispersion relation $\omega$ corresponds to a nearest neighbour hopping potential.
Therefore, we have by Lemma \ref{lem:cossum} that   
\begin{align}
    &\phi(k_0,R,u)- s \left( \omega(k_0+k') -\omega(k_0+k'-k) \right) 
    = \phi(k_0,\rs{0}(R,u,k',k), \widetilde{u}'_0).
\end{align}
For all $R\in\R^d$, $u\in\T^d$ and $k_0,k,k'\in\lattice$.
We have defined 
\begin{align}
    \label{eq:R_tilde0}
    \rs{0}(R,u,k',k)_\ell
    &\coloneqq \left| R_\ell \rme^{\ci 2\pi (u-k')_\ell} + s \left( 1 -
    \rme^{-\ci 2\pi k_\ell} \right)  \right|
    =\left|  R_\ell \rme^{\ci 2\pi (u-k')_\ell} + s z_0(k_\ell) \right| , 
    \\
    \label{eq:u_tilde_prime0}
    \widetilde{u}'_0(s,R,u,k',k)_\ell
    &\coloneqq \frac{1}{2\pi} \arg\left( R_\ell \rme^{\ci 2\pi u_\ell} + s \left( \rme^{\ci 2\pi (k')_\ell} -
     \rme^{\ci 2\pi(k'-k)_\ell} \right)  \right), 
\end{align}
and used the shorthand notation $ \widetilde{u}'_0=\widetilde{u}'_0(s,R,u,k',k)$.
Note that $\widetilde{u}'_0$ is constant in $k_0$. The explicit
form of $\widetilde{u}'_0$ will not be important for us, the reason being that the weight $\Phi^d$ in the norm $\phinorm{\cdot}$ does 
not depend on $u$ (see \eqref{eq:norm_phi}). Therefore, we will just write $\widetilde{u}'_0$.

We also define the shorthand notations
\begin{align}
    \label{eq:R_tilde1}
    \rs{1}(R,u,k',k)
    &\coloneqq \rs{0}(R,u,k',k-\bar{1}/2)=
    \left( 
    \left|  R_\ell \rme^{\ci 2\pi (u-k')_\ell} + s z_1(k_\ell) \right|
    \right)_{\ell=1}^d,
    \\
    \label{eq:u_tilde_prime1}
    \widetilde{u}'_1
    =
    \widetilde{u}'_1(s,R,u,k',k)
    &\coloneqq
    \widetilde{u}'_0(s,R,u,k',k-\bar{1}/2),
\end{align}
where $\bar{1}=(1,1,\dots,1)$. Note that $\rs{\xi'}(R,u,k',k)=\rs{0}(R,u,k',k+\xi'\bar{1}/2)$
for both $\xi'=0$ and $\xi'=1$. We also use the shorthand notation 
\begin{align}
    \rs{\xi'}(k)_\ell\coloneqq \rs{\xi'}(0,0,0,k)_\ell = |s||z_{\xi'}(k_\ell)|.
    \label{eq:R_shorthand_tilde}
\end{align}

\subsection{Rewriting the collision operator $\cbos$}
\label{sec:rewriteC}

 In the following Lemma \ref{lem:colG}, we prove that the following bilinear operator $\cbiq$, defined below, can be used together with the collision control map 
 $F[W]$ to rewrite the collision operator $\cbos$. The definition of 
$\cbiq$ is chosen so that the Lemma \ref{lem:colG} holds.

For each $\tau_0\ge0$ we define the bilinear operator $$\cbiq(\tau_0): \Cmax\times
C(\duallattice,\R)\to C(\duallattice,\R),$$  
by setting
\begin{align}
    \nonumber
    \cbiq(\tau_0)[G,W](k_0)
    &\coloneqq
    \sum_{i=1}^8 \sigma_i^{(q)} 
    \int_{-\lambda^{-2}\tau_0}^{\lambda^{-2}\tau_0}\rmd s
    \int_{\duallattice} \rmd k
           \sum_{y,y'\in\lattice}^{} 
           v_i(k,y)w_i(y')\rme^{\ci \widetilde{K}_i
           -\ci s (\omega(k_0)-\omega(k_0-k+\xi_i/2))}
           \\
    &\times
    G^{(\widetilde{n}_i)}(\rs{\xi_i}(k),k,\widetilde{u}_{\xi_i},\sigma_i^{(2)},(y+y')\bmod\lattice)
    W(k_0-\ind(i=4,5)k),
    \label{eq:cbosg}
\end{align}
where the new symbols are defined as follows. First,
\begin{alignat}{3}
    \sigma^{(q)}&\coloneqq (-\theta,-\theta,-q,q,\theta,\theta,\theta,+), && \quad \sigma^{(2)}&&\coloneqq (+,+,+,+,+,-,-,-),\\
    \widetilde{n}&\coloneqq (2,2,1,1,2,2,2,2), && \quad \xi&&\coloneqq (0,0,0,0,0,1,1,1),
\end{alignat}
\begin{align*}
    \widetilde{K}_i &\coloneqq 
    \begin{cases}
        2\pi(k-k_0)\cdot y\,, & i\leq6 \\
        2\pi k_0\cdot y\,, & i=7 \\
        2\pi k_0\cdot(y-y')+2\pi k\cdot y'\,, & i=8 
    \end{cases}
    ,\qquad
    v_i(k,y) \coloneqq 
    \begin{cases}
        \mathcal{V}(k,y)\,, & i\leq5 \\
        V*V(y)\,, & i=6,7 \\
        2V(y)\,, & i=8 
    \end{cases},
\end{align*}
 and $w_i(y')\coloneqq \ind(y'=0)$ for $i\leq 7$ and $w_8(y')\coloneqq V(y')$. Finally, using the notation \(V*V(y) = \sum_{x}V(x)V(y-x)\) for the convolution of \(V\) with itself, we set
\begin{align}
    \label{eq:mathcalV}
    \mathcal{V}(k,y)&\coloneqq  \ind(y=0)\FT{V}(k)^2+2\theta\FT{V}(k)V(y)+V*V(y) \, ,
    \\
    \label{eq:tildeu}
    (\widetilde{u}_{\xi_i})_\ell&\coloneqq
    \frac{1}{2\pi} \arg\left(  s \left( e^{\ci 2\pi \xi_i/2} -
     e^{\ci 2\pi k_\ell} \right)  \right).
\end{align}

One should keep in mind that $v$ and $w$ depend on the 
potential $V$. Moreover, we have the following Lemma \ref{lem:ineqVconst}.
 \begin{lemma}
     \label{lem:ineqVconst}
     For every $p\ge0$ and $i=1,2,\dots,8$ there holds 
     \begin{align}
         \label{eq:ineqVconst}
         \sum_{y,y'\in\lattice}^{} 
         \left( 
             \prod_{j=1}^d 
             \regabs{y_j}
             \regabs{y'_j}
         \right)^{p}
         \sup_{k\in\duallattice}|v_i(k,y)||w_i(y')|
         &\leq 
         (2^{pd}+3)\sobolevnorm{\FT{V}}{p}^2
         .
     \end{align}
 \end{lemma}
 \begin{proof}
      Since $\sobolevnorm{\FT{V*V}}{p} \leq2^{pd} \sobolevnorm{\FT{V}}{p} ^2$,
      we have that for $i=1,\dots,7$ the left hand side of \eqref{eq:ineqVconst} is bounded from above by 
    \[
        \subnorm{\FT{V}}{\infty}^2
        +
        2\subnorm{\FT{V}}{\infty}
        \sobolevnorm{\FT{V}}{p}
        +
        2^{pd} \sobolevnorm{\FT{V}}{p} ^2.
    \]
    For $i=8$, the left hand side is bounded by $2\sobolevnorm{\FT{V}}{p} ^2$.
    Then $\subnorm{\FT{V}}{\infty}\leq\sobolevnorm{\FT{V}}{p}$ implies the result.
 \end{proof}
 
\begin{lemma}
    \label{lem:colG}
    Let $\tmax>0$. There holds
    \begin{align}
        \cbiq(\tau_0)[F[W]_{t,t},W_t](k_0)
        &=
    \cbos(W_t,\tau_0)(k_0)
    \end{align}
    for all $k_0\in\duallattice$, $\tau_0\ge0$, $t\in[0,\tmax]$, $W\in
    C([0,\tmax]\times\duallattice,\R)$.
\end{lemma}

\begin{proof}
   This is a straightforward calculation. We start from the decomposition of the operator
   $\cbos(W_t,\tau)=\theta\cnls{\theta}(W_t,\tau)+q\ccl{\theta}(W_t,\tau)$.

   The following computation is for the first term in the $\cnls{\theta}(W_t,\tau)$ operator,
   \begin{align*}
    &2\pi \int_{(\duallattice)^3}\rmd k_{123}\delta_L(k_0+k_1-k_2-k_3)
    \delta_{\lambda,\tau}(\Omega(k_0,k))
    W(k_1)W(k_2)W(k_3)
    \\
    &\qquad\times\left( \FT{V}(k_1-k_2)+\theta\FT{V}(k_1-k_3) \right)^2
    \\
    &=
       \int_{-\lambda^{-2}\tau}^{\lambda^{-2}\tau }\rmd s
       \int_{(\duallattice)^2}\rmd k_{13}
       \rme^{-\ci s(\omega(k_0)+\omega(k_1)-\omega(k_0+k_1-k_3)-\omega(k_3))}
       W(k_1)W(k_0+k_1-k_3)W(k_3)
       \\
        &\qquad\times
    \left( \FT{V}(k_3-k_0)+\theta\FT{V}(k_1-k_3) \right)^2
    \\
        &=
       \int_{-\lambda^{-2}\tau}^{\lambda^{-2}\tau }\rmd s
       \int_{\duallattice}\rmd k_{3}
       \rme^{-\ci s(\omega(k_0)-\omega(k_0-k_3))}
        W(k_0-k_3)
       \\
        &\qquad\times
        \left( 
       \int_{\duallattice}\rmd k_{1}
       \rme^{-\ci s(\omega(k_1)-\omega(k_1+k_3))}
       W(k_1)W(k_1+k_3)
    \left( \FT{V}(k_3+k_1-k_0)+\theta\FT{V}(k_3) \right)^2
        \right) 
    \\
        &=
       \int_{-\lambda^{-2}\tau}^{\lambda^{-2}\tau }\rmd s
       \int_{\duallattice}\rmd k_{3}
       \rme^{-\ci s(\omega(k_0)-\omega(k_0-k_3))}
        W(k_0-k_3)
       \\
        &\qquad\times
        \sum_{y\in\lattice}^{} 
        \left( 
       \int_{\duallattice}\rmd k_{1}
       \rme^{-\ci s(\omega(k_1)-\omega(k_1+k_3))}
       W(k_1)W(k_1+k_3)
       \rme^{\ci2\pi k_1\cdot y}
        \right) 
        \rme^{\ci2\pi (k_3-k_0)\cdot y}\mathcal{V}(k_3,y).
   \end{align*}
   Here we integrated $\delta_L$ over $k_2$, applied the change of variables
   $k_3\mapsto k_3+k_1$, and used Parseval's Theorem \ref{thm:parseval} together with the identity 
   \[
       \left(
       \int_{\duallattice}\rmd k_{1}
    \left( \FT{V}(k_3+k_1-k_0)+\theta\FT{V}(k_3) \right)^2
    \rme^{\ci2\pi k_1\cdot y}
       \right) ^*
       = 
        \rme^{\ci2\pi (k_3-k_0)\cdot y}\mathcal{V}(k_3,y),
   \]
   with $\mathcal{V}$ defined in \eqref{eq:mathcalV}.
   After renaming $k_3$ as $k$, this gives the $i=5$ term in \eqref{eq:cbosg}. 
   We get the $i=4$ term from the $W(k_2)W(k_3)$ gain term in $q\cbos(W_t,\tau)$
   by replacing $W(k_1)$ by $1$ in the above computation.

   Note that the two loss terms in $\cnls{\theta}(W,\tau)$ are equal, this can be seen
   from the definition \eqref{eq:cnls}
   by swapping the symbols $k_2$ and $k_3$ in one of the terms. Consider 
   the term with $W(k_0)W(k_1)W(k_3)$. We integrate the $\delta_L$ over $k_2$ 
    and do a change of variables
   $k_3\mapsto k_3+k_1$ and use Parseval's Theorem similarly as above. 
   These two loss terms produce the $i=1$ and $i=2$ terms in
   \eqref{eq:cbosg}. The $i=3$ term is obtained by doing the
   exact same for the loss term (the term with $W(k_0)W(k_1)$) in $q\ccl{\theta}(W_t,\tau)$.

   The remaining part, i.e., terms $i=6,7,8$ in $\eqref{eq:cbosg}$, are obtained
   from the second gain term involving $W(k_0)W(k_2)W(k_3)$ in $\cnls{\theta}(W_t,\tau)$. This
   is done by integrating $\delta_L$
   over $k_1$. After this, we apply change of variables twice: first $k_3\mapsto
   k_3-k_2$, and then $k_2\mapsto -k_2$. Using the identities $\omega(-k)=\omega(k)$ and
   $\omega(k\pm1/2)=2c_0-\omega(k)$  we obtain the following form for this term
   \begin{align*}
       &\int_{-\lambda^{-2}\tau}^{\lambda^{-2}\tau }\rmd s
       \int_{\duallattice}\rmd k_{3}
       \rme^{-\ci s(\omega(k_0)-\omega(k_0-k_3+1/2))}
        W(k_0)
       \\
        &\qquad\times
        \left( 
       \int_{\duallattice}\rmd k_{2}
       \rme^{-\ci s(\omega(k_2+1/2)-\omega(k_2+k_3))}
       W(-k_2)W(k_2+k_3)
    \left( \FT{V}(k_3+k_2-k_0)+\theta\FT{V}(k_2+k_0) \right)^2
     \right).
   \end{align*}
   Expanding the square inside the integral gives us three terms. For both terms
   $\FT{V}(k_3+k_2-k_0)^2$ and $\FT{V}(k_2+k_0)^2$, it is sufficient to use
   Parseval's Theorem
   once as above. The resulting terms give the terms $i=6,7$  in \eqref{eq:cbosg}. For
   the cross term $2\theta \FT{V}(k_3+k_2-k_0)\FT{V}(k_2+k_0)$ we have to apply Parseval's Theorem
   twice to get the final term $i=8$. 
\end{proof}

\subsection{Time evolution of the collision control map}
\label{sec:fixpoint_operator}

In this subsection, we will define a bilinear operator $\mathcal{I}_\tau$ and
prove that it gives the time evolution of the collision control map $F[W]$.
Namely, if $W$ is a solution to the evolution equation \eqref{eq:Wevo} with
$\tau$, then $F[W]=\mathcal{I}_\tau[F[W],F[W]]$. This will be proved in 
Lemma \ref{lem:Fevo} below.

Given $\tau$ as in \eqref{eq:taudiff} or
\eqref{eq:tauconst}, the bilinear operator
$$\mathcal{I}_\tau:\CT{\tmax}^2\to\CT{\tmax}$$ is defined by setting
\begin{align}
    \nonumber
    \mathcal{I}_\tau[G,\mathcal{G}]_{t_1,t_2}^{(n)}(X) \coloneqq 
    \mathcal{G}_{0,0}^{(n)}(X)
    +
    \frac{1}{2}
    &
    \left( 
        \mathcal{J}_\tau[G,\mathcal{G}]_{0,t_2}^{(n)}(X)
        +
        \mathcal{J}_\tau[G,\mathcal{G}]_{t_1,t_2}^{(n)}(X)
    \right) 
    \\
    \label{eq:I}
                                      +
    \ind(n=2)
    \frac{{e^{-\ci2\pi k'\cdot x}}}{2}
    &
    \left( 
        \mathcal{J}_\tau[G,\mathcal{G}]_{0,t_1}^{(2)}(\widetilde{X})
        +
        \mathcal{J}_\tau[G,\mathcal{G}]_{t_2,t_1}^{(2)}(\widetilde{X})
    \right),
\end{align}
for $(t_1,t_2)\in[0,\tmax]^2$,  $X=(R,k',u,\sigma,x)\in\xl$,
$\widetilde{X}=(R,-\sigma k',\sigma(u-k'),\sigma,\sigma x)$ and $n\in\{1,2\}$.
Here, the other bilinear operator $$\mathcal{J}_\tau\colon \CT{\tmax}^2\to\CT{\tmax}$$ is defined
by setting 
\begin{align}
    \nonumber
    \mathcal{J}_\tau[G,\mathcal{G}]_{t_1,t_2}^{(n)}(X) &\coloneqq \sum_{i=1}^8
    \sigma^{(q)}_i
    \int_0^{t_2} \rmd t
    \int_{-\lambda^{-2}\tau(t_2,t)}^{\lambda^{-2}\tau(t_2,t)}
    \rmd s 
    \int_{\duallattice} \rmd k 
    \sum_{y,y'\in\lattice}
    v_i(k,y)w_i(y')
    \rme^{\ci K_i}
        \\
        \label{eq:J}
         &
         \times 
     G^{(\widetilde{n}_i)}_{t,t}(\Jpoint_s^i(X,k,y,y'))
     \mathcal{G}^{(n)}_{t_1,t}(\Jppoint_s^i(X,k,y,y'))
\end{align}
for $(t_1,t_2)\in[0,\tmax]^2$, $X=(R,k',u,\sigma,x)\in\xl$ and $n\in\{1,2\}$. The functions \(\Psi\) and \(\widetilde{\Psi}\) in the definition are given by
\begin{align}
    \Jpoint_s^i(X,k,y,y')&\coloneqq  (\rs{\xi_i}(k),k,\widetilde{u}_{\xi_i},\sigma^{(2)}_i,(y+y')\bmod\lattice)
    \\
    \Jppoint^i_s(X,k,y,y')&\coloneqq(\rs{\xi_i}(R,
    u,k',k),\widetilde{k}'_i,\widetilde{u}_{\xi_i}',\sigma, (x+\sigma^{(3)}_iy-y')\bmod\lattice).
\end{align}
Here $\widetilde{k}_i'\coloneqq k'-k$ for $i=4,5$ and $\widetilde{k}_i'\coloneqq k'$ otherwise, and $\sigma^{(3)}\coloneqq (-,-,-,-,-,-,+,+)$, and
\begin{align}
    K_i\coloneqq 
    \begin{cases}
        2\pi(k-k')\cdot y\,, & i\leq 6\\
        2\pi k'\cdot y\,, & i=7 \\
        2\pi k'\cdot(y-y')+2\pi k\cdot y'\,, &i=8
    \end{cases}.
\end{align}
For the other symbols in \eqref{eq:J}, we refer the reader to \eqref{eq:R_tilde0}-\eqref{eq:u_tilde_prime1}
and  
\eqref{eq:cbosg}.

As discussed in
Section \ref{sec:keyideas}, 
for a given 
$G\in \CT{\tmax}$ we will be interested 
in the following equation for $W\in C([0,\tmax]\times\duallattice,\C)$
\begin{align}
    \label{eq:WGevo}
    W_T(k_0)=W_0(k_0)
    +
    \int_0^T
    \rmd t
    \cbiq(\tau(T,t))[G_{t,t},W_t](k_0),
\end{align}
where the bilinear operator $\cbiq$ was defined in \eqref{eq:cbosg}. Note that 
by Lemma \ref{lem:colG}, if $G=F[W]$, then the evolution equations
\eqref{eq:WGevo} and \eqref{eq:Wevo} are equivalent.  We will be especially
interested in the equation \eqref{eq:WGevo} when $G$ solves 
the fixed point equation
\begin{align}
    \label{eq:Gevo}
    G &= \mathcal{I}_\tau[G,G].
\end{align}

In the following Lemma \ref{lem:Fevo} we will prove the main motivation of the 
definition given for the bilinear operator $\mathcal{I}_\tau$, namely, if 
$W$ satisfies \eqref{eq:WGevo} with $G$, then $F[W]=\mathcal{I}_\tau[G,F[W]]$.
Therefore, $\mathcal{I}_\tau$ gives the time evolution of $F[W]$.

\begin{lemma}
    \label{lem:JFconnection}
    Let $\tmax>0$ and $(W,G)\in C([0,\tmax]\times\duallattice,\R)\times \CT{\tmax}$.
    Let $\tau$ be as in \eqref{eq:taudiff} or \eqref{eq:tauconst}.
    Then
    \begin{align}
    		\label{eq:J-equa}
        \nonumber
            \mathcal{J}_\tau[G,F[W]]_{t_1,t_2}^{(n)}(R,k',u,\sigma,x)
            =
            \int_{\duallattice}&\rmd k_0 
            \rme^{\ci\phi(k_0,R,u)}
            W_{t_1}(\sigma k_0)^{n-1}
            \rme^{\ci2\pi k_0\cdot x}\nonumber\\
            &\times\left(
            \int_0^{t_2}\rmd t 
            \cbiq(\tau(t_2,t))[G_{t,t},W_t](k_0+k')\right)
    \end{align}
    for all $(R,k',u,\sigma,x)\in\xl$, $(t_1,t_2)\in[0,\tmax]^2$, and $n\in\{0,1\}$. 
\end{lemma}
\begin{proof}
    This is a straightforward computation starting from the right-hand side.
    For the following computations we use a shorthand notation $$
    \int\rmd t \rmd s \rmd k
    \coloneqq  \int_0^{t_2} \rmd t
        \int_{-\lambda^{-2}\tau(t_2,t)}^{\lambda^{-2}\tau(t_2,t)}\rmd s
        \int_{\duallattice} \rmd k 
        $$
    Using the definition \eqref{eq:cbosg} of $\cbiq$, we can write the RHS of \eqref{eq:J-equa}
        as a sum of terms $A_1,A_2,A_3$, where $A_1$ contains the terms $i=1,\dots 5$
        of the sum in \eqref{eq:cbosg}, $A_2$ the terms $i=6,7$, and $A_3$ the term 
        $i=8$. We compute as follows
    \begin{align}
    \nonumber
    A_1&=
    \int\rmd t \rmd s \rmd k
    \sum_{y\in\lattice}
        \sum_{i=1}^{5} 
        \sigma^{(q)}_i
    \mathcal{V}(k,y)
    G^{(\widetilde{n}_i)}_{t,t}(\rs{0}(k),k,\widetilde{u}_{0},+,y)
        \int_{\duallattice} \rmd k_0
            \rme^{\ci\phi(k_0,R,u)}
            W_{t_1}(\sigma k_0)^{n-1}
            \rme^{\ci2\pi k_0\cdot x}
    \\
    \nonumber
                            &\qquad\times
    e^{-\ci s(\omega(k_0+k')-\omega(k_0+k'-k))}
            W_t(k_0+k'-\ind(i=4,5)k)
            \rme^{\ci2\pi (k-k_0-k')\cdot y}
            \\
            \nonumber
                            &=
    \int\rmd t \rmd s \rmd k
    \sum_{y\in\lattice}
        \sum_{i=1}^{5} 
        \sigma^{(q)}_i
    \mathcal{V}(k,y)
    \rme^{\ci2\pi (k-k')\cdot y}
    G^{(\widetilde{n}_i)}_{t,t}(\rs{0}(k),k,\widetilde{u}_{0},+,y)
    \\
    \nonumber
                            &\qquad\times
        F[W]_{t_1,t}^{(n)}(\rs{0}(R,u,k',k),k'-\ind(i=4,5)k,
                            {\widetilde{u}'_0},\sigma,(x-y)\bmod\lattice).
    \end{align}
    \begin{align*}
        A_2&=\theta
    \int\rmd t \rmd s \rmd k
    \sum_{y\in\lattice}
         V*V(y)
         G^{(2)}_{t,t}(\rs{1}(k),k,\widetilde{u}_{1},-,y)
        \int_{\duallattice} \rmd k_0
            \rme^{\ci\phi(k_0,R,u)}
            W_{t_1}(\sigma k_0)^{n-1}
            \rme^{\ci2\pi k_0\cdot x}
    \\
    \nonumber
                            &\qquad\times
        e^{-\ci s(\omega(k_0+k')-\omega(k_0+k'-k+1/2))}W_t(k_0+k')
                               \left( \rme^{\ci2\pi(k-k_0-k')\cdot
                               y}+\rme^{\ci2\pi (k_0+k')\cdot y)} \right)  
                               \\
                            &=\theta\int\rmd t \rmd s \rmd k
    \sum_{y\in\lattice}
         V*V(y)
         G^{(2)}_{t,t}(\rs{1}(k),k,\widetilde{u}_{1},-,y)
    \\
    \nonumber
                            &\qquad\times
                            \Big( 
                                \rme^{\ci2\pi (k-k')\cdot y}
                                F[W]_{t_1,t}^{(n)}(\rs{1}(R,u,k',k),k',{\widetilde{u}'_1},\sigma,(x-y)\bmod\lattice)
                                \\
                            &\qquad\quad
                            +
                                \rme^{\ci2\pi k'\cdot y}
                                F[W]_{t_1,t}^{(n)}(\rs{1}(R,u,k',k),k',{\widetilde{u}'_1},\sigma,(x+y)\bmod\lattice)
                            \Big). 
    \end{align*}
    \begin{align*}
        A_3&=
    \int\rmd t \rmd s \rmd k
    \sum_{y,y'\in\lattice}
                            2V(y)V(y')
                            G^{(2)}_{t,t}(\rs{1}(k),k,\widetilde{u}_{1},-,y+y')
\int_{\duallattice} \rmd k_0
            \rme^{\ci\phi(k_0,R,u)}
            W_{t_1}(\sigma k_0)^{n-1}
    \\
    \nonumber
                            &\qquad\times
            \rme^{\ci2\pi k_0\cdot x}
        e^{-\ci s(\omega(k_0+k')-\omega(k_0+k'-k+1/2))}W_t(k_0+k')
         \rme^{\ci2\pi (k_0+k')\cdot (y-y')+\ci2\pi k\cdot y'}
    \\
                            &=
         \int\rmd t \rmd s \rmd k
    \sum_{y,y'\in\lattice}
                            2V(y)V(y')
                            \rme^{\ci2\pi k'\cdot (y-y')+\ci2\pi k\cdot y'}
                            G^{(2)}_{t,t}(\rs{1}(k),k,\widetilde{u}_{1},-,(y+y')\bmod\lattice)
    \\
    \nonumber
                            &\qquad\times
                            F[W]_{t_1,t}^{(n)}(\rs{1}(R,u,k',k),k',{\widetilde{u}'_1},\sigma,(x+y-y')\bmod\lattice).
    \end{align*}
    Thus, $A_1+A_2+A_3$ is equal to the LHS of \eqref{eq:J-equa}.
\end{proof}

\begin{lemma}
    \label{lem:Fevo}
    Let $\tmax>0$.
    Let $\tau$ be as in \eqref{eq:taudiff} or \eqref{eq:tauconst}.
    If $W\in C([0,\tmax]\times\duallattice,\R)$ satisfies the evolution equation \eqref{eq:WGevo} with
    $\tau$ and $G\in \CT{\tmax}$, then there holds
    \begin{align}
        \label{eq:Fevo}
    &F[W] =
    \mathcal{I}_\tau[G,F[W]]
\end{align}
In particular, if $W$ satisfies the evolution equation \eqref{eq:Wevo} with $\tau$, then
\eqref{eq:Fevo} holds for $G=F[W]$.
\end{lemma}

\begin{proof}
    We will be using the fact that $(\zeta_1+\zeta_2)^{n-1}=\zeta_1^{n-1}+\ind(n=2)\zeta_2$ for all $\zeta_1,\zeta_2\in \C$ and $n\in\{1,2\}$.

    Let $(R,k',u,\sigma)\in\xl$, let $t_1,t_2\in[0,\tmax]$, and $n\in\{1,2\}$.
    Writing $F[W]_{t_1,t_2}^{(n)}=\frac{1}{2}F[W]_{t_1,t_2}^{(n)}+\frac{1}{2}F[W]_{t_1,t_2}^{(n)}$ and then 
    applying the evolution equation \eqref{eq:WGevo} for $W_{t_1}^{n-1}$
    in the first term and for $W_{t_2}$ in the second term, we obtain
\begin{align*}
    \nonumber 
    F[W]_{t_1,t_2}&^{(n)}(R,k'u,\sigma,x) = 
                                        \frac{1}{2}
    \int_{\duallattice}\rmd k_0 
    \rme^{\ci\phi(k_0,R,u)}
    W_{t_1}(\sigma k_0)^{n-1}
        W_0(k'+k_0)
    \rme^{\ci2\pi k_0\cdot x}
    \\
    +\frac{1}{2}&
    \int_{\duallattice}\rmd k_0 
    \rme^{\ci\phi(k_0,R,u)}
    W_{t_1}(\sigma k_0)^{n-1}
    \rme^{\ci2\pi k_0\cdot x}
    \int_0^{t_2}\rmd t 
    \cbiq(\tau(t_2,t))[G_{t,t},W_t](k_0+k')
    \\
    +\frac{1}{2}&
    \int_{\duallattice}\rmd k_0 
    \rme^{\ci\phi(k_0,R,u)}
    W_0(\sigma k_0)^{n-1}
        W_{t_2}(k'+k_0)
    \rme^{\ci2\pi k_0\cdot x}
    \\
                                        +
                                        \ind(n=2)
                                        \frac{1}{2}&
    \int_{\duallattice}\rmd k_0 
    \rme^{\ci\phi(k_0,R,u)}
        W_{t_2}(k'+k_0)
    \rme^{\ci2\pi k_0\cdot x}
    \int_0^{t_1}\rmd t 
    \cbiq(\tau(t_1,t))[G_{t,t},W_t](\sigma k_0).
\end{align*}
Using \eqref{eq:WGevo} again for $W_{t_1}^{n-1}$ in the first term and for $W_{t_2}$ in the third term, we arrive at
\begin{align}
    \nonumber 
    &F[W]_{t_1,t_2}^{(n)}(R,k'u,\sigma,x) = 
    F[W]_{0,0}^{(n)} (R,k'u,\sigma,x)
    \\
    \nonumber 
                                        &+
                                        \frac{1}{2}
    \int_{\duallattice}\rmd k_0 
    \rme^{\ci\phi(k_0,R,u)}
    \left(
    W_0(\sigma k_0)^{n-1}
    +
    W_{t_1}(\sigma k_0)^{n-1}
    \right)
    \rme^{\ci2\pi k_0\cdot x}
    \\
    \nonumber
    &\qquad\times
    \int_0^{t_2}\rmd t 
    \cbiq(\tau(t_2,t))[G_{t,t},W_t](k_0+k')
    \\
    \nonumber
                                        &+
                                        \ind(n=2)
                                        \frac{1}{2}
    \int_{\duallattice}\rmd k_0 
    \rme^{\ci\phi(k_0,R,u)}
    \left( 
        W_0(k'+k_0)
        +
        W_{t_2}(k'+k_0)
    \right) 
    \rme^{\ci2\pi k_0\cdot x}
    \\
    \nonumber
    &\qquad\times
    \int_0^{t_1}\rmd t 
    \cbiq(\tau(t_1,t))[G_{t,t},W_t](\sigma k_0).
\end{align}
After doing a change of variables $k_0\mapsto \sigma k_0-k'$ for the last term
and noticing that\\ $\phi(\sigma k_0-k',R,u)=\phi(k_0,R,\sigma(u-k'))$, we note that \eqref{eq:Fevo} follows from Lemma \ref{lem:JFconnection}.

If $W$ satisfies the evolution equation \eqref{eq:Wevo} with $\tau$, then by
Lemma \eqref{lem:colG} $W$ satisfies \eqref{eq:WGevo} with $G=F[W]$, and hence
\eqref{eq:Fevo} holds with $G=F[W]$.
\end{proof}

\section{Estimates of the weight function $\Phi$}
\label{sec:estimates}

We will now go through all the estimates related to the weight function $\Phi$ defined 
in \eqref{eq:weight_phi} and used in the norm that controls the collision control
map $F[W]$. These estimates will be needed in order to study the fixed point equations 
given by the bilinear operator $\mathcal{I}_\tau$ defined in the previous section, and 
the evolution equations \eqref{eq:WGevo} and \eqref{eq:Wevo}.
Recall that $\Phi$ depends on the chosen parameter $\pp$ and the lattice size $L\ge2$.

\begin{lemma}
    \label{lem:3Phi_bound}
    For
    every $\pp\in(0,1)$, $\xi'\in\{0,1\}$, $r\in\R$, \(k \in \T\), $k',u\in \T$, $x,y,y'\in\onelattice$, $\sigma\in\{-,+\}$ and \(s \in
    [-L^{\pp}/4,L^{\pp}/4]\), we have 
	\begin{align*}
        \frac{\Phi(r,x)}
        {\Phi( R', \widetilde{x})
        \Phi(|sz_{\xi'}(k)|,\widetilde{y})} 
        &\leq 
        c_{\pp,9} \left(\mean{y}\mean{y'}\right)^{2/3}
        \max
        \left(
            \frac{1}{\regabs{ |r|-|s||z_{\xi'}(k)|}^{1/2}}
            ,
            \frac{1}{\regabs{s|z_{\xi'}(k)|}^{(1-\pp)/2}}
        \right),
	\end{align*}
    where $R' = |r\rme^{\ci2\pi(u-k')}+sz_{\xi'}(k)|$, and $\widetilde{x}=(x+\sigma y-y') {\bmod \onelattice}$ and $\widetilde{y}=(y+y'){\bmod\onelattice}$.
\end{lemma}
\begin{proof}
    By Corollary \ref{cor:modregabs}, there holds
    \begin{align}
        \label{eq:x2estim}
        \frac{\mean{x}}{3\mean{y}\mean{y'}} 
        \leq
        \mean{\widetilde{x}}
        \leq
        3\mean{x}\mean{y}\mean{y'}.
    \end{align}
    Recall the definition of $\Phi$ from
    \eqref{eq:weight_phi}. Note that $\varphi(|sz_{\xi'}(k)|)=\varphi_1(|sz_{\xi'}(k)|)$,
    since $|sz_{\xi'}(k)|\leq L^{\pp}/2$. We have to consider different cases for $\varphi(r)$ and $\varphi(R')$.

    We first assume that $|r| \leq L^{\pp}$. This means that $\gamma(r)=0$ and
    $\varphi(r)=\varphi_1(r)=\mean{r}^{1/2}$. By triangle inequality $|R'|\leq |r|+|sz_{\xi'}(k)|\leq 3L^{\pp}/2$, and 
    so, $\gamma(R')=0$. 
	\begin{itemize}
        \item[1.]
            $R'
            \leq {|r|}/2$. Then $|sz_{\xi'}(k)|
            \geq {|r|/2}$, and therefore 
            \begin{align*}
                \frac{\Phi(r,x)}{
                \Phi(|sz_{\xi'}(k)|,\widetilde{y})} 
                &\leq
                \frac{
                    \mean{x}^{1/3}\max \left( 1, \frac{\mean{r}^{1/2}}{\mean{x}^{1/2}} \right) 
                }{
                    \mean{\widetilde{y}}^{-1/6}  {\mean{|sz_{\xi'}(k)|}^{1/2}}
                }
                \leq 
                \mean{x}^{1/3}
                \mean{\widetilde{y}}^{1/6} 
                     \max \left( \frac{1}{\mean{|sz_{\xi'}(k)|}^{1/2}}, \frac{\sqrt{2}}{\mean{x}^{1/2}} \right) 
                 .
            \end{align*}
            Using the reverse triangle inequality for $R'$ we get 
            \begin{align*}
                 \frac{1}{\Phi(
                 R', \widetilde{x})}
                  &\leq
            \mean{\widetilde{x}}^{-1/3}\min(1,{\frac{\mean{\widetilde{x}}^{1/2}}{\mean{ |r|-|s||z_{\xi'}(k)|}^{1/2} }}).
            \end{align*}
            Together, these estimates give the desired
            bound, namely, 
            \begin{align*}
                \frac{\Phi(r,x)}
                {\Phi( R', \widetilde{x})
                \Phi(|sz_{\xi'}(k)|,\widetilde{y})} 
                            &\leq 
                            \sqrt{2} \frac{\mean{x}^{1/3}}{\mean{\widetilde{x}}^{1/3}}\mean{\widetilde{y}}^{1/6}
                            \max 
                            \left( 
                                \frac{1}{\mean{|sz_{\xi'}(k)|}^{1/2}}
                                ,
                                \frac{\mean{\widetilde{x}}^{1/2}}{\mean{x}^{1/2}}
                                \frac{1}{\mean{ |r|-|s||z_{\xi'}(k)|}^{1/2} }
                            \right) 
                            \\
                            &\leq 
                            3\sqrt{ \mean{y}\mean{y'}}
                            \max 
                            \left( 
                                \frac{1}{\mean{|sz_{\xi'}(k)|}^{1/2}}
                                ,
                                \frac{1}{\mean{ |r|-|s||z_{\xi'}(k)|}^{1/2} }
                            \right) .
            \end{align*}
            where we used Corollary \ref{cor:modregabs} to estimate $\mean{\widetilde{y}}\leq 2\mean{y}\mean{y'}$, and
            \eqref{eq:x2estim}  to estimate
            \begin{align}
                \label{eq:xxestim}
                \frac{\mean{x}^{1/3}}{\mean{\widetilde{x}}^{1/3}}
                \frac{\mean{\widetilde{x}}^{1/2}}{\mean{x}^{1/2}} 
                \leq
                (3\mean{y}\mean{y'})^{1/6}
                \leq
                (3\mean{y}\mean{y'})^{1/3}
                ,
                \text{  and   }
                \frac{\mean{x}^{1/3}}{\mean{\widetilde{x}}^{1/3}}
                \leq 
                (3\mean{y}\mean{y'})^{1/3}
            \end{align}
            
        \item[2.]
            $R'
            \in (r/2, L^{\pp})$. Then $\Phi(R',\widetilde{x})\ge
            \Phi(r,\widetilde{x})/\sqrt{2} $. Therefore, \eqref{eq:xxestim}
            gives 
            \begin{align*}
                \frac{\Phi(r,x)}{\Phi(R',\widetilde{x})}
                \leq
                \sqrt{2} 
                \frac{\mean{x}^{1/3}}{\mean{\widetilde{x}}^{1/3}}\max(1,\frac{\mean{\widetilde{x}}^{1/2}}{\mean{x}^{1/2}})
                \leq 
                \sqrt{2} 
                \left( 3 \mean{y}\mean{y'} \right) ^{1/3}
            \end{align*}

            Using again $\mean{\widetilde{y}}\leq 2\mean{y}\mean{y'}$ gives the 
            the trivial bound 
            \begin{align}
                \label{eq:trivial_bound_phi}
                 \Phi(|sz_{\xi'}(k)|,\widetilde{y})\ge
                   (2\mean{y}\mean{y'})^{-1/6} \regabs{s|z_{\xi'}(k)|}^{1/2}. 
            \end{align}
            Hence, putting everything together gives 
            \begin{align*}
                \frac{\Phi(r,x)}
                {\Phi( R', \widetilde{x})
                \Phi(|sz_{\xi'}(k)|,\widetilde{y})} 
                            &\leq 
                            3 \sqrt{\mean{y}\mean{y'}}
                                \frac{1}{\mean{|sz_{\xi'}(k)|}^{1/2}}
            \end{align*}
            as needed.

        \item[3.] 
            $R' \geq L^{\pp}$. Recall the definition of $\Phi(R',\widetilde{x})$. 

            Since
            $R'\leq |r| + |s||z_{\xi'}(k)|$ and $\abs{r}  \leq L^{\pp}$, 
            we have 
            \[0\leq{R' - L^{\pp}}
                \leq 
                |r| +|s||z_{\xi'}(k)|-L^{\pp}
                \leq
            |s||z_{\xi'}(k)|\]
		
		Consequently, we have the lower bound
		\begin{align}
            \Phi(R',\widetilde{x})
            \geq
            \mean{\widetilde{x}}^{1/3}
             \max 
             \left( 
                 1,
                 \mean{\widetilde{x}}^{-1/2}
            \frac{\regabs{L^{\pp}}^{1/2}}{\regabs{|s||z_{\xi'}(k)|}^{\pp/2}}
             \right).
		\end{align}
        Using this together with $\Phi(r,x)\leq \Phi(L^{\pp},x)$ and \eqref{eq:xxestim} gives us the following bound 
		\begin{align*}
            \frac{\Phi(r,x)}{\Phi(R',\widetilde{x})}
            \leq 
                \frac{ \mean{x}^{1/3} }
                { \mean{\widetilde{x}}^{1/3} }
                \max 
                \left( 
                    1,
                    \frac{ \mean{\widetilde{x}}^{1/2} }
                { \mean{x}^{1/2} }
            \regabs{|s||z_{\xi'}(k)|}^{\pp/2}
                \right) 
                \leq 
                \left( 3 \mean{y}\mean{y'} \right) ^{1/3}
            \regabs{|s||z_{\xi'}(k)|}^{\pp/2}.
		\end{align*}
         The trivial bound \eqref{eq:trivial_bound_phi} for $\Phi(|sz_{\xi'}(k)|,\widetilde{y})$ implies 
        \begin{align*}
                \frac{\Phi(r,x)}
                {\Phi( R', \widetilde{x})
                \Phi(|sz_{\xi'}(k)|,\widetilde{y})} 
                            &\leq 
                            3 \sqrt{\mean{y}\mean{y'}}
                                \frac{1}{\mean{|sz_{\xi'}(k)|}^{1/2-\pp/2}}.
        \end{align*}
	\end{itemize}

    Finally, suppose that $|r| > L^{\pp}$. 
    \begin{itemize}
        \item[4.]
            $R'> L ^{\pp}$.
            Triangle inequality together with monotonicity and subadditivity of 
            $\regabs{\cdot}^{\pp/2}$ imply
            \begin{align*}
                \frac{\Phi(r,x)}{\Phi( R',\widetilde{x})} 
                &\leq
                \frac{ \mean{x}^{1/3-\gamma(r)} }
                { \mean{\widetilde{x}}^{1/3-\gamma(R')} }
                \max 
                \left( 
                    1,
                    \frac{ \mean{\widetilde{x}}^{1/2} }
                { \mean{x}^{1/2} }
                \frac{\regabs{|r\rme^{\ci2\pi(u-k')}+sz_{\xi'}(k)|-L^{\pp}}^{\pp/2}}
                {\regabs{|r|-L^{\pp}}^{\pp/2}}
                \right) 
                \\
                &\leq
                \frac{ \mean{x}^{1/3-\gamma(r)} }
                { \mean{\widetilde{x}}^{1/3-\gamma(R')} }
                 \max 
                \left( 
                    1,
                    \frac{ \mean{\widetilde{x}}^{1/2} }
                { \mean{x}^{1/2} }
                \frac{\regabs{|r|-L^{\pp}}^{\pp/2}+\regabs{|sz_{\xi'}(k)|}^{\pp/2}}
                {\regabs{|r|-L^{\pp}}^{\pp/2}}
                \right) 
                \\
                &\leq
                \left( 3 \mean{y}\mean{y'} \right) ^{1/2}
                \mean{x}^{\gamma(R')-\gamma(r)}
                                (1+\regabs{|sz_{\xi'}(k)|}^{1/14}),
            \end{align*}
            where in the last step we used \eqref{eq:x2estim} and $0\leq\gamma\leq1/3$ 
            to estimate 
            \begin{align*}
                \frac{ 1 }
                { \mean{\widetilde{x}}^{1/3-\gamma(R')} }
            \leq
            (3\mean{y}\mean{y'}^{1/3})
                \frac{ 1 }
                { \mean{x}^{1/3-\gamma(R')} }
            \text{ and }
                { \mean{\widetilde{x}}^{1/6+\gamma(R')} }
                \leq
                \left( 3 \mean{y}\mean{y'} \right) ^{1/2}
                \mean{x}^{1/6+\gamma(R')}.
            \end{align*}
            Next, we prove that
            \begin{align}
                \label{eq:gamma_diff_bound_thing}
                \mean{x}^{\gamma(R')-\gamma(r)}\leq\mean{16^{1/(1-\pp)}/2}^{1/3}.
            \end{align}
            This is true for $L<16^{1/(1-\pp)}$ 
            since $|x|\leq L/2$ and $0\leq \gamma\leq 1/3$.
            
            Suppose now that $L\ge 16^{1/(1-\pp)}$. Then the Lipschitz constant of $\gamma$ is $\frac{32}{3}L^{-1}$.
            Therefore, 
            $$\gamma(R')-\gamma(r)\leq\gamma(|r|+L^{\pp}/2)-\gamma(|r|)\leq \frac{16}{3}L^{-(1-\pp)},$$ where in addition to the Lipschitz continuity of $\gamma$, we used
            that $\gamma$ is increasing and $R'\leq |r|+L^{1/7}/2$ by triangle
            inequality. Hence,
            $\mean{x}^{\gamma(R')-\gamma(r)}\leq\mean{L/2}^{(16/3)(L^{-(1-\pp)})}$.
            Note that
            $\mean{L/2}^{(16/3)(L^{-(1-\pp)})}$ is
            decreasing for $L\ge 16^{1/(1-\pp)}$. And therefore,
            the bound \eqref{eq:gamma_diff_bound_thing} holds for any $L$.

            Thus, the trivial bound \eqref{eq:trivial_bound_phi} for $\Phi(|sz_{\xi'}(k)|,\widetilde{y})$
            implies
            \begin{align*}
                \frac{\Phi(r,x)}{\Phi(R',\widetilde{x})\Phi(|sz_{\xi'}(k)|,\widetilde{y})}
            \leq 
            2^{7/6}\sqrt{3}\mean{16^{1/(1-\pp)}/2}^{1/3} ( \mean{y}\mean{y'})^{2/3}
                                \frac{1}{\mean{|sz_{\xi'}(k)|}^{1/2-\pp/2}}
            \end{align*}
        \item[5.]
            $R'\leq L ^{\pp}$. Then $\gamma(R')=0$ and we estimate $\gamma(r)\ge0$.
            By the triangle inequality the lower bound 
            \begin{align*}
                R'=|r\rme^{\ci2\pi(u-k')}+sz_{\xi'}(k)|
                &\ge L^{\pp}/2
            \end{align*}
            holds
            since $|z_{\xi'}(k)|\leq2$ and $|s|\leq L^{\pp}/4$. Therefore,
            $\Phi(R',\widetilde{x})\ge \varphi(L^{\pp},\widetilde{x})/\sqrt{2}$. Then, together with
            $\Phi(r,x)\leq\Phi(L^{\pp},x)$ and \eqref{eq:xxestim}, we get
            \begin{align*}
                \frac{\Phi(r,x)}{\Phi( R',\widetilde{x})}
                \leq 
                \sqrt{2}
                 \frac{ \mean{x}^{1/3} }
                 { \mean{\widetilde{x}}^{1/3} }
                 \max 
                \left( 
                    1,
                    \frac{ \mean{\widetilde{x}}^{1/2} }
                { \mean{x}^{1/2} }
                \right) 
                \leq \sqrt{2} 
                \left( 3 \mean{y}\mean{y'} \right) ^{1/3}.
            \end{align*}
            Then the
            trivial bound \eqref{eq:trivial_bound_phi} for $\Phi(|sz_{\xi'}(k)|,\widetilde{y})$ gives once again the
            desired bound.
    \end{itemize}
\end{proof}

\begin{lemma}
    \label{lem:exp_lowerbound}
    For all $k\in[-1/2,1/2]$ there holds 
    \begin{align}
        \label{eq:exp_lowerbound}
        |1-e^{\ci 2\pi k}| \ge 4|k|.
    \end{align}
\end{lemma}
\begin{proof}
    By trigonometric identities $|1-e^{\ci 2\pi k}|=2|\sin(\pi k)|$. The result follows then from 
     the fact that $\sin(\pi k)-2k\ge0$ for $k\in[0,1/2]$, which can be proved by a standard calculus argument.
\end{proof}
\begin{lemma}
    \label{lem:3Phi_int_bound}
    Fix $d\ge3$ and $\pp\in(0,1-2/d)$.
    Let $R\in\R^d$, and $\xi'\in\{0,1\}$, and $x,y,y'\in\lattice$ and
    $\sigma\in\{-,+\}$. For each fixed $k \in \duallattice$ with
    $k_\ell\ne \xi'/2$ for all $\ell=1,\dots,d$, we have
	\begin{align}
         \int_{-L^{\pp}/4}^{L^{\pp}/4}\rmd s
         \frac{\Phi^d(R,x)}
         {\Phi^d(\rs{\xi'}(R,u,k',k),\widetilde{x})
         \Phi^d(\rs{\xi'}(k),\widetilde{y})} 
        &\leq 
        c_{\pp,7}
        {\prod_{\ell= 1}^d} 
        \frac{\Big(\mean{y_\ell}\mean{y'_\ell}\Big)^{2/3}}{|k_\ell-\xi'/2+n_\ell|^{1/d}},
    \end{align}
    where $\widetilde{x}=(x+\sigma y-y')\bmod \lattice$ and
    $\widetilde{y}=(y+y')\bmod\lattice$ and $n_\ell\in\Z$ is the integer for
    which $k_\ell-\xi'/2+n_\ell\in(-1/2,1/2]$.
\end{lemma}

\begin{proof}
    Recall that $\Phi^d(R,x) = \prod_{\ell=1}^d\Phi(R_\ell,x_\ell)$. From generalized H\"older, we get
     \begin{align}
        & \int_{-L^{\pp}/4}^{L^{\pp}/4}\rmd s
        {\prod_{\ell= 1}^d} 
        \frac{\Phi(R_\ell,x_{\ell})}
        {\Phi(\rs{\xi'}(R,u,k',k)_\ell,\widetilde{x}_{\ell})
        \Phi(\rs{\xi'}(k)_\ell,\widetilde{y}_{\ell})} 
        \\
        &\leq 
        {\prod_{\ell= 1}^d} 
        \left( 
        \int_{-L^{\pp}/4}^{L^{\pp}/4}\rmd s
        \left( 
            \frac{\Phi(R_\ell,x_{\ell})}
            {\Phi(\rs{\xi'}(R,u,k',k)_\ell,\widetilde{x}_{\ell})
            \Phi(\rs{\xi'}(k)_\ell,\widetilde{y}_{\ell})} 
        \right)^d
        \right)^{1/d}
         \\
        &\leq 
        4c_{\pp,7}
        {\prod_{\ell= 1}^d} 
        \left(\mean{y_\ell}\mean{y'_\ell}\right)^{2/3}
        \frac{1}{|z_{\xi'}(k_\ell)|^{1/d}},
	\end{align}
    where the last step follows from Lemma \ref{lem:3Phi_bound} and the fact
    that for each $p>1$, $v>0$ and $r\in\R$ there holds $\int_{\R}
    {\regabs{r-vs}^{-p}}\rmd s = \int_{\R} {\regabs{s}^{-p}}\rmd s/v$ and
    $$\int_{\R} {\regabs{s}^{-p}}\rmd
    s=\sqrt{\pi}\frac{\Gamma((p-1)/2)}{\Gamma(p/2)}<\infty,$$ 
    where $\Gamma$ is the gamma
    function.  In our case, $p=d(1-\pp)/2>1$ and $v=|z_{\xi'}(k_\ell)|$ and $r=|R_\ell|$. From Lemma \ref{lem:exp_lowerbound} we get
    $|z_{\xi'}(k_\ell)|\ge4|k_\ell-\xi'/2+n_\ell|$,  where $n_\ell\in\Z$ is the
    integer for which $k_\ell-\xi'/2+n_\ell\in(-1/2,1/2]$. 
\end{proof}

\begin{lemma}
\label{lem:3Phi_2int_bound} 
Fix $d\ge3$ and $\pp\in(0,1-2/d)$.
For any choice of $\xi'\in\{0,1\}$, $R\in \R^d$,$u\in \T^d$, $k'\in\duallattice$, $x,y,y'\in\lattice$, $\sigma\in\{-,+\}$, $L
\geq 2$ and $\tau_0>0$, we have that
\begin{align*}
    \int_{-\lambda^{-2}\tau_0}^{\lambda^{-2}\tau_0} \rmd s 
    \int_{\duallattice} \rmd k
    \frac{\Phi^d(R,x)}
    {\Phi^d(\rs{\xi'}(R,u,k',k),\widetilde{x})
    \Phi^d(\rs{\xi'}(k),\widetilde{y})} 
    \leq c_{\pp,8} 
        \left({{\prod_{\ell= 1}^d} 
        \mean{y_\ell}\mean{y'_\ell}}\right)^{2/3}
\end{align*} 
whenever $L \geq 
L_\pp(\lambda,\tau_0)$. Here $\widetilde{x}=(x+\sigma y-y')\bmod \lattice$ and
    $\widetilde{y}=(y+y')\bmod\lattice$.
\end{lemma}

\begin{proof}
    We split the set $\duallattice$ into two sets, namely, ${\duallattice\setminus N}$ 
    and ${N} $, where $$N=\{k\in\duallattice \ : \ k_\ell = \xi'/2 \text{ for
    some } \ell=1,\dots,d\}.$$

    We start by considering the region $N$. Take $k\in N$ and such $\ell$ that
    $k_\ell= \xi'/2$. Then $z_{\xi'}(k_\ell)=0$, and so, Lemma \ref{lem:3Phi_bound} gives
    \begin{align}
        \label{eq:3Phi1}
        \frac{\Phi(R_\ell,x_{\ell})}
        {\Phi(\rs{\xi'}(R,u,k',k)_\ell,\widetilde{x}_{\ell})
        \Phi(\rs{\xi'}(k)_\ell,\widetilde{y}_{\ell})} 
        &\le
        c_{\pp,9} \left({\mean{y_\ell}\mean{y'_\ell}}\right)^{2/3}
        .
    \end{align}
    Since $\#N\leq dL^{d-1}$ and $\int_{\duallattice} \rmd k=1$ we obtain 
    \begin{align*}
        \int_{-\tau_0\lambda^{-2}}^{\tau_0\lambda^{-2}} \rmd s 
        \int_{N} \rmd k
        \frac{\Phi^d(R,x)}
        {\Phi^d(\rs{\xi'}(R,u,k',k),\widetilde{x})
        \Phi^d(\rs{\xi'}(k),\widetilde{y})} 
         &\leq
        \int_{-\tau_0\lambda^{-2}}^{\tau_0\lambda^{-2}} \rmd s 
        \int_{N} \rmd k
        {\prod_{\ell= 1}^d} 
        c_{\pp,9} \left({\mean{y_\ell}\mean{y'_\ell}}\right)^{2/3}
    \\
         &\leq
         2
         d
         \tau_0\lambda^{-2}L^{-1}
        {\prod_{\ell= 1}^d} 
        c_{\pp,9} \left({\mean{y_\ell}\mean{y'_\ell}}\right)^{2/3}
    \end{align*}
    By $L\ge L_\pp(\lambda,\tau_0)$ we have $\tau_0\lambda^{-2}L^{-1}\leq 1$.
    In the complement, we use $L\ge(4\tau_0\lambda^{-2})^{1/\pp}$ together with
    Lemma \ref{lem:3Phi_int_bound} to obtain 
    \begin{align*}
        &\int_{-\tau_0\lambda^{-2}}^{\tau_0\lambda^{-2}} \rmd s 
        \int_{\duallattice\setminus N} \rmd k
        \frac{\Phi^d(R,x_1)}{\Phi^d(\rs{\xi'}(R,u,k',k),x_2)
        \Phi^d(\rs{\xi'}(k),x_3)} 
         \\
         &\leq
         c_{\pp,7}
         \left(
         \int_{\oneduallattice\setminus \{0\}} \rmd k_1
         \frac{1}{|k_1|^{1/d}}
        \right)^d
        {\prod_{\ell= 1}^d} 
         \left({\mean{y_\ell}\mean{y'_\ell}}\right)^{2/3}
        \\
         &\leq
         c_{\pp,7}
         \left( 
         2\int_0^{1/2} \rmd x x^{-1/d}
         \right) ^d
        {\prod_{\ell= 1}^d} 
         \left({\mean{y_\ell}\mean{y'_\ell}}\right)^{2/3}
         \\
         &= 
         2
         c_{\pp,7}
         \left( 
             \frac{d}{d-1}
         \right) ^d
        {\prod_{\ell= 1}^d} 
         \left({\mean{y_\ell}\mean{y'_\ell}}\right)^{2/3}.
    \end{align*}
\end{proof}

\begin{lemma}
    \label{lem:singularint}
    Let $\pp\in(0,1)$.
    There holds 
    \begin{align}
        \label{eq:singularint}
        \int_0^{1/2} \rmd k 
        \frac{1}{\abs{r-s \sin(\pi k)}^{(1-\pp)/2}} 
        &\leq (1+\frac{2}{\pi\pp})s^{-(1-\pp)/2}
    \end{align}
    for all $s>0$ and $r\ge0$.
\end{lemma}

\begin{proof}
    Denote the integral on the LHS of \eqref{eq:singularint} by $I$. Change of
    variables $y=\sin(\pi k)$, gives $\rmd k = \frac{\rmd
    y}{\pi\sqrt{1-y^2}}\leq \frac{\rmd
    y}{\pi\sqrt{1-y}}$. Then 
    \begin{align*}
        I\leq \int_0^1 \frac{\rmd
    y}{\pi\sqrt{1-y}}      
    \frac{1}{\abs{r-s y}^{(1-\pp)/2}}
    =
    \frac{s^{-(1-\pp)/2}}{\pi}
    \int_0^1 \frac{\rmd x}{\sqrt{x}}  \frac{1}{\abs{x-\alpha}^{(1-\pp)/2}},
    \end{align*}
    where we did another change of variables $x=1-y$ and defined $\alpha = 1-r/s$.
    Note that $\alpha\in(-\infty,1]$. 

    Suppose first that $\alpha\leq0$. Then 
    \begin{align*}
        \int_0^1 \frac{\rmd x}{\sqrt{x}}  \frac{1}{\abs{x-\alpha}^{(1-\pp)/2}}
        &\leq  
        \int_0^1 \frac{\rmd x}{x^{1/2+(1-\pp)/2}} = 2/\pp.
    \end{align*}

    Suppose now that $\alpha\in(0,1]$. Then, after change of variables $x=\alpha
    x'$, we have
    \begin{align*}
        \int_0^1 \frac{\rmd x}{\sqrt{x}}  \frac{1}{\abs{x-\alpha}^{(1-\pp)/2}},
        =
        \alpha^{\pp/2}
        \int_0^{1/\alpha} \frac{\rmd x'}{\sqrt{x'}}  \frac{1}{\abs{x'-1}^{(1-\pp)/2}}.
    \end{align*}

The fact that $\frac{\rmd}{\rmd x'} 2\arcsin(\sqrt{x'})=1/\sqrt{x'{(1-x')}} $ implies 
\begin{align*}
     \int_0^{1} \frac{\rmd x'}{\sqrt{x'}}  \frac{1}{\abs{x'-1}^{(1-\pp)/2}}
     \leq 
     \int_0^{1} \frac{\rmd x'}{\sqrt{x'}}  \frac{1}{\abs{x'-1}^{1/2}}
     = 
     \pi.
\end{align*}
If $1/\alpha>1$, then 
\begin{align*}
     \int_1^{1/\alpha} \frac{\rmd x'}{\sqrt{x'}}  \frac{1}{\abs{x'-1}^{(1-\pp)/2}}
     \leq 
     \int_1^{1/\alpha} \frac{\rmd x'}{\sqrt{x'-1}}  \frac{1}{\abs{x'-1}^{(1-\pp)/2}}
     = \frac{2}{\pp} (1/\alpha-1)^{\pp/2}.
\end{align*}
Then, employing $\alpha^{\pp/2} ( \pi +(2/\pp)(1/\alpha-1)^{\pp/2} )\leq \pi+2/\pp$ for
$\alpha\in(0,1]$ finishes the proof.
\end{proof}

\begin{lemma}
    \label{lem:int_over_weight}
    Let $d\ge 3$ and $\pp\in(0,1-2/d)$.
    Let $R\in \R_+^d$, $\xi'\in\{0,1\}$, $T>0$. Suppose 
    both $\tau,\widetilde{\tau}\colon\R_+^2\to\R$ satisfy one of the conditions \eqref{eq:taudiff} or \eqref{eq:tauconst}. If $L\ge L_\pp(\lambda, \max(\tau,\widetilde{\tau})(T,0))$, then
    \begin{align}
        \label{eq:int_over_weight1}
                \int_0^{T}
                \rmd t
                \int_{\lambda^{-2}{\tau}(T,t)}^{\lambda^{-2}\widetilde\tau(T,t)}
                |\rmd s|
                \int_{\duallattice}\rmd k
                \prod_{\ell=1}^d
                \frac{1}{\regabs{ R_\ell-|s||z_{\xi'}(k_\ell)|}^{(1-\pp)/2}}
            \leq \mathcal{E}_\pp(\lambda,\tau,\widetilde\tau, T),
    \end{align}
    where the error term $\mathcal{E}_\pp$ was defined in Definition \ref{def:error}.
    Moreover, 
    \begin{align}
        \label{eq:int_over_weight2}
                \int_{-\infty}^{\infty}
                \rmd s
                \int_{\duallattice}\rmd k
                \prod_{\ell=1}^d
                \frac{1}{\regabs{ R_\ell-|s||z_{\xi'}(k_\ell)|}^{(1-\pp)/2}}
                \leq c_{\pp,6}.
    \end{align}
\end{lemma}
\begin{proof}
    We start by doing a change of variables $k_\ell\mapsto k_\ell+\xi'/2$ and estimating 
    the $k$ integral by Lebesgue integral, 
    \begin{align}
        \label{eq:someint}
                \int_{\duallattice}\rmd k
                \prod_{\ell=1}^d
                \frac{1}{\regabs{ R_\ell-|s||z_{\xi'}(k_\ell)|}^{(1-\pp)/2}}
                \leq 
                \prod_{\ell=1}^d
                2
                \int_0^{1/2}\rmd k_\ell
                \frac{1}{\regabs{ R_\ell-2|s||\sin(\pi k_\ell)|}^{(1-\pp)/2}}.
    \end{align}
    Then using the trivial bound $1$ for \eqref{eq:someint} when $|s|\leq1$, together with Lemma \ref{lem:singularint} for 
    $|s|>1$, implies that the integral \eqref{eq:int_over_weight2}
    is bounded by
    \begin{align}
        c_{\pp,6}=
        2+
        2^{d+1}
        (1+\frac{2}{\pi\pp})^d
        \int_1^\infty \rmd s 
        (2s)^{-d(1-\pp)/2}
        <\infty.
    \end{align}
    This proves \eqref{eq:int_over_weight2}.
    
    The integral \eqref{eq:int_over_weight1} is zero if $\tau=\widetilde\tau$. 
    Suppose $\tau\ne\widetilde\tau$. Then without loss of generality we can assume that 
    $\widetilde{\tau}=T_0$ on $(0,T)^2$ for some constant $T_0>0$.
    By \eqref{eq:someint}, the $t$-integrand in \eqref{eq:int_over_weight1} is 
    bounded by 
    \begin{align}
        I(t)\coloneqq    
                \int_{\lambda^{-2}{\tau}(T,t)}^{\lambda^{-2}T_0}
                |\rmd s|
                \prod_{\ell=1}^d
                2
                \int_0^{1/2}\rmd k_\ell
                \frac{1}{\regabs{ R_\ell-2|s||\sin(\pi k_\ell)|}^{(1-\pp)/2}}.
    \end{align}

    Suppose first that $\tau$ is also a constant and denote it by $\widetilde{T}_0>0$.
    Let 
    \begin{align}
    \tildepp &= (1-\pp)/2,
    \\
    \eta_d&=(2(1+\frac{2}{\pi\pp})2^{-\tildepp})^d.
    \end{align}
    Note that $1/d<\tildepp <1/2$.
    Then 
    Lemma \ref{lem:singularint} implies 
    \begin{align*}
        I(t)&\leq
        \eta_d
        \Bigg|
        \int_{\lambda^{-2}\widetilde{T}_0}^{\lambda^{-2}T_0}
                \rmd s
                  s^{-d\tildepp}
        \Bigg|
        \\ 
        &\leq
        \frac{\eta_d}{d\tildepp-1} \left| T_0^{-(d\tildepp-1)}-\widetilde{T}_0^{-(d\tildepp-1)} \right|  
        \lambda^{2d\tildepp-2}.
    \end{align*}
    Hence, the desired result follows for this case. 

    Now suppose  that  $\tau(T,t)=T-t$ and $\widetilde{\tau}(T,t)=T_0$. 
    To estimate the $k_\ell$ -integrals we note that the integrand is trivially 
    bounded by $1$. Using this trivial bound up to $s\leq 1$ and 
    for $s>1$ using Lemma \ref{lem:singularint} we get for $t\in[0,T_0]$ that 
    \begin{align}
        \nonumber
        I(T-t)&\leq
                \int_{\lambda^{-2}t}^{\lambda^{-2}T_0}
                \rmd s
                \left(  \ind(s \leq 1) + \eta_d s^{-d\tildepp}\ind(s>1) \right) 
        \\
        \nonumber
        &\leq 
        \ind(\lambda^{-2}t\leq 1)
        \left( 
        \int_0^{1}\rmd s  
        +
        \int_{1}^\infty\rmd s \eta_d  s^{-d\tildepp}
        \right) 
        +
        \ind(\lambda^{-2}t> 1)
        \int_{\lambda^{-2}t}^\infty
        \eta_d s^{-d\tildepp}
        \\
        \label{eq:someestimate2}
        &=
         C_2\ind(\lambda^{-2}t\leq 1)
         +
        C_1\ind(\lambda^{-2}t> 1)
        (\lambda^{-2}t)^{-(d\tildepp-1)},
    \end{align}
    where $C_1\coloneqq  {\eta_d}/{(d\tildepp-1)}$ and $C_2\coloneqq  1 +C_1 $.  
    Let us now assume that $T_0\ge T$.
    Then 
    \begin{align}
        \nonumber
        &\int_0^{T} \rmd t I(t) 
        =\int_0^{T} \rmd t I(T-t) 
        \\
        \nonumber
        &\leq 
        C_2 \min(\sqrt{T},\lambda)^2
        +
        C_1
        \ind(T> \lambda^2)
         \lambda^{2d\tildepp-2}
        \int_{\lambda^{2}}^{T} \rmd t\,
        t^{-(d\tildepp-1)}
        \\
        \nonumber
        &\leq 
        C_2 \min(\sqrt{T},\lambda)^2
        \\
        \nonumber
        &\quad
        +
        C_1
        \ind(T> \lambda^2)
        \left(  \ind(d\tildepp< 2) \frac{T^{2-d\tildepp}}{2-d\tildepp}\lambda^{2d\tildepp-2}
        +
        \ind(d\tildepp=2)
        \ln(T\lambda^{-2})
        \lambda^{2}
            +
        \ind(d\tildepp > 2) \frac{1}{d\tildepp-2}
        \lambda^2
        \right)
        \\
        \nonumber
        &\leq
        \ind(d\tildepp\ne2)
        \alpha_d T^{1-p_d/2}\min(\sqrt{T},\lambda)^{p_d}
        \\
        \label{eq:someestimate}
        &\quad+
        \ind(d\tildepp=2)
        \left(
        C_1
        \ln(T\lambda^{-2}+\rme)
        +
        C_2
        \right)
        \min(\sqrt{T},\lambda)^{2},
    \end{align}
    where in the last step we used that $\ind(T>\lambda^2)\lambda\leq \min(\sqrt{T},\lambda)$ and, for the $d\tildepp<2$ case, we also used $\min(\sqrt{T},\lambda)^2\leq T^{1-p_d/2}\min(\sqrt{T},\lambda)^{p_d}$. Recall that $p_d=2d\tildepp-2$ for $d\tildepp<2$ and $p_d=2$ otherwise.

    Finally, let us assume that $T_0\leq T$. Then 
    \begin{align}
    \label{eq:someestimate11}
        \int_0^{T} \rmd t I(t) 
        =
        \int_0^{T_0} \rmd t I(T-t) 
        +
        \int_{T_0}^{T} \rmd t I(T-t).
    \end{align}
    Now \eqref{eq:someestimate2} and \eqref{eq:someestimate} imply for the first integral that 
    \begin{align}
        \int_0^{T_0} \rmd t I(T-t) 
        \nonumber
        &\leq
        \ind(d\tildepp\ne2)
        \alpha_d T_0^{1-p_d/2}\min(\sqrt{T_0},\lambda)^{p_d}
        \\
    \label{eq:someestimate12}
        &\quad+
        \ind(d\tildepp=2)
        \left(
        C_1
        \ln(T_0\lambda^{-2}+\rme)
        +
        C_2
        \right)
        \min(\sqrt{T_0},\lambda)^{2}
    \end{align}
    Using the same argument as in \eqref{eq:someestimate2} for $t\in[T_0,T]$ gives 
    \begin{align}
    \nonumber
        I(T-t)&\leq
        C_2\ind(\lambda^{-2}T_0\leq 1)
         +
        C_1\ind(\lambda^{-2}T_0> 1)
        (\lambda^{-2}T_0)^{-(d\tildepp-1)}
        \\
    \label{eq:someestimate13}
              &\leq
          \left(
          C_2
          +
          C_1
          \right)
           \min\left(1,\frac{\lambda}{\sqrt{T_0}}\right)^{2}.
    \end{align}
    Hence, for $d\tildepp\ne2$ we can use $C_2+C_1\leq\alpha_d$ 
    together with \eqref{eq:someestimate11} and \eqref{eq:someestimate12} to obtain the desired result in this case, namely
    \begin{align*}
        \int_0^{T} \rmd t I(t) 
        &\leq
         \alpha_d T
         \min\left(1,\frac{\lambda}{\sqrt{T_0}}\right)^{p_d}.
    \end{align*}
    For $d\tildepp\ne2$ using $C_1\leq C_1 \ln(T_0\lambda^{-2}+\rme)$ in \eqref{eq:someestimate13}
    together with \eqref{eq:someestimate11} and \eqref{eq:someestimate12} imply
    \begin{align*}
        \int_0^{T} \rmd t I(t) 
        &\leq
        \left(
        C_1
        \ln(T_0\lambda^{-2}+\rme)
        +
        C_2
        \right)
        T
         \min\left(1,\frac{\lambda}{\sqrt{T_0}}\right)^{2}.
    \end{align*}
    Recalling that $C_2=C_1+1$ and renaming $\eta_{d,\pp}=C_1$ completes the proof.
\end{proof}

\begin{corollary}
    \label{cor:1Phi_int}
    For all $L\ge2$, $x\in\lattice$,  and $\xi'\in\{0,1\}$ there
    holds
    \begin{align}
        \label{eq:1Phi_int}
        \int_{-L^{1/7}/2}^{L^{1/7}/2} \rmd s 
        \int_{\duallattice} \rmd k 
        \frac{1}{\Phi^d(\rs{\xi'}(k),x)}
        &\leq
        c_{6}
        \prod_{j=1}^d\mean{x_j}^{1/6},
    \end{align}
    where $c_{6}=\inf\{c_{\pp,6}\colon 0<\pp<1-2/d\}$. 
\end{corollary}
\begin{proof}
    Recall from  \eqref{eq:R_shorthand_tilde} that $\rs{\xi'}(k)_j=|s||z_{\xi'}(k_j)|$. 
    Then $|\rs{\xi'}(k)_j|\leq  L^{1/7}$ for $|s|\leq L^{1/7}/2$. 
    Therefore, $\varphi(\rs{\xi'}(k)_j)
    =\mean{|s||z_{\xi'}(k_j)|}^{1/2}$. Hence,
    \begin{align}
        \frac{1}{\Phi^d(\rs{\xi'}(k),x)}
        \leq 
        \prod_{j=1}^d
        \frac{\mean{x_j}^{1/6}}{\mean{|s||z_{\xi'}(k_j)|}^{1/2}}
    \end{align}
    Thus, the bound \eqref{eq:int_over_weight2} from Lemma \ref{lem:int_over_weight} together with the fact that $\mean{\cdot}^{(1-\pp)/2}\leq \mean{\cdot}^{1/2}$ for any $\pp\in(0,1-2/d)$ 
    gives the desired result.
\end{proof}

\section{Fixed point equations given by $\mathcal{I}$}
\label{sec:fixpoint_eqs}

As discussed in Section \ref{sec:keyideas}, our proof strategy involves
proving existence of unique solutions to the fixed point equations
$G=\mathcal{I}_\tau[G,G]$ and $\mathcal{G}=\mathcal{I}_\tau[G,\mathcal{G}]$. Recall that in
\eqref{eq:I} the bilinear operator $\mathcal{I}_\tau$ was defined via the
bilinear operator $\mathcal{J}_\tau$, which in turn is given by \eqref{eq:J}.  Therefore, the following Lemma
\ref{lem:J_estim} about $\mathcal{J}$ is the main work horse for proving
boundedness and contractivity of $\mathcal{I}_\tau$ and also error between two
different choices of $\tau$.

\begin{lemma}
    \label{lem:J_estim} 
    Take $d\ge3$ and $\pp\in(0,1-2/d)$.
    Let $\lambda>0$ and $\tmax>0$. Suppose that both $\tau,\widetilde{\tau}\colon\R_+^2\to\R_+$ 
    satisfy one of the conditions \eqref{eq:taudiff} or \eqref{eq:tauconst}. Let $L\ge
    L_\pp(\lambda,\max(\tau,\widetilde{\tau})(\tmax,0))$. Let
    $t_1,t_2\in[0,\tmax]$. Then 
    \begin{align}
        \nonumber
        &\maxnorm{
        \mathcal{J}_\tau[G,\mathcal{G}]_{t_1,t_2}
            -
           \mathcal{J}_{\widetilde{\tau}}[\widetilde{G},\widetilde{\mathcal{G}}]_{t_1,t_2}
        } \\
        &\leq \nonumber
        \frac{c_{\pp,2}}{4}
        M_V
        \int_0^{t_2} \rmd t 
        \left( 
            \maxnorm{\widetilde{\mathcal{G}}_{t_1,t}} 
            \maxnorm{
            G_{t,t}
            -
            \widetilde{G}_{t,t}
            }
            +
            \maxnorm{
            G_{t,t}
            }
            \maxnorm{
            \mathcal{G}_{t_1,t}
            -
            \widetilde{\mathcal{G}}_{t_1,t}
            }
        \right) 
        \\
        &+
        \frac{c_{\pp,4}}{2}
        M_V
        \ppnorm{\widetilde{G}}{t_2}
        \ppnorm{\widetilde{\mathcal{G}}}{{\max(t_1,t_2)}}
        \mathcal{E}_\pp(\lambda,\tau,\widetilde{\tau},t_2)
        \label{eq:J_estim}
    \end{align}
    for all $G,\mathcal{G},\widetilde{G}, \mathcal{\widetilde{G}}\in \CT{\tmax}$. Here
    $\mathcal{E}_\pp$ is defined in Definition \ref{def:error}.
\end{lemma}

\begin{proof}
    Fix $X=(R,k',u,\sigma,x)\in\xl$ and $t_1,t_2\in[0,\tmax]$. Let
    $\bar{v}_i(y)=\sup_{k\in\duallattice}|v_i(k,y)|$. In the following computation,
    we use the shorthand notations
    \begin{align*}
        G(1)&= 
            G_{t,t}^{(\widetilde{n}_j)}(\Jpoint_s^j(X,k,y,y'))
            ,
            \qquad
            \mathcal{G}(2)=
            \mathcal{G}_{t_1,t}^{(n)}(\Jppoint_s^j(X,k,y,y'))
            \\
        \widetilde{G}(1)&=
            \widetilde{G}_{t,t}^{(\widetilde{n}_j)}(\Jpoint_s^j(X,k,y,y')),
            \qquad
        \widetilde{\mathcal{G}}(2)=
            \widetilde{\mathcal{G}}_{t_1,t}^{(n)}(\Jppoint_s^j(X,k,y,y'))
    \end{align*}
    We have \begin{align*}
       & 
       \Phi^d(R,x)
        \left| 
        \mathcal{J}^{(n)}_\tau[G,\mathcal{G}]_{t_1,t_2}(X)
            -
            \mathcal{J}^{(n)}_{\widetilde{\tau}}
            [\widetilde{G},\mathcal{\widetilde{G}}]_{t_1,t_2}(X)
        \right| 
        \\
       &\leq
        \sum_{j=1}^{8} 
       \Phi^d(R,x)
        \int_0^{t_2}\rmd t 
        \int_{\duallattice}\rmd k
        \sum_{y,y'\in\lattice} 
        \bar{v}_j(y)|w_j(y')|
        \\
        &\qquad\times
        \Bigg| 
            \int_{-\lambda^{-2}\tau(t_2,t)}^{\lambda^{-2}\tau(t_2,t)}
            \rmd s
            G(1)
            \mathcal{G}(2)
       -
           \int_{-\lambda^{-2}\widetilde{\tau}(t_2,t)}^{\lambda^{-2}\widetilde{\tau}(t_2,t)}
            \rmd s
        \widetilde{G}(1)
        \widetilde{\mathcal{G}}(2)
        \Bigg| 
        \\
       &\leq
        \sum_{j=1}^{8} 
        \sum_{y,y'\in\lattice} 
        \bar{v}_j(y)|w_j(y')|
        \int_0^{t_2}\rmd t 
        \int_{\duallattice}\rmd k
        \int_{-\lambda^{-2}\tau(t_2,t)}^{\lambda^{-2}\tau(t_2,t)}
        \rmd s
        \Phi^d(R,x)
        \Big| 
            G(1)
            \mathcal{G}(2)
            -
        \widetilde{G}(1)
        \widetilde{\mathcal{G}}(2)
        \Big| 
        \\
       &\qquad+ 
        \sum_{j=1}^{8} 
        \sum_{y,y'\in\lattice} 
        \bar{v}_j(y)|w_j(y')|
        \int_0^{t_2}\rmd t 
        \int_{\duallattice}\rmd k
            \int_{-\lambda^{-2}{\widetilde{\tau}}(t_2,t)}^{-\lambda^{-2}{\tau}(t_2,t)}
            \rmd s
        \Phi^d(R,x)
        \Big| 
        \widetilde{G}(1)
        \widetilde{\mathcal{G}}(2)
        \Big| 
        \\
       &\qquad+ 
        \sum_{j=1}^{8} 
        \sum_{y,y'\in\lattice} 
        \bar{v}_j(y)|w_j(y')|
        \int_0^{t_2}\rmd t 
        \int_{\duallattice}\rmd k
            \int_{\lambda^{-2}{\tau}(t_2,t)}^{\lambda^{-2}\widetilde{\tau}(t_2,t)}
            \rmd s
        \Phi^d(R,x)
        \Big| 
        \widetilde{G}(1)
        \widetilde{\mathcal{G}}(2)
        \Big| 
        \\
       &=:
        \sum_{j=1}^{8} 
        (A(j)+B_-(j)+B_+(j))
    \end{align*}
    In each of the terms $A(j)$, we can telescope with $G(1)\widetilde{\mathcal{G}}(2)$
    and then use Lemma \ref{lem:3Phi_2int_bound} together with Corollary \ref{cor:modregabs} and Lemma \ref{lem:ineqVconst}
    to obtain
    \begin{align*}
        A(j) &\leq
        \frac{c_{\pp,2}}{32}
        M_V
        \int_0^{t_2} \rmd t 
        \left( 
            \|\widetilde{\mathcal{G}}_{t_1,t}^{(n)}\|_{\Phi^d}
            \|
            G^{(\widetilde{n}_j)}_{t,t}
            -
            \widetilde{G}^{(\widetilde{n}_j)}_{t,t}
            \|_{\Phi^d}
            +
            \|
            G^{(\widetilde{n}_j)}_{t,t}
            \|_{\Phi^d}
            \|
            \mathcal{G}_{t_1,t}^{(n)}
            -
            \widetilde{\mathcal{G}}_{t_1,t}^{(n)}
            \|_{\Phi^d}
        \right).
    \end{align*}
    By Lemma \ref{lem:3Phi_bound} and Lemma \ref{lem:ineqVconst} we have
    \begin{align*}
        B_\pm(j)
        &\leq
        (2^{2d/3}+3)
        c_{\pp,9}^d
        M_V
        \int_0^{t_2}\rmd t
            \|
            \widetilde{G}^{(\widetilde{n}_j)}_{t,t}
            \|_{\Phi^d}
            \|\widetilde{\mathcal{G}}_{t_1,t}^{(n)}\|_{\Phi^d}
            I_\pm(t,j)
            \\
        &\leq
        (2^{2d/3}+3)
        c_{\pp,9}^d
        M_V
        \ppnorm{\widetilde{G}}{t_2}
        \ppnorm{\widetilde{\mathcal{G}}}{{\max(t_1,t_2)}}
        \int_0^{t_2}\rmd t
            I_\pm(t,j)
    \end{align*}
    where we defined 
    \begin{align*}
            I_\pm(t,j)
            &\coloneqq
            \Bigg|
                \int_{\pm\lambda^{-2}{\tau}(t_2,t)}^{\pm\lambda^{-2}\widetilde\tau(t_2,t)}
                \rmd s
        \int_{\duallattice}\rmd k
                \prod_{\ell=1}^d
         \max\left(
             \frac{1}{\regabs{ |R_\ell|-|s||{z_{\xi_j}(k_\ell)}|}^{1/2}}
            ,
        \frac{1}{\regabs{s|{z_{\xi_j}(k_\ell)}|}^{(1-\pp)/2}}\right)
            \Bigg|
    \end{align*}
    Lemma \eqref{lem:int_over_weight}
    finishes the proof.
\end{proof}

\begin{proposition}
    \label{prop:Ggronwall}
    Take $d\ge3$ and $\pp\in(0,1-2/d)$.
    Take $\tmax>0$. Suppose $\tau(T,t)=T-t$ and $L\ge L_\pp(\lambda,\tmax)$,
    or $\tau(T,t)=T_0$ for some $T_0>0$ and $L\ge L_\pp(\lambda,T_0)$.  Let 
    $T\in[0,\tmax]$. If
    $G\in\CT{\tmax}$ satisfies Equation \eqref{eq:Gevo}, then there
    holds 
    \begin{align}
        \label{eq:Gnorm_bound0}
        \ppnorm{G}{T}
        &\leq
        \frac{ \maxnorm{G_{0,0}}}{\sqrt{1-c_{\pp,2} M_V\maxnorm{G_{0,0}}T}}
    \end{align}
    whenever $T < \frac{1}{c_{\pp,2} M_V \maxnorm{G_{0,0}} }$.
    In particular, if $T\leq
    \frac{1}{2c_{\pp,2} M_V \maxnorm{G_{0,0}} }$ then
    \begin{align}
        \label{eq:Gnorm_bound1}
        \ppnorm{G}{T}
        &\leq
        \sqrt{2} \maxnorm{G_{0,0}},
    \end{align}
    and 
    if 
    $T\leq\frac{1}{c_{\pp,2} M_V(1+\maxnorm{G_{0,0}})^2}$
    then
    \begin{align}
        \label{eq:Gnorm_bound2}
        \ppnorm{G}{T}
        &\leq
        \frac{2}{\sqrt{3}}
        \maxnorm{G_{0,0}}.
    \end{align}
\end{proposition}

\begin{proof}
    By Lemma \ref{lem:J_estim}, we have that 
    \begin{align}
        \|\mathcal{J_\tau}[G,G]_{t_1,t_2}\|_{\max}
        \leq
        \frac{c_{\pp,2} M_V}{4}
        \int_0^{t_2}
        \rmd t
        \|
        G_{t,t}
        \|_{\max}
        \|
        G_{t_1,t}
        \|_{\max}.
    \end{align}
    If $G$ satisfies the evolution equation \eqref{eq:Gevo}, then 
    \begin{align*}
        \|
        G_{t_1,t_2}
        \|_{\max}
        &\leq
        \|
        G_{0,0}
        \|_{\max}
        +
        \frac{c_{\pp,2} M_V}{8}
        \int_0^{t_2}
        \rmd t
        \Big(
        \|
        G_{t,t}
        \|_{\max}
        \|
        G_{0,t}
        \|_{\max}
        +
        \|
        G_{t,t}
        \|_{\max}
        \|
        G_{t_1,t}
        \|_{\max}
        \Big)
        \\
            &+
            \frac{c_{\pp,2} M_V}{8}
        \int_0^{t_1}
        \rmd t
        \Big(
        \|
        G_{t,t}
        \|_{\max}
        \|
        G_{0,t}
        \|_{\max}
        +
        \|
        G_{t,t}
        \|_{\max}
        \|
        G_{t_2,t}
        \|_{\max}
        \Big)
    \end{align*}
    Let $g(T)\coloneqq \sup_{t_1,t_2\in[0,T]} \| G_{t_1,t_2} \|_{\max}$. Then we get from 
    above that 
    \begin{align}
        \nonumber
        g(T) &\leq
        g(0)
        +
        \frac{c_{\pp,2} M_V}{4}
        \int_0^{T}\rmd t g(t)^2
        +
        \frac{c_{\pp,2} M_V}{4}
        g(T)\int_0^{T}\rmd t g(t)
    \\
    \label{eq:lem-gestim}
             &\leq
        g(0)
        +
        \frac{c_{\pp,2} M_V}{2}
        g(T)\int_0^{T}\rmd t g(t)
    \end{align}

    Integrating \eqref{eq:lem-gestim}, we get the following quadratic
    inequality for $h(s)\coloneqq \int_0^s\rmd Tg(T)$
    \begin{align*}
        h(s)\leq g(0)s+\frac{c_{\pp,2} M_V}{2} \int_0^s\rmd T g(T)\int_0^T\rmd tg(t) 
        = g(0)s+\frac{c_{\pp,2} M_V}{4}h(s)^2,
    \end{align*}
    where we used Fubini to compute $\int_0^s\rmd T g(T)\int_0^T\rmd tg(t)=\int_0^s\rmd tg(t)\int_t^s\rmd Tg(T)$.
    This means that for $s$ satisfying $0\leq s\leq \frac{1}{c_{\pp,2} M_Vg(0)}$, we have either
    \begin{align*}
        h(s)\leq \frac{2}{c_{\pp,2} M_V}\left(1-\sqrt{1-c_{\pp,2} M_Vg(0)s}\right)=:r_-(s),
    \end{align*}
    or
    \begin{align*}
        h(s)\ge \frac{2}{c_{\pp,2} M_V}\left(1+\sqrt{1-c_{\pp,2} M_Vg(0)s}\right)=:r_+(s).
    \end{align*}
    From the fact that $h(0)=r_-(0)=0$ together with continuity of $h$ and that $r_+(s)>r_-(s)\ge0$
    for all $0\leq s<\frac{1}{c_{\pp,2} M_Vg(0)}$, we deduce that $h(s)\leq r_-(s)$ for all 
    $0\leq s\leq \frac{1}{c_{\pp,2} M_Vg(0)}$. Using this with \eqref{eq:lem-gestim}
    gives
    \begin{align}
        g(T) &\leq
    \frac{g(0)}{\sqrt{1-c_{\pp,2} M_Vg(0)T}}
    \end{align}
    for $T<\frac{1}{c_{\pp,2} M_Vg(0)}$. This proves \eqref{eq:Gnorm_bound0}, and hence 
    \eqref{eq:Gnorm_bound1} follows. Finally, \eqref{eq:Gnorm_bound2} follows from 
    the fact that $\sup_{x>0} x/(1+x)^2=1/4$.
\end{proof}

\begin{proposition}
    \label{prop:Gfixed_point} 
    Take $d\ge3$ and $\pp\in(0,1-2/d)$.
    Take $G_\text{in}\in \Cmax$. Let
    \[
        \tmax<\frac{1}{c_{\pp,2} M_V(1+\maxnorm{\gin})^2}.
    \]
    Suppose $\tau(T,t)=T-t$ and
    $L\ge L_\pp(\lambda,\tmax)$, or $\tau(T,t)=T_0$ for some $T_0>0$ and $L\ge
    L_\pp(\lambda,T_0)$. Then there exists a  unique solution
    $G\in\CT{\tmax}$ to \eqref{eq:Gevo} with $G_{0,0}=G_\text{in}$.
\end{proposition}

\begin{proof}
    Let 
    $\mathscr{B}=\{G\in\CT{\tmax}\ : \ \ppnorm{G-\gin}{\tmax}\leq1, \ G_{0,0}=\gin \}$, where in 
    $\ppnorm{G-\gin}{\tmax}$ we have identified $\gin$ with its constant extension, i.e. the function  $g\in\CT{\tmax}$ that satisfies $g_{t_1,t_2}=\gin$ for all $(t_1,t_2)\in[0,\tmax]^2$.
    Note that $\mathscr{B}$ is a complete metric space, since it is a closed subset 
    of the Banach space $\CT{\tmax}$. We start by proving the existence and uniqueness 
    of a solution on $\mathscr{B}$.
    This is going to be a standard Banach fixed--point argument. 

    Note that $\ppnorm{G}{\tmax}\leq 1
    +\maxnorm{\gin}$ for each $G\in \mathscr{B}$. Hence, Lemma \ref{lem:J_estim} implies
    \\
    $\ppnorm{\mathcal{J}_\tau[G,G]-\mathcal{J}_\tau[\widetilde{G},\widetilde{G}]}{\tmax} 
    \leq \frac{c_{\pp,2}}{2} M_V(1+\maxnorm{\gin})\tmax\ppnorm{G-\widetilde{G}}{\tmax}$ for $G,\widetilde{G}\in 
    \mathscr{B}$. Then, triangle inequality implies 
    \begin{align}
        \label{eq:Gcontraction}
        \ppnorm{\mathcal{I}_{\tau}[G,G]-\mathcal{I}_{{\tau}}[\widetilde{G},\widetilde{G}]}{\tmax}
        &\leq
        c_{\pp,2} M_V(1+\maxnorm{\gin})\tmax
        \ppnorm{G-\widetilde{G}}{\tmax}
    \end{align}
    for all $G,\widetilde{G}\in \mathscr{B}$. This proves that $\mathcal{I}_\tau$ is a
    contraction since $\tmax<
    \frac{1}{c_{\pp,2} M_V(1+\maxnorm{\gin})^2}$.

    Definition of $\mathcal{I}_\tau$, together with the triangle inequality
    and Lemma \ref{lem:J_estim} imply
    \begin{align}
        \label{eq:Guniqueness}
    \ppnorm{\mathcal{I}_\tau[G,G]-\gin}{\tmax}
    \leq \frac{1}{2}4\ppnorm{\mathcal{J}_\tau[G,G]}{\tmax}
    \leq \frac{c_{\pp,2}}{2} M_V\ppnorm{G}{\tmax}^2\tmax
    <
    \frac{\ppnorm{G}{\tmax}^2}{2(1+\maxnorm{\gin})^2}
    \end{align}
    for all $G\in\CT{\tmax}$ satisfying $G_{0,0}=\gin$.
    Hence, for all $G\in \mathscr{B}$ there holds
    $\ppnorm{\mathcal{I}_\tau[G,G]-\gin}{\tmax}\leq 1$.  This proves that
    $\mathcal{I}_\tau$ maps $\mathscr{B}$ to itself, since
    $\mathcal{I}_\tau[G,G]_{0,0}=G_{0,0}$ by definition of $\mathcal{I}_\tau$. Thus we
    get from Banach fixed point theorem that on $\mathscr{B}$ there exists a
    unique solution
    to \eqref{eq:Gevo} with the initial condition $\gin$.
    
    Suppose that $G\in\CT{\tmax}$, is a solution to \eqref{eq:Gevo} with 
    $G_{0,0}=\gin$. Then $G=\mathcal{I}_\tau[G,G]$, and so, \eqref{eq:Guniqueness}
    together with Proposition \ref{prop:Ggronwall}
    imply
    \begin{align*}
        \ppnorm{G-\gin}{\tmax}
        &<
    \frac{1}{2(1+\maxnorm{\gin})^2}
    \frac{\maxnorm{\gin}^2}{1-\maxnorm{\gin}\frac{1}{(1+\maxnorm{\gin})^2}}
        \leq 1.
    \end{align*}
    Hence, $G\in \mathscr{B}$. This proves uniqueness on the whole space 
    $\CT{\tmax}$.
\end{proof}

The following Proposition \eqref{prop:Gincontinuity} finishes the well-posedness 
of the fixed point equation \eqref{eq:Gevo} by proving that its solutions are 
continuous in  the initial data with respect to the $\maxnorm{\cdot}$ norm.
\begin{proposition}
     \label{prop:Gincontinuity} 
    Take $d\ge3$ and $\pp\in(0,1-2/d)$.
     Let $\lambda>0$ and $\tmax>0$. Suppose that $\tau\colon\R_+^2\to\R_+$ 
    satisfy one of the conditions \eqref{eq:taudiff} or \eqref{eq:tauconst}. Let $L\ge
    L_\pp(\lambda,\tau(\tmax,0))$. 
    Suppose  $G,\widetilde{G}\in\CT{\tmax}$ are solutions to \eqref{eq:Gevo}.
    Let $T\in[0,\tmax]$. Then 
      \begin{align}
         \label{eq:Gincontinuity1}
         \sup_{t_1,t_2\in [0,T]}
         \maxnorm{G_{t_1,t_2}-\widetilde{G}_{t_1,t_2}}
         \leq 
         \frac{
         \maxnorm{G_{0,0}-\widetilde{G}_{0,0}}}
         {1-(\frac{1}{2}c_{\pp,2}M_V \ppnorm{G}{T})T}
         \rme^{\frac{T\ppnorm{\widetilde{G}}{T}}{1-(\frac{1}{2}c_{\pp,2}M_V \ppnorm{G}{T})T}}
     \end{align}
     whenever $(\frac{1}{2}c_{\pp,2} M_V\ppnorm{G}{T})T<1$, e.g. for $T< \frac{2}{c_{\pp,2} M_V\ppnorm{G}{\tmax}}$. In 
     particular, if $T\leq\frac{1}{c_{\pp,2} M_V(1+\maxnorm{G_{0,0}})^2}$, then
     \begin{align}
         \label{eq:Gincontinuity2}
         \ppnorm{G-\widetilde{G}}{T} \leq 
         \frac{6}{5}\maxnorm{G_{0,0}-\widetilde{G}_{0,0}}
         \rme^{\frac{6}{5}T\ppnorm{\widetilde{G}}{T}}.
     \end{align}
 \end{proposition}

 \begin{proof}
     Suppose $(\frac{1}{2}c_{\pp,2} M_V\ppnorm{G}{T})T<1$.
     By the definition \eqref{eq:I} of $\mathcal{I}$ and Lemma \ref{lem:J_estim} we have 
     that 
     \begin{align*}
         \ppnorm{G-\widetilde{G}}{T}
         &\leq
         \maxnorm{G_{0,0}-\widetilde{G}_{0,0}}
         +\frac{1}{2}c_{\pp,2}M_V 
         \left( 
             \ppnorm{\widetilde{G}}{T}
             \int_0^T \rmd t 
             \ppnorm{G-\widetilde{G}}{t}
             +
             T
             \ppnorm{{G}}{T}
             \ppnorm{G-\widetilde{G}}{T}
         \right).
     \end{align*}
     Solving for $\ppnorm{G-\widetilde{G}}{T}$ implies 
     \begin{align}
         \ppnorm{G-\widetilde{G}}{T}
         \leq 
         \frac{1}{1-(\frac{1}{2}c_{\pp,2}M_V \ppnorm{G}{T})T}
         \left( 
         \maxnorm{G_{0,0}-\widetilde{G}_{0,0}}
         +
             \ppnorm{\widetilde{G}}{T}
             \int_0^T \rmd t 
             \ppnorm{G-\widetilde{G}}{t}
         \right).
     \end{align}
     By applying linear Gr\"onwall's inequality, we obtain \eqref{eq:Gincontinuity1}.
     Hence, the bound \eqref{eq:Gincontinuity2} follows from
     the upper bound \eqref{eq:Gnorm_bound2} together with 
     \begin{align*}
         \frac{\frac{1}{2}c_{\pp,2} M_V\frac{2}{\sqrt{3}}\maxnorm{G_{0,0}}}{c_{\pp,2} M_V(1+\maxnorm{G_{0,0}})^2}
         \leq 
         \frac{1}{6}.
     \end{align*}
 \end{proof}

\begin{proposition}
    \label{prop:HGfixed_point} 
    Take $d\ge3$ and $\pp\in(0,1-2/d)$.
    Take $\mathcal{G}_\text{in}\in \Cmax$ and $G\in\CT{\tmax}$. Let
    \[
        \tmax<\frac{2}{c_{\pp,2} M_V\ppnorm{G}{\tmax}}.
    \]
    Suppose $\tau(T,t)=T-t$
    and $L\ge L_\pp(\lambda,\tmax)$, or $\tau(T,t)=T_0$ for some $T_0>0$ and
    $L\ge
    L_\pp(\lambda,T_0)$. Then there exists a unique 
    $\mathcal{G}\in\CT{\tmax}$, $\mathcal{G}_{0,0}=\mathcal{G}_\text{in}$ satisfying $\mathcal{G}=\mathcal{I}[G,\mathcal{G}]$.
\end{proposition}

\begin{proof}
    Let $\mathscr{B}=\{ \mathcal{G}\in\CT{\tmax}\ : \  \mathcal{G}_{0,0}=\mathcal{G}_\text{in}\}$. Note that
    $\mathscr{B}$ is a complete metric space as a closed subset of the Banach space 
    $\CT{\tmax}$.  
    This is going to be a standard Banach fixed point argument. 
    
    By definition of $\mathscr{B}$ and $\mathcal{I}$, we have that $\mathcal{G}\mapsto \mathcal{I}[G,\mathcal{G}]$
    maps $\mathscr{B}$ to itself.

    Lemma \ref{lem:J_estim} implies
    $\ppnorm{\mathcal{J}_\tau[G,\mathcal{G}]-\mathcal{J}_\tau[{G},\widetilde{\mathcal{G}}]}{\tmax} 
    \leq\frac{1}{4} c_{\pp,2}
    M_V\ppnorm{G}{\tmax}\tmax\ppnorm{\mathcal{G}-\widetilde{\mathcal{G}}}{\tmax}$
    for $\mathcal{G},\widetilde{\mathcal{G}}\in 
    \mathscr{B}$. After this, the triangle inequality implies 
    \begin{align}
        \label{eq:Hcontraction}
        \ppnorm{\mathcal{I}_\tau[G,\mathcal{G}]-\mathcal{I}_\tau[{G},\widetilde{\mathcal{G}}]}{\tmax}
        &\leq
        \frac{1}{2}c_{\pp,2} M_V\ppnorm{G}{\tmax}\tmax
        \ppnorm{\mathcal{G}-\widetilde{\mathcal{G}}}{\tmax}
    \end{align}
    for all $\mathcal{G},\widetilde{\mathcal{G}}\in \mathscr{B}$. This proves that $\mathcal{I}[G,\cdot]$ is a
    contraction, since $\tmax<
    \frac{2}{c_{\pp,2} M_V\ppnorm{G}{\tmax}}$.
\end{proof}

\section{Proofs of the main results}
\label{sec:proofs}

Equipped with the Lemmas and Propositions from Sections \ref{sec:lemmata}--\ref{sec:fixpoint_eqs}, we are ready to prove the main results of the article.

\subsection{Proof of the well-posedness result}

We start by finding a candidate solution by studying the evolution equation 
\eqref{eq:WGevo}.

\begin{proposition}
    \label{prop:WGgronwall}
    Take $d\ge3$ and $\pp\in(0,1-2/d)$.
    Take $\tmax>0$. 
    Suppose $\tau(T,t)=T-t$ and
    $L\ge L_\pp(\lambda,\tmax)$, or $\tau(T,t)=T_0$ for some $T_0>0$ and $L\ge
    L_\pp(\lambda,T_0)$. 
    If the functions $W,\widetilde{W}\in
    C([0,\tmax]\times\duallattice,  \R)$ satisfy the evolution equation \eqref{eq:WGevo},
    respectively,
    with the functions $G, \widetilde{G}\in \CT{\tmax}$, then 
    \begin{align}
        \subnorm{W_T-\widetilde{W}_T}{\infty}
        &\leq
        \left( 
        \subnorm{W_0-\widetilde{W}_0}{\infty}
        +
        c_{3}
        \sobolevnorm{\FT{V}}{\frac{1}{6}}^2
        \int_0^T
        \rmd t
        \subnorm{\widetilde{W}_t}{\infty}
        \ppnorm{G-\widetilde{G}}{t}
        \right) 
        \rme^{c_{3}
        \sobolevnorm{\FT{V}}{\frac{1}{6}}^2
        \int_0^T\rmd t \maxnorm{G_{t,t}}}
    \end{align}
    for all $T\in [0,\tmax]  $. In particular,  then also
    \begin{align}
        \subnorm{W_T}{\infty}
        &\leq
        \subnorm{W_0}{\infty}
        \rme^{c_{3}
        \sobolevnorm{\FT{V}}{\frac{1}{6}}^2
        \int_0^T\rmd t \maxnorm{G_{t,t}}}.
    \end{align}
\end{proposition}

\begin{proof}
    The evolution equation \eqref{eq:WGevo}, the inequality $\lambda^{-2}\tau(T,t)\leq
    L^{\pp}/4$, Corollary \ref{cor:1Phi_int} for $x=y+y'$, Corollary \ref{cor:modregabs} 
    and Lemma \eqref{lem:ineqVconst} together imply
    \begin{align}
        &|W_T(k_0)-\widetilde{W}_T(k_0)|
        \leq
        |W_0(k_0)-\widetilde{W}_0(k_0)|
        \\
        &+c_{3}
        \sobolevnorm{\FT{V}}{\frac{1}{6}}^2
        \int_0^T\rmd t 
        \left( 
            \subnorm{W_t-\widetilde{W}_t}{\infty}
         \maxnorm{G_{t,t}}
         +
         \subnorm{\widetilde{W}_t}{\infty}
        \maxnorm{G_{t,t}-\widetilde{G}_{t,t}}.
        \right)
    \end{align}
    Recall that $c_{3}= c_{6}2^{\frac{d}{6}}
    (2^{\frac{d}{6}}+3)$. Taking supremum over $k_0\in\duallattice$ and then
    using the standard linear version of Gr\"onwall's lemma leads to the desired result. Note that the bound
    for $\subnorm{W_T}{\infty}$ follows automatically, as one can take
    $\widetilde{W}=0$ and $\widetilde{G}=0$.
\end{proof}

\begin{proposition}
    \label{prop:WGfixed_point}
    Take $d\ge3$ and $\pp\in(0,1-2/d)$.
    Take $\win:\duallattice\to\R$ and\\ 
    $G\in\CT{\tmax}$. Let 
    \[
        \tmax<
        \frac{1}{c_{3}
        \sobolevnorm{\FT{V}}{\frac{1}{6}}^2
        \ppnorm{G}{\tmax}}.
    \]
    Suppose $\tau(T,t)=T-t$ and
    $L\ge L_\pp(\lambda,\tmax)$, or $\tau(T,t)=T_0$ for some $T_0>0$ and $L\ge
    L_\pp(\lambda,T_0)$. 
      Then there exists a unique solution $W\in
    C([0,\tmax]\times\duallattice,  \C)$ to the equation
    \eqref{eq:WGevo} with $G$ and $W_0=W_\text{in}$.
\end{proposition}
\begin{proof}
    Let $\mathscr{B}=\{W\in C([0,\tmax]\times\duallattice,\C) \ :
    \ W_0=\win \}$. Note that
    $\mathscr{B}$ is a complete metric space as a closed subset of the Banach space
    $C([0,\tmax]\times\duallattice,\C)$ equipped with the norm $
    \sup_{T\in[0,\tmax]}\subnorm{W_T}{\infty}$. We want to prove existence and
    uniqueness on $\mathscr{B}$.  This is going to be a standard Banach fixed point
    argument. 

    Let $\mathcal{Q}: C([0,\tmax]\times\duallattice,\C) \to
    C([0,\tmax]\times\duallattice,\C)$ be the map defined by the right hand side 
    of the evolution equation \eqref{eq:WGevo}. By definition $\mathcal{Q}$
    maps $\mathscr{B}$ to itself. A similar argument as in the proof of Proposition
    \ref{prop:WGgronwall} gives 
         \begin{align}
         \sup_{T\in[0,\tmax]}\subnorm{\mathcal{Q}[W]_T-\mathcal{Q}[W']_T}{\infty}
        &\leq
        c_{3}
        \sobolevnorm{\FT{V}}{\frac{1}{6}}^2
        \ppnorm{G}{\tmax}
        \tmax
         \sup_{T\in[0,\tmax]}
          \subnorm{W_T-W'_T}{\infty}
     \end{align}
     for all $W,W'\in \mathscr{B}$.
     This proves that $\mathcal{Q}$ is a contraction on $\mathscr{B}$, since $\tmax <
     \frac{1}{c_{3}\sobolevnorm{\FT{V}}{\frac{1}{6}}\ppnorm{G}{\tmax}}$.
     Thus the Banach fixed point theorem gives the desired result.
\end{proof}

We are now ready to prove the well-posedness Theorem \ref{thm:Wwell_posedness}.

\begin{proofof}{Theorem \ref{thm:Wwell_posedness}:}
    
    Let $c_{\pp,1} = \max(c_{\pp,2}, \frac{\sqrt{3}}{6}c_{3})$.
    Recall that $M_V=\mv$ and  let 
    \[
        0\leq\tmax< \frac{1}{c_{\pp,1}
        M_V (1+\maxnorm{F[\win]})^2}.
    \]

    Then by Proposition \ref{prop:Gfixed_point} there exists a  unique solution
    $G\in\CT{\tmax}$ to \eqref{eq:Gevo} with $G_{0,0}=F[\win]$.
     This means that $G=\mathcal{I}[G,G]$.

      By Proposition 
     \ref{prop:Ggronwall} we have 
     $\ppnorm{G}{\tmax}\leq\frac{2}{\sqrt{3}}\maxnorm{F[\win]}$. Note that $\sobolevnorm{\FT{V}}{\frac{1}{6}}^2\leq M_V$ and by
     Proposition \ref{prop:WGfixed_point}, there exists a unique solution $W\in
     C([0,\tmax]\times\duallattice,  \C)$ to the equation
    \eqref{eq:WGevo} with $G$ and $W_0=W_\text{in}$. 

    By Lemma \ref{lem:Fevo}, there holds $F[W]=\mathcal{I}[G,F[W]]$. Since
    $G=\mathcal{I}[G,G]$ and $G_{0,0}=F[W]_{0,0}$,
    we have by Proposition \ref{prop:HGfixed_point} that $F[W]=G$. Hence, Lemma
    \ref{lem:colG} implies that $W$ is a unique solution to the evolution equation \eqref{eq:Wevo}.
    
    Since $\cbos(W,\tau_0)^*=\cbos(W^*,\tau_0)$ for all $\tau_0>0$ and $\win^*=\win$, we have that also 
    $W^*$ is a solution to the evolution equation \eqref{eq:Wevo}. By uniqueness, $W=W^*$. Hence, 
    $W$ is real valued, namely, $W\in C([0,\tmax]\times\duallattice,  \R)$.

    The upper bounds \eqref{eq:Wsupbound} and \eqref{eq:Fnormpropagation} follow 
    from Propositions \ref{prop:WGgronwall} and \ref{prop:Ggronwall}, respectively.
    
    From \eqref{eq:FisFTofW} and the definition of the weight function $\Phi$, we get 
    that $\supsobolevnorm{W_T}{1/3}\leq \phinorm{F[W]^{(1)}_{0,T}}$. By Proposition
    \ref{prop:Ggronwall} we have $\ppnorm{F[W]}{T}\leq\frac{2}{\sqrt{3}}\maxnorm{F[\win]}$
    Hence, the bound \eqref{eq:Wsobolevbound} holds.

     It remains to prove that the solution $W$ depends continuously on the
     initial data $\win$. Suppose that $\widetilde{W}\in
     C([0,\tmax]\times\duallattice,  \C)$ is another solution to the equation
     \eqref{eq:WGevo}. Then Proposition \ref{prop:Gincontinuity} implies 
     \begin{align*}
         \ppnorm{F[W]-F[\widetilde{W}]}{T} \leq 
         \frac{6}{5}\maxnorm{F[W]_{0,0}-F[\widetilde{W}]_{0,0}}
         \rme^{\frac{6}{5}T\ppnorm{F[\widetilde{W}]}{T}}.
     \end{align*}
     By the definition of $F$ and $\maxnorm{\cdot}$, we have that 
     \begin{align*}
         \maxnorm{F[W]_{0,0}-F[\widetilde{W}]_{0,0}}
         &\leq
         \sup_{R\in\R^d,x\in\lattice} \Phi^d(R,x)
         \max 
         \left( 
             1,\subnorm{W_0}{\infty}+\subnorm{\widetilde{W}_0}{\infty}
         \right) 
         \subnorm{W_0-\widetilde{W}_0}{\infty}.
     \end{align*}
     Hence, from $\Phi^d(R,x)\leq \mean{L/2}^{d/3}$ and Proposition 
     \ref{prop:WGgronwall} we obtain that 
     \begin{align}
         \subnorm{W_T-\widetilde{W}_T}{\infty}
         \leq 
         C
        \mean{L/2}^{d/3}
         \subnorm{W_0-\widetilde{W}_0}{\infty}
     \end{align}
     for all $T\in[0,\tmax]$, where the finite constant $C>0$ is independent of 
     $T$, $L$ and $\lambda$, but may depend on $\pp$, the dimension $d$, the final kinetic time $\tmax$, and 
     the norms $\sobolevnorm{\FT{V}}{\frac{1}{6}}$, $\maxnorm{F[W]_{0,0}}$, 
     $\maxnorm{F[\widetilde{W}]_{0,0}}$, $\subnorm{W_0}{\infty}$,
     and $\subnorm{\widetilde{W}_0}{\infty}$. Thus, the solution is continuous on the 
     initial data.
\end{proofof}

\subsection{Proof of the error result}
The following proposition will be needed for the error estimate in Theorem \ref{thm:error}.
\begin{proposition}
     \label{prop:Gerror} 
    Take $d\ge3$ and $\pp\in(0,1-2/d)$.
     Let $\lambda>0$ and $\tmax>0$. Suppose that both $\tau,\widetilde{\tau}\colon\R_+^2\to\R_+$ 
    satisfy one of the conditions \eqref{eq:taudiff} or \eqref{eq:tauconst}. Let $L\ge
    L_\pp(\lambda,\max(\tau,\widetilde{\tau})(\tmax,0))$.
     Let $\gin\in\Cmax$ and suppose 
     \[
         \tmax<\frac{1}{2c_{\pp,2}
         M_V\maxnorm{\gin}}.
     \]
     Take $G,\widetilde{G}\in\CT{\tmax}$,
     with $G_{0,0}=\widetilde{G}_{0,0}=\gin$.  If $G=\mathcal{I}_\tau[G,G]$ and
     $\widetilde{G}=\mathcal{I}_{\widetilde{\tau}}[\widetilde{G},\widetilde{G}]$, then
     \begin{align}
         \ppnorm{G-\widetilde{G}}{T} \leq c_{\pp,4} M_V \ppnorm{\widetilde{G}}{T}^2
         \mathcal{E}_\pp(\lambda,\tau,\widetilde{\tau},T)
     \end{align}
     for all $T\in[0,\tmax]$. Here, the error term $\mathcal{E}_\pp$ was defined in 
     Definition \ref{def:error}.
 \end{proposition}

 \begin{proof}
     Lemma \ref{lem:J_estim} together with \eqref{eq:Gnorm_bound1} implies 
     \begin{align}
         \ppnorm{G-\widetilde{G}}{\tmax} 
         &\leq
         \frac{\sqrt{2}}{2}c_{\pp,2} M_V\tmax\maxnorm{\gin}\ppnorm{G-\widetilde{G}}{\tmax} 
         +\frac{1}{2}c_{\pp,4}M_V\ppnorm{\widetilde{G}}{\tmax}^2\mathcal{E}_\pp(\lambda,\tau,\widetilde{\tau},T).
     \end{align}
     Hence, the result follows from $T\leq\tmax<\frac{1}{2c_{\pp,2} M_V\maxnorm{\gin}}$ and $2\sqrt{2}/8\leq1/2$.
 \end{proof}

 \begin{proofof}{Theorem \ref{thm:error}}
     To make the 
    following estimate more readable, we use the shorthand notations
    \begin{align}
        F[W](1) = 
        F[W]_{t,t}^{(\widetilde{n}_i)}(\rs{\xi_i}(k),k,\widetilde{u}_{\xi_i},\sigma_i^{(2)},y+y')
        , \quad 
        W(2) = 
        W_t(k_0-\ind(i=4,5)k).
    \end{align}
    By Lemma \ref{lem:colG}, we have 
    \begin{align*}
        &\left| W_T(k_0)-\widetilde{W}(k_0)_T \right|  
        \leq 
        \int_0^T \rmd t
    \sum_{i=1}^8  
    \int_{\duallattice} \rmd k
           \sum_{y,y'\in\lattice}^{} 
           \sup_{k'\in\duallattice}|v_i(k',y)||w_i(y')|
           \\
    &
    \times
    \Bigg(
    \int_{-\lambda^{-2}\tau(T,t)}^{\lambda^{-2}\tau(T,t)}ds
    \Big|
    F[W](1)W(2)-F[\widetilde{W}](1)\widetilde{W}(2)
    \Big|
    +
    \sum_{\sigma\in\{\pm\}}
    \int_{\sigma\lambda^{-2}\widetilde\tau(T,t)}^{\sigma\lambda^{-2}\tau(T,t)}|\rmd s|
    \left| F[\widetilde{W}](1)\widetilde{W}(2) \right| 
    \Bigg).
    \end{align*}
After estimating 
$$|F[W](1)W(2)-F[\widetilde{W}](1)\widetilde{W}(2)  | \leq
|F[W](1)||W(2)-\widetilde{W}(2)| +|\widetilde{W}(2)|
|F[\widetilde{W}](1)-F[{W}](1)|,$$ 
and writing in each term $1=\Phi^d(\rs{\xi_i}(k),y+y')/\Phi^d(\rs{\xi_i}(k),y+y')$,
we obtain 
\begin{align}
        \left| W_T(k_0)-\widetilde{W}(k_0)_T \right|  
        &\leq 
        8c_{3}
        \sobolevnorm{\FT{V}}{\frac{1}{6}}^2
        \ppnorm{F[W]}{T}
        \int_0^T \rmd t
        \subnorm{W_t-\widetilde{W}_t}{\infty}
        \\
        &+
        c_{\pp,4}c_{5}c_{6}
        \sobolevnorm{\FT{V}}{\frac{1}{6}}^2
        M_V
        T
        \ppnorm{F[\widetilde{W}]}{T}^2
        \mathcal{E}_\pp(\lambda,\tau,\widetilde\tau,T)
        \sup_{T'\in[0,T]} \subnorm{\widetilde{W}_{T'}}{\infty}  
        \\
        &+
        c_{5}
        \sobolevnorm{\FT{V}}{\frac{1}{6}}^2
        \mathcal{E}_\pp(\lambda,\tau,\widetilde\tau,T)
        \ppnorm{F[\widetilde{W}]}{T}
        \sup_{T'\in[0,T]} \subnorm{\widetilde{W}_{T'}}{\infty}.
\end{align}
For the bound on the $|F[W](1)||W(2)-\widetilde{W}(2)|$ term, we used 
Corollary \ref{cor:1Phi_int} and Lemma \ref{lem:ineqVconst}. In order to bound the 
$|\widetilde{W}(2)|
|F[\widetilde{W}](1)-F[{W}](1)|$ term, we used Corollary \ref{cor:1Phi_int}, Lemma \ref{lem:ineqVconst}
and Proposition \ref{prop:Gerror}. For the $|F[\widetilde{W}](1)\widetilde{W}(2)|$ term 
we used Lemmas \ref{lem:int_over_weight} and \ref{lem:ineqVconst}. Also, Corollary 
\ref{cor:modregabs} was used for each term to estimate
$\mean{y_j+y'_j}\leq2\mean{y_j}\mean{y'_j}$ before applying Lemma \ref{lem:ineqVconst}.
After taking the supremum over $k_0$ and applying linear
Gr\"onwall's inequality, we obtain the result.

\end{proofof}

\bigskip

\noindent \textbf{Acknowledgements}

The authors thank Phan Thành Nam, Herbert Spohn,  Gigliola Staffilani, and Minh Binh Tran for helpful discussions relevant to the project.
The work has been supported by the Academy of Finland, via an Academy project (project No. 339228) and the {\em Finnish Centre of Excellence in Randomness and Structures\/} (project Nos. 346306 and 364213).
 
The funders had no role in study design, analysis, decision to publish, or preparation of the manuscript.

\smallskip

\noindent\textbf{Compliance with ethical standards} 
\smallskip

\noindent \textbf{Data availability statement} Data sharing not applicable to this article as no datasets
were generated or analysed during the current study.
\smallskip

\noindent \textbf{Conflict of interest} The authors declare that they have no conflict of interest.

\appendix

\section{Appendix}

\subsection{Derivation of the truncated hierarchy for DNLS}
\label{sec:derivtrunchier}

We outline in this Appendix the derivation of the evolution equation  \eqref{eq:Wevo}
from the truncated hierarchy of the DNLS model.  This topic was already covered
in \cite{lukkarinen_wick_2016} 
for the special case of onsite
interaction potential, but since the more general $V$ will change the outcome,
we will present a streamlined version of the computations here, also putting it
more clearly in the context of iterated Duhamel expansions.

For notational simplicity, we now drop the Fourier-transform ``hat'' from the
notations, i.e., in this Section we denote simply $a_t(k,\sigma)$ instead of
$\FT{a}_t(k,\sigma)$.
As the original lattice field will not appear anymore, this should not cause  confusion.
Also, we will use letter $b$ to denote the pair of arguments of this field, i.e., $b=(k,\sigma)$, while given
an index sequence of such pairs, $I=(b_\ell)_{\ell=1}^n$, we will implicitly
use the notation $k_\ell$ for the first component of $b_\ell$ and $\sigma_\ell$
for the second component of $b_\ell$.  Given some position index $\ell$ of $I$,
the notation $I\setminus \ell$ denotes the sequence where the element at
position $\ell$ has been cancelled, i.e., $I\setminus \ell = (b_j)_{j\ne
\ell}$.  We will also freely continue to make use of the notations introduced
in the main text, mainly coming from
\cite{lukkarinen_wick_2016}. 
With these conventions, we may rewrite the evolution equation (\ref{eq:FTa_evol})
in the form
\begin{align*}
    \partial_t &a _t(k,\sigma)=-\ci\sigma\omega(k)a_t(k,\sigma)
    \\
    &
    -\ci \sigma\lambda
    \int_{(\duallattice)^3}\rmd k'_1 \rmd k'_2 \rmd k'_3
    \delta_L(k-k'_1-k'_2-k'_3)
    \FT{V}(k'_1,k'_2,k'_3;\sigma)
    a_t(b'_1)a_t(b'_2)a_t(b'_3)\,,
\end{align*}
where the sign choices are $\sigma'_1=-$, $\sigma'_2=\sigma$, $\sigma'_3=+$,
and thus depend on the input $\sigma$.

By the results in \cite{lukkarinen_wick_2016}
,
the cumulants of a deterministic evolution satisfy in general
\[
 \partial_t \kappa[(a_t)_I] = \sum_{\ell=1}^n
 \E[\dot{a}_t(k_\ell,\sigma_\ell)\, \wick{(a_t)^{I\setminus \ell}}]\,.
\]
On the other hand, by the moment-to-cumulants property of Wick polynomials, here
\[
 \E[a_t(k_\ell,\sigma_\ell)\, \wick{(a_t)^{I\setminus \ell}}]
 = \kappa[(a_t)_I]\,,
\]
so the contribution from the first, linear terms arising from the
time-derivates leads to a term
\[
 -\ci\Omega(I) \kappa[(a_t)_I]\,, \qquad
 \Omega(I)\coloneqq \sum_{\ell=1}^n\sigma_\ell\omega(k_\ell) \,,
\]
in the rate function for the cumulant.

To compute the contribution from the second, nonlinear term to the rate
function, we use the moment-to-cumulants formula taking into account the
removal of partitions which have a cluster from indices inside a Wick
polynomial.  Therefore, for any three random variables $a'_j$, $j=1,2,3$, whose
odd moments are zero, we have
\begin{align}\label{eq:nonlinclustering}
& \E[a'_1 a'_2 a'_3 \, \wick{(a_t)^{I\setminus \ell}}]
 = \kappa[a'_1, a'_2, a'_3,(a_t)_{I\setminus \ell}]
 + \sum_{j=1}^3 \E[\prod_{i=1;i\ne j}^3 a'_i] \kappa[a'_j,(a_t)_{I\setminus \ell}] \nonumber \\
 & \qquad + \sum_{j=1}^3\sum_{\emptyset\ne A_1 \subsetneq I\setminus \ell}
  \kappa[a'_j,(a_t)_{A_1}] \kappa[(a'_i)_{i\ne j}, (a_t)_{I\setminus \ell\setminus A_1}]
  \nonumber \\
 & \qquad + \sum_{(A_1,A_2,A_3) \text{ ordered partition of } I\setminus \ell} \prod_{j=1}^3
  \kappa[a'_j,(a_t)_{A_j}]\,.
\end{align}
In particular, we may use this formula for $a'_j=\FT{a}_t(b'_j)$ which is
needed to compute the effect of the nonlinear part of the time-derivative.  We
may thus conclude that the cumulants satisfy an evolution hierarchy
of the form
\[
 \partial_t \kappa[(a_t)_I] = -\ci\Omega(I) \kappa[(a_t)_I]
 -\ci \lambda \mathcal{B}[\kappa_t](I)\,,
\]
where each $\mathcal{B}[\kappa_t](I)$ is a polynomial of the cumulants of the
fields $a_t$ which depends on the choice of the index sequence $I$ but is never
of an order greater than three and only depends on cumulants up to order $n+2$ if
$I$ has length $n$.  We now exponentiate the free evolution part to derive a
Duhamel formula, yielding a hierarchy of integral equations
\[
 \kappa[(a_t)_I] = \rme^{-\ci \Omega(I) t}\kappa[(a_0)_I]
 -\ci \lambda \int_0^t\!\rmd s\, \rme^{-\ci \Omega(I) (t-s)}
  \mathcal{B}[\kappa_s](I)\,.
\]
Similarly to Picard iteration, one can now look for solutions of the hierarchy
by iterating this equation.  For small enough $\lambda$, one
might expect the iterations to eventually converge for some kinetic time-interval, although this is far from
obvious.  However, a selective iteration of the formula is always possible, to
obtain new relations satisfied by the original hierarchy of cumulants.

What we call here a truncated hierarchy is obtained by the following procedure:
iterate the Duhamel formula once, and then assume that we can ignore the effect
to the second order cumulant of all cumulants of order higher than two.  One
can motivate this procedure by assuming that the state of the system remains
close enough to a Gaussian distribution so that the non-Gaussian effects yield
only error terms in the new evolution hierarchy.  Since the evolution does not
propagate Gaussianity, it is a highly non-trivial assumption, similar to the
propagation of chaos in standard kinetic theory. We do not attempt to make any
justification of the assumption here, apart from repeating what was said
earlier: the solutions to the truncated hierarchy can be used as reference
solutions for the covariance and this could facilitate the analysis of the
evolution of the remaining terms.

As explained in the main text, by translation and gauge invariance,
for $n=2$, it is sufficient to study the cumulant of the sequence
$I=(b_1,b_2)$ for $b_1=(-k,-)$.
$b_2=b=(k,+)$.  Now there are two terms in $ \mathcal{B}[\kappa_t](I)$
corresponding to $\ell=1$ and $\ell=2$.  For $\ell=1$, we have
$I\setminus\ell = ((k,+))$ while for $\ell=2$,
$I\setminus\ell = ((-k,-))$.  Hence, we can parametrize these two cases by the
sign of the removed term, i.e., by using as a summation variable
$\sigma\coloneqq \sigma_\ell$ instead of $\ell$, and then $I\setminus\ell = ((-\sigma
k, -\sigma))$.  Since $I\setminus\ell$ has only one element, only the first two
terms in (\ref{eq:nonlinclustering}) contribute to this case, yielding
\begin{align*}
& \E[a'_1 a'_2 a'_3 \, \wick{a_t(-\sigma k, -\sigma)}]
 = \kappa[a'_1, a'_2, a'_3,a_t(-\sigma k, -\sigma)]
 + \sum_{j=1}^3 \E[\prod_{i=1;i\ne j}^3 a'_i] \kappa[a'_j,a_t(-\sigma k, -\sigma)]\\
& \quad = \kappa[a'_1, a'_2, a'_3,a_t(-\sigma k, -\sigma)]
 +  \delta_L(k'_1+k'_3) \delta_L(k'_2-\sigma k)
 W_{\lambda^2 t}^{\lambda,L}(k'_3)
 W_{\lambda^2 t}^{\lambda,L}(k)\\
& \qquad
 +  \delta_L(k'_2+k'_i) \delta_L(k'_j-\sigma k)
 W_{\lambda^2 t}^{\lambda,L}(\sigma k'_2)
 W_{\lambda^2 t}^{\lambda,L}(k)
\,,
\end{align*}
where $i=1$, $j=3$, if $\sigma =+$, and $i=3$, $j=1$,  otherwise.

Since for this case $\Omega(I) = -\omega(-k) + \omega(k)=0$, by the assumed
symmetry of $\omega$, the lowest term in the Duhamel integrated hierarchy
satisfies
\begin{align*}
&
 \kappa[(a_t)_I] = \kappa[(a_0)_I]
 -\ci \lambda \int_0^t\!\rmd s\,
  \mathcal{B}[\kappa_s](I)\,,
\end{align*}
with
\begin{align*}
&\mathcal{B}[\kappa_s](I)
 = \sum_{\sigma = \pm} \sigma
 \int_{(\duallattice)^3}\rmd k_1 \rmd k_2 \rmd k_3 \delta_L(\sigma k-k_1-k_2-k_3)\FT{V}(k_1,k_2,k_3;\sigma)\\
 & \qquad \quad \times
     \kappa[
    a_s(k_1,-),a_s(k_2,\sigma),a_s(k_3,+),a_s(-\sigma k,-\sigma)]\\
 & \qquad
 + \sum_{\sigma = \pm} \sigma |\Lambda|
  W_{\lambda^2 s}^{\lambda,L}(k)
 \int_{\duallattice}\rmd k_3 \FT{V}(-k_3,\sigma k,k_3;\sigma)
 W_{\lambda^2 s}^{\lambda,L}(k_3)\\
 & \qquad
 +  |\Lambda|
  W_{\lambda^2 s}^{\lambda,L}(k)
 \int_{\duallattice}\rmd k_2 \left(
 \FT{V}(-k_2,k_2,k;+)
 W_{\lambda^2 s}^{\lambda,L}(k_2)
 -\FT{V}(-k,k_2,-k_2;-)
 W_{\lambda^2 s}^{\lambda,L}(-k_2)\right)\,.
\end{align*}
In the last line, both potential terms evaluate to $\FT{V}(0)$, and thus the
two integrals cancel each other out.  The second line simplifies to
\[
   |\Lambda|
  W_{\lambda^2 s}^{\lambda,L}(k)
  \int_{\duallattice}\rmd k_3 \left(\FT{V}(k-k_3)-\FT{V}(-k+k_3)\right)W_{\lambda^2 s}^{\lambda,L}(k_3) = 0\,,
\]
since $\FT{V}$ is symmetric.  Therefore,
\begin{align*}
&
W_{\lambda^2 t}^{\lambda,L}(k) = \frac{1}{|\Lambda|}\kappa[(a_t)_I] = W^{\lambda,L}_0(k)
 -\ci \lambda \int_0^t\!\rmd s\,\sum_{\sigma = \pm} \sigma
 \int_{(\duallattice)^3}\rmd k_1 \rmd k_2 \rmd k_3 \delta_L(\sigma k-k_1-k_2-k_3)\\
 & \qquad \quad \times
   \FT{V}(k_1,k_2,k_3;\sigma) \frac{1}{|\Lambda|} \kappa[
    a_s(k_1,-),a_s(k_2,\sigma),a_s(k_3,+),a_s(-\sigma k,-\sigma)]
 \,.
\end{align*}

Thus, to compute the effect of the first iteration, we need to check what
happens if $I=(b_i)_{i=1}^4$, with $b_1=(k_1,-)$, $b_2=(k_2,\sigma)$,
$b_3=(k_3,+)$, $b_4=(-\sigma k,-\sigma)$, for given $\sigma\in \{\pm\}$,
$k\in \Lambda^*$, and we may also assume that $k_1+k_2+k_3 =\sigma k$.
Considering (\ref{eq:nonlinclustering}) with $n=4$, only the last line produces
terms which do not contain higher order cumulants.  Therefore, for this $I$,
the first iteration amounts to using the identity
\begin{align*}
&
  \kappa[(a_t)_I] = (\text{higher order terms})
 -\ci \lambda \sum_{\ell=1}^4 \sigma_\ell \int_0^t\!\rmd s\,
  \rme^{-\ci \Omega(I) (t-s)}\nonumber \\
  &\qquad \times  \int_{(\duallattice)^3}\rmd k'_1 \rmd k'_2 \rmd k'_3
    \delta_L(k_\ell-k'_1-k'_2-k'_3)
   \FT{V}(k'_1,k'_2,k'_3;\sigma_\ell)
   \sum_{\pi\in S_3} \prod_{j=1; j\ne \ell}^4 \kappa[a_s(b'_j),
    a_s(b_{\pi(j)})]
 \,,
\end{align*}
where the sum of $\pi$ goes over the six possible permutations of
$\{1,2,3,4\}\setminus \{\ell\}$.  In fact, only those permutations with
matching signs, i.e., with $\sigma_{\pi(1)}=+$, $\sigma_{\pi(2)}=-\sigma_\ell$,
$\sigma_{\pi(3)}=-$, yield a non-zero contribution. We
call such permutations ``allowed''.
Similarly, the resulting
$\delta_L$-terms can only contribute if $k'_1=-k_{\pi(1)}$, $k'_2=-k_{\pi(2)}$,
$k'_3=-k_{\pi(3)}$, in which case their sum is equal to the missing term,
$k_\ell$.  This saturates the first $\delta_L$-factor.  Therefore,
\begin{align*}
&
 \frac{1}{|\Lambda|} \kappa[(a_t)_I] = (\text{higher order terms})\\
  &\qquad
 -\ci \lambda\int_0^t\!\rmd s\,
  \rme^{-\ci \Omega(I) (t-s)} \prod_{j=1}^4 W_{\lambda^2 s}^{\lambda,L}(\sigma_j k_j)
   \sum_{\ell=1}^4 \sum_{\text{allowed }\pi}
   \frac{\sigma_\ell \FT{V}(-k_{\pi(1)},-k_{\pi(2)},-k_{\pi(3)};\sigma_\ell)}{W_{\lambda^2 s}^{\lambda,L}(\sigma_\ell k_\ell) }
 \,,
\end{align*}
where we can swap the signs of the first three arguments of $\FT{V}$, due to
symmetry.  The remaining combinatorics need to be worked out by hand.  If
$\sigma=+$, we have $(\sigma_\ell)_\ell=(-,+,+,-)$, and $\sigma=-$, we have
$(\sigma_\ell)_\ell=(-,-,+,+)$.  Thus, for $\sigma=+$, we obtain \begin{align*}
&
  \sum_{\ell=1}^4 \sum_{\text{allowed }\pi}
   \frac{\sigma_\ell \FT{V}(k_{\pi(1)},k_{\pi(2)},k_{\pi(3)};\sigma_\ell)}{W_{\lambda^2 s}^{\lambda,L}(\sigma_\ell k_\ell) }
\\ & \quad
= -\frac{1}{W_{\lambda^2 s}^{\lambda,L}(-k_1) }
\left(\FT{V}(k_{2},k_{3},k_{4};-)+\FT{V}(k_{3},k_{2},k_{4};-)\right)
\\ & \qquad
 +\frac{1}{W_{\lambda^2 s}^{\lambda,L}(k_2) }
\left(\FT{V}(k_{3},k_{1},k_{4};+)+\FT{V}(k_{3},k_{4},k_{1};+)\right)
\\ & \qquad
 +\frac{1}{W_{\lambda^2 s}^{\lambda,L}(k_3) }
\left(\FT{V}(k_{2},k_{1},k_{4};+)+\FT{V}(k_{2},k_{4},k_{1};+)\right)
\\ & \qquad
 -\frac{1}{W_{\lambda^2 s}^{\lambda,L}(-k_4) }
\left(\FT{V}(k_{2},k_{3},k_{1};-)+\FT{V}(k_{3},k_{2},k_{1};-)\right)
\\ & \quad
=
\left(\FT{V}(k_1+k_2)+\FT{V}(k_1+k_3)\right)\sum_{\ell=1}^4
\frac{\sigma_\ell}{W_{\lambda^2 s}^{\lambda,L}(\sigma_\ell k_\ell) }
\,,
\end{align*}
where, in the last step, we have used the definitions and symmetry of $\FT{V}$
and the assumption $\sum_{\ell=1}^4 k_\ell =0$.  Analogous computation for
$\sigma=-$ results in
\begin{align*}
&
  \sum_{\ell=1}^4 \sum_{\text{allowed }\pi}
   \frac{\sigma_\ell \FT{V}(k_{\pi(1)},k_{\pi(2)},k_{\pi(3)};\sigma_\ell)}{W_{\lambda^2 s}^{\lambda,L}(\sigma_\ell k_\ell) }
=
\left(\FT{V}(k_2+k_3)+\FT{V}(k_1+k_3)\right)\sum_{\ell=1}^4
\frac{\sigma_\ell}{W_{\lambda^2 s}^{\lambda,L}(\sigma_\ell k_\ell) }
\,.
\end{align*}
We can now insert these results into the evolution equation of $W$,
yielding
\begin{align*}
&
W_{\lambda^2 t}^{\lambda,L}(k) = (\text{higher order terms}) +W^{\lambda,L}_0(k)
 - \lambda^2 \int_0^t\!\rmd t' \int_0^{t'}\!\rmd s\,
  \int_{(\duallattice)^3}\rmd k_1 \rmd k_2 \rmd k_3
\\ & \qquad \times
\biggl(
   \delta_L(k-k_1-k_2-k_3) \FT{V}(k_1+k_2)
  \left(\FT{V}(k_1+k_2)+\FT{V}(k_1+k_3)\right)
\\ & \qquad \qquad\times
  \rme^{-\ci \Omega_+ (t'-s)} \sum_{\ell=1}^4
  \sigma_\ell \prod_{j\ne \ell}\left. W_{\lambda^2 s}^{\lambda,L}(\sigma_j k_j)\right|_{k_4=-k,\,\sigma=(-,+,+,-)}
\\ & \qquad\quad -
   \delta_L(-k-k_1-k_2-k_3) \FT{V}(k_2+k_3)
   \left(\FT{V}(k_2+k_3)+\FT{V}(k_1+k_3)\right)
\\ & \qquad \qquad\times
  \rme^{-\ci \Omega_- (t'-s)} \sum_{\ell=1}^4
  \sigma_\ell \prod_{j\ne \ell} \left. W_{\lambda^2 s}^{\lambda,L}(\sigma_j k_j)\right|_{k_4=k,\,\sigma=(-,-,+,+)} \biggr)
 \,,
\end{align*}
where
\[
 \Omega_+ = -\omega_1+ \omega_2+\omega_3-\omega_4\,\qquad
 \Omega_- = -\omega_1- \omega_2+\omega_3+\omega_4\,.
\]
In the second term, coming from $\sigma=-$, we first swap the $k_1$ and $k_3$
integrals, which amounts to a swap $1\leftrightarrow 3$ in that term.  After
that, we change the sign of each of the three new $k$-variables.  In the new
variables, the potential factors and the remaining $\delta_L$-term match with the
those in the first term, $\sigma\to -(-,+,+,-)$, and $\Omega_- \to -\Omega_+$
since $\omega(-k)=\omega(k)$.
Therefore, after dropping the higher order terms,
\begin{align*}
&
W_{\lambda^2 t}^{\lambda,L}(k) \approx W^{\lambda,L}_0(k)
 - \lambda^2 \int_0^{t}\!\rmd s\int_s^t\!\rmd t' \,
  \int_{(\duallattice)^3}\rmd k_1 \rmd k_2 \rmd k_3
\\ & \qquad \times
   \delta_L(k-k_1-k_2-k_3) \FT{V}(k_1+k_2)
  \left(\FT{V}(k_1+k_2)+\FT{V}(k_1+k_3)\right)
\\ & \qquad \qquad\times
  \left(\rme^{-\ci \Omega_+ (t'-s)}+\rme^{\ci \Omega_+ (t'-s)} \right) \sum_{\ell=1}^4
  \sigma_\ell \prod_{j\ne \ell}\left. W_{\lambda^2 s}^{\lambda,L}(\sigma_j k_j)\right|_{k_4=-k,\,\sigma=(-,+,+,-)}
 \,.
\end{align*}
The remaining $t'$-integral yields a factor
\[
 \int_s^t\!\rmd t' \,
 \left(\rme^{-\ci \Omega_+ (t'-s)}+\rme^{\ci \Omega_+ (t'-s)} \right) =
 2\pi \delta_{\lambda,\lambda^2(t-s)}(\omega_0+\omega_1-\omega_2-\omega_3)\,.
\]
In the remaining term, we can symmetrize over the labels $2$ and $3$.
Finally, rescaling $s$ and changing interaction variables $k_j\to \sigma_j k_j$,
we obtain
\begin{align*}
&
W_{T}^{\lambda,L}(k_0) = (\text{higher order terms}) + W^{\lambda,L}_0(k_0)
 + \pi \int_0^{T}\!\rmd r
  \int_{(\duallattice)^3}\rmd k_1 \rmd k_2 \rmd k_3
\\ & \qquad \times
   \delta_L(k_0+k_1-k_2-k_3)
  \left(\FT{V}(k_1-k_2)+\FT{V}(k_1-k_3)\right)^2\delta_{\lambda,T-r}(\omega_0+\omega_1-\omega_2-\omega_3)
\\ & \qquad \times
\sum_{\ell=0}^3
(-  \sigma_\ell) \prod_{j\ne \ell}\left. W_{r}^{\lambda,L}(k_j)\right|_{\,\sigma=(-,-,+,+)}
\\ & \quad
= W^{\lambda,L}_0(k_0) + \int_0^{T}\!\rmd r\,
\cnls{+}(W^{\lambda,L},\tau_1)(k_0) + (\text{higher
order terms})\,,
\end{align*}
with $\tau_1(T,t)=T-t$ and using the collision operator defined in
(\ref{eq:cnls}).  This justifies the formula (\ref{eq:WrealevoNLS}) for the
DNLS case which has $q=0$ and $\theta=+$.

\subsection{Proof the propagator estimates}
\label{sec:propagator}

We will next give the proof of Proposition \ref{prop:propagator_multi_d}. The 
proof will be a rather direct consequence of the following one-dimensional propagator
estimate.

\begin{lemma}
    \label{lem:propagator}
Consider the following extended one-dimensional lattice propagator function
 \[
  Q(x;R,u,L) \coloneqq   \int_{\oneduallattice}\! \rmd k \,\rme^{\ci 2 \pi k x+ \ci R 
  \cos(2 \pi (k+u))}\,,\qquad x\in \onelattice\,,\ R,u\in \R\,,\ L\ge 2\,,
 \]
and its infinite volume counterpart
\[
 I(y;R,u)\coloneqq  
 \int_{-\frac{1}{2}}^{\frac{1}{2}}\! \rmd k\, \rme^{\ci 2 \pi k y+ \ci R \cos(2 \pi (k+u))}\,, \quad y\in \Z\,,\ R,u\in \R \,.
\]
There are pure constants $C,\delta_0>0$ such that for all
$L\ge 2$, $R,u\in \R$, and $x\in \onelattice$, 
\[
 \left|Q(x;R,u,L)-I(x;R,u)\right|\le C \rme^{-\delta_0\left(L-4 |R|\right)_+} \,.
\]
and
\[
 |I(y;R,u)|\le \rme^{-2\delta_0(|y|-2 |R|)_+}\,, \quad y\in \Z\,.
\]
We can also find a constant $C$ such that 
\[
 |I(y;R,u)|\le C \mean{R}^{-\frac{1}{2}} \mean{y}^{\frac{1}{6}}\,, \quad y\in \Z\,.
\]
In particular,
\[
 |I(y;R,u)|\le C 
\regabs{y}^{-\frac{1}{3}}
 \min
 \left( 
 1, \sqrt{\frac{\regabs{y}}{\mean{R}}} 
 \right) 
 \,, \quad y\in \Z\,.
\]

In addition, for any power $p\ge 0$ we can find a constant $C(p)$ such that 
for all $L\ge 2$, $R,u\in \R$
\[
 \sum_{y\in \Z} \mean{y}^p |I(y;R,u)| \le C(p) \mean{R}^{p+\frac{1}{2}}\,,
\]
and
\[
 \sum_{x\in \onelattice} \mean{x}^p |Q(x;R,u,L)| \le C(p)  \mean{\min(|R|,L)}^{p+\frac{1}{2}}\,.
\]
\end{lemma}
\begin{proof}
For the proof, let us fix an allowed choice of $R,u,L$ and consider $Q(x)\coloneqq Q(x;R,u,L)$ for $x\in \onelattice$.
We denote
 \[
  g(k) \coloneqq  \ci R \cos(2 \pi (k+u))\,.
 \]
The function $g$ is entire and $1$-periodic.  In addition,
\[
 g'(k) = -\ci 2 \pi  R\sin(2 \pi (k+u))\,,\quad
 g''(k) = -\ci (2\pi)^2  R\cos(2 \pi (k+u))\,.
\]
Thus if $r\in \R$, we have
\[
 g(k+\ci r)-g(k)= \int_{0}^1\!\rmd s\, \ci r g'(k+\ci s r)\,, 
\]
and thus 
\[
 \left|g(k+\ci r)-g(k)\right|\le 
 |r| 2\pi |R| \rme^{2\pi |r|}\,.
\] 
Therefore, if we fix $r_0\coloneqq  \frac{1}{2\pi}\ln 2>0$ and denote $p_0=2\pi r_0>0$, then $r_0$ is a pure constant and for all $|r|\le r_0$ we obtain
\[
  \left|g(k+\ci r)-g(k)\right|\le 2 p_0 |R|\,.
\]

Let us next consider the integral defined above, i.e., set
\[
 I(y)\coloneqq  \int_{-\frac{1}{2}}^{\frac{1}{2}}\! \rmd k\, \rme^{\ci 2 \pi k y+g(k)}=
 \int_{-\frac{1}{2}}^{\frac{1}{2}}\! \rmd k\, \rme^{\ci 2 \pi k y+ \ci R \cos(2 \pi (k+u))}\,, \quad y\in \Z\,.
\]
The sequence $(I(y))_{y\in \Z}$ corresponds to the Fourier series
of the entire, $1$-periodic function 
\[
 k\mapsto \rme^{\ci R \cos(2 \pi (k+u))}\,.
\]
Since this function is integrable, the Fourier series is in $\ell_2$.  In fact, it is even exponentially decaying at infinity, as the following standard trick relying on Cauchy's theorem shows:  using analyticity we may change the integration contour from the real axis to parallel to the axis at $k\pm \ci r_0$, where the sign is equal to the sign of $y$.  By periodicity, the ``boundary paths'' cancel each other out.  Therefore, we find that
\[
 |I(y)|\le \int_{-\frac{1}{2}}^{\frac{1}{2}}\! \rmd k\, \left|\rme^{-2\pi r_0|y|+\ci 2 \pi k y+g(k\pm \ci r_0)}\right| \le  \rme^{-p_0|y|}
 \int_{-\frac{1}{2}}^{\frac{1}{2}}\! \rmd k\, \rme^{2p_0 |R|} \left|\rme^{g(k)}\right|
 = \rme^{-p_0(|y|-2 |R|)}\,.
\]

By standard theory, we know that then the inverse Fourier transform works pointwise everywhere, i.e., that
\[
 \rme^{\ci R \cos(2 \pi (k+u))} = \sum_{y\in \Z^d} \rme^{-\ci 2\pi k y} I(y)\,,\qquad k\in \R\,.
\]
In particular, using this for $k\in \oneduallattice$, we obtain that if $x\in \onelattice$, then 
\[
  Q(x) = \sum_{y\in \Z} I(y) \int_{\oneduallattice}\! \rmd k \,\rme^{\ci 2\pi k(x-y)} =  \sum_{m\in \Z} I(x+L m)\, .
\]
The term $m=0$ is equal to $I(x)$, and for $|m|\ge 1$, the earlier exponential decay bound yields an estimate
\[
 |I(x+L m)|\le \rme^{-p_0(|x+Lm|-2 |R|)}
 \le \rme^{-p_0(L|m|-|x|-2 |R|)}
 \le \rme^{2 p_0 |R|-p_0 L \left(|m|-\frac{1}{2}\right)}
 \,,
\]
where in the last inequality we used the bound $|x|\le \frac{L}{2}$ which holds for any $x\in \onelattice$.  Therefore,
\[
 \sum_{|m|\ge 1} |I(x+L m)| \le 2 \rme^{2 p_0 |R|-p_0\frac{1}{2}L}
 \!\sum_{m=1}^\infty 
 \rme^{-p_0 L \left(m-1\right)}
= \rme^{\frac{p_0}{2}\left(4 |R|-L\right)}\frac{2}{1\!-\!\rme^{-p_0 L}} 
\le \rme^{\frac{p_0}{2}\left(4 |R|-L\right)}\frac{2}{1\!-\!\rme^{-2 p_0 }} \,,
\]
where in the final inequality we used the assumption $L\ge 2$.
Denoting $C\coloneqq \frac{2}{1-\rme^{-2 p_0 }}>2$ and $\delta_0\coloneqq p_0/2$, we have proven for the case $4 |R|\le L$.  If $4|R|>L$, we use the obvious bounds $|I(x)|,|Q(x)|\le 2$, which imply that 
\[
 |Q(x)-I(x)|\le 2 \le C = C \rme^{-\delta_0\left(L-4 |R|\right)_+} \,.
\]
Thus the stated estimate holds also in this case for the above choice of constants $C$ and $\delta_0$.

Let us next consider the summability estimates for a fixed $p\ge 0$.  By Parseval's theorem, the $\ell_2$-norm of $I(y)$ is equal to one.  Therefore, by Schwartz inequality,
\[
 \sum_{|y|\le 2 |R|} \mean{y}^p |I(y)|\le \left(\sum_{|y|\le 2 |R|} \mean{y}^{2 p}\right)^{\frac{1}{2}}
 \le C(p) \mean{R}^{p+\frac{1}{2}}\,.
\]
On the other hand, by the earlier estimates
\[
 \sum_{|y|> 2 |R|} \mean{y}^p |I(y)|\le
 \sum_{|y|> 2 |R|} \mean{y}^p \rme^{-p_0(|y|-2 |R|)}
 \le C(p,p_0) \mean{R}^p\,.
\]
Summing up these estimates proves the bound for $I(y)$.

For the corresponding $Q(x)$ estimate, let us first point out that by discrete Parseval's theorem (or explicit computation) we have $\sum_{x\in \onelattice} |Q(x)|^2=1$.  Therefore, we have an a priori estimate
\[
\sum_{x\in \onelattice} \mean{x}^p |Q(x;R,u,L)| \le 
\left(\sum_{x\in \onelattice} \mean{x}^{2 p}\right)^{\frac{1}{2}} \le C L^{p+\frac{1}{2}}\,.
\]
This estimate will be used when $|R|\ge \frac{1}{8}L$, yielding a bound of the form $C(p) \min(L,|R|)^{p+\frac{1}{2}}$.
If $|R|< \frac{1}{8}L$, we resort to the approximation by $I(x)$:
\begin{align*}
 &  \sum_{x\in \onelattice} \mean{x}^p |Q(x)| \le
 \sum_{x\in \onelattice} \mean{x}^p |I(x)| 
 + \sum_{x\in \onelattice} \mean{x}^p 
 C \rme^{-\delta_0\left(L-4 |R|\right)} \\
 & \quad \le C \mean{R}^{p+\frac{1}{2}}
 + C L^{p+1}\rme^{-\delta_0 \frac{1}{2} L}
 \le C(p)  \mean{R}^{p+\frac{1}{2}}\,.
\end{align*}

For the final bound on $I(y)$, consider first the case $|y|\le \frac{|R|}{2}$.  Standard saddle point estimates then show that there is $C>0$ such that 
\[
 |I(y)|\le C \mean{R}^{-\frac{1}{2}}\,,
\]
since there are then exactly two saddle points and at both of them the second derivative has a modulus proportional to $\sqrt{R^2-y^2}\ge |R|\sqrt{3/4}$.
If $|y|\ge  \frac{|R|}{2}$, we rely on the global bound proven in
\cite[Sec.\ 2]{ho_fermi_1997}, which
implies that always $|I(y)|=|J_y(R)|\le |R|^{-\frac{1}{3}}$, where $J_y(R)$
denotes the Bessel function of the 
first kind.  Thus, for these values of $y$ we have
\[
 |I(y)|\le |R|^{-\frac{1}{3}} \le |R|^{-\frac{1}{2}} (2 |y|)^{\frac{1}{6}}\le 2 \mean{y}^{\frac{1}{6}}|R|^{-\frac{1}{2}}\,.
\]
This proves the existence of the constant for $|R|\ge 1$, and it can then be adjusted to cover also the case $|R|<1$ since $|I(y)|\le 1$.

Using the already proved bounds, we get 
\begin{align*}
    |I(y)| \le &
    C\ind(|y|\le4|R|)
    \regabs{y}^{\frac{1}{6}}\regabs{R}^{-\frac{1}{2}}
    +
    \ind(|y|>4|R|)
     \rme^{-2\delta_0(|y|-2 |R|)_+}
     \\
     \le &
    C\ind(|y|\le4|R|)
    \regabs{y}^{\frac{1}{6}-\frac{1}{2}}
    +
    \ind(|y|>4|R|)
     \rme^{-\delta_0|y|}
     \\
     \le &
     C 
    \regabs{y}^{-\frac{1}{3}},
\end{align*}
where in the last step we used that the exponential decays faster than any polynomial.
Taking minimum between this $\regabs{y}^{-\frac{1}{3}} $ bound and the $ \mean{R}^{-\frac{1}{2}} \mean{y}^{\frac{1}{6}}$ 
bounds gives the final estimate.

This concludes the proof of the Lemma.
\end{proof}

\begin{proofof}{Proposition \ref{prop:propagator_multi_d}}
 Note that
 \[
   Q(x;R,u,L) = \prod_{i=1}^d Q(x_i;R_i,u_i,L)\,,\qquad
   I(y;R,u) = \prod_{i=1}^d I(y_i;R_i,u_i)\,.
 \]
 Since $|R|_\infty\le \frac{L}{8}$ implies $|R_i|\le  \frac{L}{8}$ for all $i$,
 the first estimate follows by telescoping and using the bounds
 $|Q(x)|,|I(y)|\le 1$.  The bounds for $I(y)$ are immediate consequences of
 taking the product of the corresponding one-dimensional bounds given by Lemma \ref{lem:propagator}.
\end{proofof}

\bibliographystyle{abbrv} 
\bibliography{bibliography} 

\end{document}